\documentclass[cmp,numbook]{svjour}
\pdfoutput=1

\usepackage{cite}
\usepackage[arrow, matrix, curve]{xy} 
\usepackage[T1]{fontenc}    
\usepackage[final,bookmarksnumbered]{hyperref} 
\usepackage[USenglish]{babel} 
\usepackage{amsmath}         
\usepackage{amssymb}         
\usepackage{amsthm}
\usepackage{amsfonts}        
\usepackage[ascii]{inputenc} 
\usepackage{tikz}            
\usepackage{tikz-3dplot}
\usetikzlibrary{matrix,arrows}
\usepackage{aliascnt}        
\usepackage{microtype}       
\usepackage{graphics}        
\usepackage[all]{hypcap}
\usepackage{braket}
\usepackage{times}

\newcommand{\newjointcountertheorem}[3]{
	\newaliascnt{#1}{#2}
	\newtheorem{#1}[#1]{#3}
	\aliascntresetthe{#1}
}

\newtheorem{thm}{Theorem}[section]
\newjointcountertheorem{lem}{thm}{Lemma}
\newjointcountertheorem{cor}{thm}{Corollary}
\newjointcountertheorem{prp}{thm}{Proposition}
\newjointcountertheorem{cnj}{thm}{Conjecture}
\newjointcountertheorem{que}{thm}{Question}
\newjointcountertheorem{fct}{thm}{Fact}
\newjointcountertheorem{obs}{thm}{Observation}
\newjointcountertheorem{alg}{thm}{Algorithm}
\newjointcountertheorem{asm}{thm}{Assumption}
\theoremstyle{definition}
\newjointcountertheorem{dfn}{thm}{Definition}
\newjointcountertheorem{ntn}{thm}{Notation}
\theoremstyle{remark}
\newjointcountertheorem{rem}{thm}{Remark}
\newjointcountertheorem{nte}{thm}{Note}
\newjointcountertheorem{exl}{thm}{Example}

\DeclareMathOperator{\sign}{sign}
\DeclareMathOperator{\diag}{diag}

\DeclareMathOperator{\Sym}{Sym}
\newcommand{\Alt}{\Lambda}
\newcommand{\DuHe}{\mathrm{DH}}

\DeclareMathOperator{\vol}{vol}

\DeclareMathOperator{\Hom}{Hom}
\DeclareMathOperator{\SU}{SU}
\DeclareMathOperator{\U}{U}
\DeclareMathOperator{\SO}{SO}
\DeclareMathOperator{\conv}{conv}

\DeclareMathOperator{\supp}{supp}
\DeclareMathOperator{\linspan}{span}

\DeclareMathOperator{\Res}{Res}
\DeclareMathOperator{\spec}{eig}

\def\Snospace~{\S{}}

\newcommand{\qbinom}[2]{\genfrac{[}{]}{0pt}{}{#1}{#2}}
\newcommand{\su}{\mathfrak{su}}
\newcommand{\tr}[1]{\mathrm{tr}\left(#1\right)}

\renewcommand{\phi}{\varphi}
\renewcommand{\epsilon}{\varepsilon}
\newcommand{\restrict}[2]{\left.#1\vphantom{\big|}\right|_{#2}}
\renewcommand{\matrix}[1]{\left(\begin{smallmatrix}#1\end{smallmatrix}\right)}

\newcommand{\norm}[1]{\lVert#1\rVert}
\newcommand{\abs}[1]{\lvert#1\rvert}

\newcommand{\PP}{\mathbb P}  
\newcommand{\CC}{\mathbb C}
\newcommand{\RR}{\mathbb R}
\newcommand{\QQ}{\mathbb Q}
\newcommand{\ZZ}{\mathbb Z}
\newcommand{\NN}{\mathbb N}
\newcommand{\Prob}{\mathbf P}  
\newcommand{\Hbf}{\mathbf H}  
\newcommand{\Id}{\mathbf 1}

\allowdisplaybreaks[4] 

\title{Eigenvalue Distributions of Reduced Density Matrices}
\author{Matthias Christandl\inst{1} \and Brent Doran\inst{2} \and Stavros Kousidis\inst{3} \and Michael Walter\inst{4}}
\institute{
  Institute for Theoretical Physics, ETH Zurich, Wolfgang--Pauli--Strasse 27, 8093 Zurich, Switzerland, \\
  Department of Mathematical Sciences, University of Copenhagen, Universitetsparken 5, 2100 Copenhagen, Denmark,
  \email{christandl@math.ku.dk}
  \and
  Department of Mathematics, ETH Zurich, R\"amistrasse 101, 8092 Zurich, Switzerland, \\
  \email{brent.doran@math.ethz.ch}
  \and
  Institute for Theoretical Physics, ETH Zurich, Wolfgang--Pauli--Strasse 27, 8093 Zurich, Switzerland, \\
  Institute of Physics, University of Freiburg,
  Rheinstrasse 10, 79104 Freiburg, Germany, \\ Rosenthalstr. 17, 53859 Niederkassel, Germany
  \email{st.kousidis@googlemail.com}
  \and
  Institute for Theoretical Physics, ETH Zurich, Wolfgang--Pauli--Strasse 27, 8093 Zurich, Switzerland, \\
  \email{mwalter@phys.ethz.ch}
}

\begin{document}

\maketitle

\begin{abstract}
Given a random quantum state of multiple distinguishable or indistinguishable particles, we provide an effective method, rooted in symplectic geometry, to compute the joint probability distribution of the eigenvalues of its one-body reduced density matrices.
As a corollary, by taking the distribution's support, which is a convex moment polytope, we recover a complete solution to the one-body quantum marginal problem.
We obtain the probability distribution by reducing to the corresponding distribution of diagonal entries (i.e., to the quantitative version of a classical marginal problem), which is then determined algorithmically. This reduction applies more generally to symplectic geometry, relating invariant measures for the coadjoint action of a compact Lie group to their projections onto a Cartan subalgebra, and can also be quantized to provide an efficient algorithm for computing bounded height Kronecker and plethysm coefficients.
\end{abstract}

\section{Introduction}

The pure state of a quantum system is described by a vector in a complex Hilbert space, or more precisely by a point in the corresponding projective space. Herein we consider the finite dimensional case, for instance a spin system. Since the Hilbert space for multiple particles is given by the tensor product of the Hilbert spaces of the individual particles, its dimension grows exponentially with the number of particles. This exponential behavior is therefore the key obstruction to classical modeling of quantum systems.  The observation is as old as quantum theory itself, and physicists ever since have tried to find ways around it.  In that spirit, our paper presents an effective method to extract those physical features that ``only depend on the one-body eigenvalues'' associated to randomly-chosen quantum states of any fixed number of particles. Our methods are rooted in geometry, and so further a nascent dialog between algebraic and symplectic geometry on the one hand, and the theory of quantum computation and quantum information on the other, with new results for each subject. In particular, we focus on moment polytopes and Duistermaat--Heckman measures, which in recent years have arisen in many mathematical contexts and yet are difficult to compute. We therefore begin with some context, and make an effort in the body of the paper to build a correspondence in terminologies.

Typically, to address the aforementioned exponential complexity, physicists make use of a very simple yet powerful observation: Important properties such as energy and entropy often do not depend on the whole wavefunction but rather on only a small part, namely the \emph{reduced density matrix}, or quantum marginal, of a few particles. For instance, the binding energy of a molecule is given as a minimization over two-electron reduced density matrices arising from $N$-electron wavefunctions. Mathematically, the reduced density matrix is given as the contraction of (or trace over) the indices of the projection operator onto the wavefunction over the remaining particles (see \eqref{definition reduced density matrix} for the precise definition). The problem of characterizing the set of possible reduced density matrices, known as the \emph{quantum marginal problem} in quantum information theory and as the $N$-representability problem in quantum chemistry \cite{ruskai69,colemanyukalov00}, has therefore been considered one of the most fundamental problems in quantum theory \cite{stillinger95}.  The general problem is computationally intractable, even on a quantum computer; more precisely it is QMA-complete and NP-hard \cite{liu06, liuchristandlverstraete07}.   However, the characterization of the one-body reduced density matrices of a pure global quantum state \cite{klyachko04, daftuarhayden04, klyachko06} admits a very elegant mathematical interpretation: it amounts to a description of the possible eigenvalues of the reduced density matrices; and this requires the computation of moment polytopes for coadjoint orbits of unitary groups, as first observed in \cite{christandlmitchison06, daftuarhayden04, klyachko04}.  An interesting consequence is that defining linear inequalities for the polytope can have physical interpretation: for instance, the Pauli principle is simply one linear inequality bounding the polytope of a fermionic system (see \cite{coleman63,borlanddennis72,ruskai07,klyachkoaltunbulak08,klyachko09}, \cite{higuchi03,higuchisudberyszulc03,bravyi04}, \cite{eiserttycrudolphetal08} and \cite{walterdorangrossetal12} for further examples).

Random states play a fundamental role in physics.  In classical \emph{statistical mechanics}, the canonical state (or, ``canonical ensemble'') is the marginal probability distribution (Gibbs measure, or Boltzmann distribution) of states of the system arising from a uniformly random configuration of system and bath, subject to an energy constraint and fixed particle numbers for both system and bath \cite{huang90}. In quantum statistical mechanics, the canonical state of the system is the reduced density matrix of the uniform state on a subspace (encoding the energy constraint) of system and bath. In fact it can be shown that the canonical state almost always approximates the reduced density matrix of a random pure state in this subspace, as for large systems a concentration of measure occurs~\cite{popescushortwinter06, lloyd06, goldsteinlebowitztumulkaetal06}. Earlier, and also out of thermodynamic considerations, Lloyd and Pagels in a seminal work computed the distribution of eigenvalues of the reduced density matrix of the system when system and bath are in a random pure state~\cite{lloydpagels88}, i.e., in the simplest case of two large ``particles'', one representing the system and one the bath.  In this paper we are concerned with computing exact eigenvalue distributions for any number of particles, rather than with asymptotic results or with simple model systems of two particles.  More precisely, we are concerned with the following question:
\begin{quote}
 \emph{What is the probability distribution of the eigenvalues of the one-body reduced density matrices of a pure many-particle quantum state drawn at random from the unitarily invariant distribution?}
\end{quote}
In this article, we answer this question completely by describing an explicit algorithm to compute the joint eigenvalue distribution of the reduced density matrices for an arbitrary number of particles of any statistics (distinguishable, Bose, Fermi). As a special case we easily recover the Lloyd--Pagels result for two distinguishable particles.  As a corollary, our work naturally leads to a solution of the one-body quantum marginal problem in terms of a finite union of polyhedral chambers, thereby providing a complementary perspective on the work of Klyachko and of Berenstein--Sjamaar, and more recently of Ressayre, who each instead provided procedures based on geometric invariant theory (in particular the Hilbert--Mumford criterion) to list characterizing linear inequalities \cite{berensteinsjamaar00,klyachko04,ressayre10} for the moment polytope.

We emphasize that the computed eigenvalue distributions can be directly used to infer distributions of R\'enyi and von Neumann entropies, both of which play a fundamental role in statistical physics and quantum information theory and are functions of the eigenvalues only. In particular, one can recover the average entropy of a subsystem \cite{lubkin78,page93}, which featured in an analysis of the black hole entropy paradox~\cite{page94,HaydenPreskill}. For applications to the study of quantum entanglement, see \autoref{physical applications}.

\medskip

From a mathematical perspective, the eigenvalue distributions that we compute are \emph{Duister\-maat--Heck\-man measures}, which are defined using the push-forward of the Liouville measure on a symplectic manifold along the moment map (see \autoref{notation} for the precise definition which we use in this article) \cite{heckman82,guilleminsternberg82b,guilleminsternberg84,guilleminlermansternberg88,guilleminlermansternberg96,guilleminprato90}.  The support of such a measure is a moment polytope, which in our physical context is the solution to the one-body quantum marginal problem.
We will work and establish our results in this more general symplectic setting, in similar spirit to the existing solution to the one-body quantum marginal problem \cite{berensteinsjamaar00,klyachko04}.  We remark that, in particular, the one-body quantum marginal problem subsumes the well-known problem of determining the possible eigenvalues of the sum $A + B$ of two Hermitian matrices $A$, $B$ with fixed eigenvalues, which goes back at least to Weyl \cite{Weyl12,helmkerosenthal95,klyachko98,knutsontao99,fulton00,knutsontao01,knutsontaowoodward03}. The corresponding eigenvalue distribution for randomly-chosen matrices $A$ and $B$ has also been studied in the literature, and our methods allow us to recover the main results of \cite{dooleyrepkawildberger93} (\autoref{main theorem horn}); see also \cite{frumkingoldberger06} for a more concrete approach.
However, one striking difference between existing approaches to the one-body quantum marginal problem and the corollary to our approach is that the subtleties associated with the non-Abelian nature of their solutions can be completely bypassed.  For instance, their reliance upon cohomology of Schubert cycles, and the interplay of different Weyl groups and sub-tori that feature because of repeated use of the Hilbert--Mumford criterion, may be seen as incidental.  The problem, even in the full generality of the symplectic setting, is at heart an Abelian one whose essential combinatorics is encoded in the maximal torus action together with taking finitely many explicit derivatives.  Remarkably, this is more than philosophy, as it has real import for computation.

\medskip

The first mathematical contribution of this paper is an effective technique for providing explicit Duister\-maat--Heck\-man measures under rather weaker assumptions than appear in the literature (\autoref{algorithms}).  The fact that this subsumes our main question above, and in particular the one-body quantum marginal problem, while the existing literature does not, is the crucial second point of this work (\autoref{purification}).  A third feature is our statement of the Abelianization procedure via the derivative principle for invariant measures (\autoref{main theorem measures}); though we have not seen this principle so formulated in the literature, there certainly are predecessors and the result should be equivalent to one of Harish-Chandra \cite{harishchandra57}.  That a ``quantized'' version of our algorithm can be used to compute (efficiently, unlike existing algorithms, see \autoref{multiplicities for projective space}) multiplicities for the branching problem of representation theory may be seen as the fourth mathematical consequence of our approach.

The basic strategy we follow to address our main question above is a sequence of reductions.  Firstly, our general quantum problem is replaced by an equivalent but more tractable one by ``purification'' of the quantum state (\autoref{math phys dictionary}).  This reduced problem is seen to satisfy a weak non-degeneracy assumption, so that the image of the moment map does not lie entirely on a wall in the relevant Weyl chamber.  Under this assumption we can reduce via a \emph{derivative principle} to the Duister\-maat--Heck\-man measure for the maximal torus action (\autoref{derivative principle}), which we evaluate by a ``single-summand'' algorithm along the lines of Boysal--Vergne \cite{boysalvergne09} (\autoref{projective space}).
This derivative principle holds for more general $K$-invariant measures on the dual of the Lie algebra, $\mathfrak k^*$:
Every invariant measure can be reconstructed from its projection onto the dual of the Lie algebra of the maximal torus by taking partial derivatives in the direction of negative roots (\autoref{main theorem measures}).  Again, we remark that such a reduction is not possible on the level of the supports of the Duister\-maat--Heck\-man measures, i.e., on the moment polytopes.  We also remark that in the case of the one-body quantum marginal problem, the Duister\-maat--Heck\-man measure for the maximal torus action is the joint distribution of the diagonal entries of the one-body reduced density matrices (as opposed to their eigenvalues). Amusingly, this distribution can be viewed as the solution to the quantitative version of a classical marginal problem (\autoref{classical}).
The single-summand algorithm alluded to above is an effective method for computing Duister\-maat--Heck\-man measures for torus actions on projective spaces (\autoref{projective space algorithm}). Its name stems from the fact that it amounts to evaluating a single summand of the kind that occurs in the well-known Heckman formula of Guillemin--Lerman--Sternberg \cite{guilleminlermansternberg88}, which expresses the measure as an alternating sum of iterated convolutions of Heaviside measures. We also describe an algorithm based on this latter formula (\autoref{abelian heckman algorithm}) that can in particular be applied to projections of coadjoint orbits (\autoref{heckman for coadjoint orbit reductions}), and hence to the setting of Berenstein--Sjamaar \cite{berensteinsjamaar00}.

Whenever a Hamiltonian group action can be quantized in a certain technical sense, Duister\-maat--Heck\-man measures have an interpretation as the asymptotic limit of associated rep\-re\-sen\-ta\-tion-theoretic multiplicities \cite{heckman82,guilleminsternberg82b,sjamaar95,meinrenken96,meinrenkensjamaar99,vergne98}.
In the second part of this article (\autoref{multiplicities section} and \autoref{kronecker section}), we thus study the representation theory connected to the one-body quantum marginal problem; here, the relevant multiplicities include the Kronecker coefficients, which play a major role in the representation theory of the unitary and symmetric groups \cite{fulton97}, as well as in Mulmuley and Sohoni's geometric complexity theory approach to the P vs.\ NP problem in computer science \cite{mulmuleysohoni01,mulmuleysohoni08,mulmuley07,burgisserlandsbergmaniveletal11}.
It has been observed that the existence of a pure tripartite quantum state with given marginal eigenvalue spectra is equivalent to the asymptotic non-vanishing of an associated sequence of Kronecker coefficients \cite{christandlmitchison06,klyachko04,christandlharrowmitchison07}, see also \cite{daftuarhayden04,knutson09,burgisserchristandlikenmeyer11,burgisserchristandlikenmeyer11b}. For a similar connection in the context of Horn's conjecture and the Littlewood--Richardson coefficients, see \cite{lidskii82,knutson00,christandl08}. Interestingly, the behavior of these multiplicities can also be
related to thermodynamics and statistical physics following the lines of Okounkov \cite{okounkov00}.


Our main results in this context are quantized versions of our earlier theorems: Kronecker
coefficients can be computed by applying finite difference operators
to weight multiplicities which are related
to the classical marginal problem: Instead of measuring the volume of a polytope, one
has to count the number of lattice points in the polytope. This can be computed efficiently using Barvinok's algorithm
\cite{barvinok94}, and so leads to an efficient algorithm for
computing Kronecker coefficients for Young diagrams with bounded
height. Again, we shall establish the results in greater generality
and recover the version for Kronecker coefficients as a special case.

\subsection{Notation and Conventions}
\label{notation}

Throughout this article, $K$ will denote a compact, connected Lie group with maximal torus $T \subseteq K$, rank $r = \dim T$, Weyl group $W$, respective Lie algebras $\mathfrak k$ and $\mathfrak t$, and integral lattice $\Lambda = \ker \restrict \exp {\mathfrak t}$ \cite{cartersegalmacdonald95,kirillov08}.
We write $\pi_{K,T} \colon \mathfrak k^* \rightarrow \mathfrak t^*$ for the projection dual to the inclusion $\mathfrak t \subseteq \mathfrak k$.
We will think of weights as elements of the dual lattice, $\Lambda^* = \Hom_{\ZZ}(\Lambda, \ZZ) \subseteq \mathfrak t^*$, and identify a character $\chi \colon T \rightarrow \U(1)$ with the weight $d\chi/{2 \pi i} \in \mathfrak t^*$.
We denote by $d\lambda$ the Lebesgue measure on $\mathfrak t^*$ that is normalized in such a way that any fundamental domain of the weight lattice $\Lambda^*$ has unit measure. Let us also choose a positive Weyl chamber $\mathfrak t^*_+ \subseteq \mathfrak t^*$; this determines a set of positive roots $\{ \alpha > 0 \} = \{ \alpha_1, \ldots, \alpha_R \} \subseteq \mathfrak t^*$.
The set of negative roots is by definition $\{ -\alpha : \alpha > 0 \}$.
We write $\mathfrak t^*_{>0}$ for the interior of the positive Weyl chamber, which contains the strictly dominant weights.
We will often identify $\mathfrak k$ and its dual $\mathfrak k^*$, as well as $\mathfrak t$ and $\mathfrak t^*$, via some fixed $K$-invariant inner product $\braket{-,-}$ on $\mathfrak k$.

For the special unitary group $\SU(d)$, which we always take to be the group of unitary $d\times{}d$-matrices with unit determinant, we use the maximal torus consisting of diagonal matrices, on which the Weyl group $S_d$ acts by permuting diagonal entries. The Lie algebra $\su(d)$ consists of anti-Hermitian matrices with trace zero, and our choice of invariant inner product is $\braket{X, Y} = -\tr{X Y}$. Using it to identify $\mathfrak t$ and $\mathfrak t^*$, a positive Weyl chamber $\mathfrak t^*_+$ is given by the set of diagonal matrices $\lambda = \diag(\lambda_1, \ldots, \lambda_d)$ with purely imaginary entries, summing to zero and arranged in such a way that $i \lambda_1 \geq \ldots \geq i \lambda_d$. This corresponds to choosing the positive roots $\alpha_{j,k}(\lambda) = i(\lambda_j - \lambda_k)$ with $j < k$. The points in the interior $\mathfrak t^*_{>0}$ are those $\lambda \in \mathfrak t^*_+$ with all distinct eigenvalues, i.e., $i \lambda_1 > \ldots > i \lambda_d$.


\medskip

Let $M$ be a compact, connected Hamiltonian $K$-manifold of dimension $2n$, with symplectic form $\omega_M$ and a choice of moment map $\Phi_K \colon M \rightarrow \mathfrak k^*$ \cite{cannasdasilva08,guilleminsternberg84}. The intersection $\Delta_K(M) = \Phi_K(M) \cap \mathfrak t^*_+$ of its image with the
positive Weyl chamber is a compact convex polytope, called the \emph{moment polytope} or \emph{Kirwan polytope} \cite{guilleminsternberg82,kirwan84b,guilleminsjamaar05}.
If $M$ is a coadjoint $K$-orbit $\mathcal O_\lambda$ through some $\lambda \in \mathfrak t^*_+$, it will always be equipped with the Kirillov--Kostant--Souriau symplectic form and the moment map induced by the inclusion $\mathcal O_\lambda \subseteq \mathfrak k^*$. Evidently, $\Delta_K(\mathcal O_\lambda) = \{\lambda\}$.
Throughout this article, we will always impose the following non-degeneracy condition:

\begin{asm}
  \label{main assumption}
  The moment polytope $\Delta_K(M)$ has non-empty intersection with the interior of the positive Weyl chamber, $\mathfrak t^*_{>0}$.
\end{asm}

In view of \cite[Lemma 3.9]{lermanmeinrenkentolmanetal98} and well-known facts about compact Lie group actions, this assumption in fact implies the following: The set $\Phi_K^{-1}(K \cdot \mathfrak t^*_{>0})$ is an open, dense subset of $M$ whose complement has Liouville measure zero.
%
We show in \autoref{purification} that \autoref{main assumption} does \emph{not} restrict the applicability of our techniques to the problem of computing eigenvalue distributions of reduced density matrices.

\medskip

The \emph{Duister\-maat--Heck\-man measure} $\DuHe^K_M$ is then defined as follows \cite{duistermaatheckman82}:
Push forward the Liouville measure $\mu_M = \omega_M^n / ((2 \pi)^n n!)$ on $M$ along the moment map $\Phi_K$, compose with the push-forward along the quotient map $\tau_K \colon \mathfrak k^* \rightarrow \mathfrak t^*_+$ which sends all points in a coadjoint orbit $\mathcal O_\lambda$ to $\lambda$ in the positive Weyl chamber; then divide the resulting measure by the polynomial $p_K(\lambda) = \prod_{\alpha > 0} {\langle \lambda, \alpha \rangle} / {\langle \rho, \alpha \rangle}$, where $\rho$ is half the sum of positive roots. That is,
\begin{equation}
  \label{definition duistermaat-heckman measure}
  \DuHe^K_M = \frac 1 {p_K} (\tau_K)_* (\Phi_K)_* (\mu_M).
\end{equation}
\autoref{main assumption} ensures that $\DuHe^K_M$ is a locally finite measure on the interior of the positive Weyl chamber. Its support is equal to the moment polytope.
Note that $p_K(\lambda)$ is equal to the Liouville volume of a maximal-dimensional coadjoint orbit $\mathcal O_\lambda$ \cite[Proposition 7.26]{berlinegetzlervergne03}. Therefore, the above is a natural definition to use in our context: It is normalized so that the Duister\-maat--Heck\-man measure associated with the action of $K$ on a generic coadjoint orbit $\mathcal O_\lambda$ is a probability distribution concentrated at the point $\lambda$.

If $H$ is another compact, connected Lie group, with Lie algebra $\mathfrak h$, acting on $M$ via a group homomorphism $\phi \colon H \rightarrow K$, then this action is also Hamiltonian, with a moment map given by the composition
\begin{equation}
  \label{restriction to subgroups}
  \Phi_H = (d\phi)^* \circ \Phi_K \colon M \rightarrow \mathfrak k^* \rightarrow \mathfrak h^*.
\end{equation}
This in turn determines a Duister\-maat--Heck\-man measure $\DuHe^H_M$.
In particular, we can associate moment maps and Duister\-maat--Heck\-man measures with all closed subgroups of $K$. In the case of the maximal torus $T \subseteq K$, we shall call $\Phi_T$ the \emph{Abelian moment map} and $\DuHe^T_M$ the \emph{Abelian Duister\-maat--Heck\-man measure}, in order to distinguish them from the \emph{non-Abelian moment map} $\Phi_K$ and the \emph{non-Abelian Duister\-maat--Heck\-man measure} $\DuHe^K_M$, respectively. Explicitly,
\begin{equation}
  \label{definition duistermaat-heckman measure abelian}
  \DuHe^T_M = (\Phi_T)_* (\mu_M) = (\pi_{K,T})_* (\Phi_K)_* (\mu_M).
\end{equation}

Throughout this paper we shall assume for simplicity that the Abelian moment polytope $\Delta_T(M)$ is of maximal dimension.
This can always be arranged for by replacing $T$ by the quotient $T / \bigcap_{m \in M} T_m$, where $T_m$ denotes the $T$-stabilizer of a point $m \in M$.
If $T$ is the maximal torus of a semisimple Lie group, it follows already from \autoref{main assumption} that $\Delta_T(M)$ is of maximal dimension.
As a consequence,
$T$ acts locally freely on a dense, open subset whose complement has Liouville measure zero \cite[Lemma 3.1]{duistermaatheckman82}. In particular, generic points in $M$ are regular for the Abelian moment map. Therefore the Abelian Duister\-maat--Heck\-man measure is absolutely continuous with respect to Lebesgue measure on $\mathfrak t^*$.
By the Duister\-maat--Heck\-man Theorem, $\DuHe^T_M$ in fact has a polynomial density function of degree at most $n-r$
on each connected component of the set of regular values \cite[Corollary 3.3]{duistermaatheckman82}. We shall call these components the \emph{regular chambers}, and one can show that, except for the unbounded one, every such chamber is an open convex polytope. If the closures of two regular chambers have a common boundary of maximal dimension (i.e., of codimension one) then we shall say that the two chambers are \emph{adjacent} and call the common boundary a \emph{critical wall}.

All Hilbert spaces which we consider in this article are complex and finite-dimensional.
We write $P_\psi$ for the orthogonal projection onto a one-dimensional subspace $\CC \psi$, and $\norm{X}_2 = \sqrt{\tr{X^*X}} = \sum_j s_j^2$ for the Hilbert--Schmidt norm of an operator $X$ with singular values $(s_j)$.
We use $\braket{-,-}$ to denote inner products as well as the pairing between measures (or more general distributions) and test functions.
We write $\delta_p$ for the \emph{Dirac measure} at $p$, i.e., the probability measure concentrated at the point $p$, and $H_\omega$ for the \emph{Heaviside measure} which is defined by $\langle H_\omega, f \rangle = \int_0^\infty f(t \omega) dt$.
We sometimes use the letter $\Prob$ for probability distributions.

Throughout the paper when we speak of the quantum marginal problem we always refer to its one-body version as described in \autoref{math phys dictionary}.
Finally, we offer a word of caution for people acquainted with the theory of geometric quantization \cite{guilleminsternberg77,woodhouse92}:
Our quantum states do not arise via some quantization procedure from a classical symplectic phase space. In contrast, herein, as detailed in \autoref{math phys dictionary} below, the spaces of quantum states themselves are Hamiltonian manifolds. Probability distributions can be realized as quantum states of a special form, and the passage from quantum to classical is related to passing from a non-Abelian group to its maximal torus (see \autoref{classical} for precise statements). The ``semiclassical limit'' well-known in geometric quantization does not have an analogous physical meaning in our setting; its significance is solely to connect the symplectic geometry with representation theory (see \autoref{semiclassicallimit subsection}).

\section{Density Matrices and Purification}
\label{math phys dictionary}

The applicability of symplectic geometry to the quantum marginal problem relies on the close relation between the Lie algebra of $\SU(d)$ and the density matrices of quantum mechanics, and on the fact that restricting to certain subgroups has the physical meaning of passing to reduced density matrices, which describe the quantum state of subsystems. In this section we will describe this relationship in some detail (\autoref{density matrices}, \autoref{reduced density matrices}), and show how one can reduce the general problem of computing joint eigenvalue distributions of reduced density matrices to the case of globally pure quantum states, that is, to the Duister\-maat--Heck\-man measure for a projective space (\autoref{purification}). We briefly discuss how probability distributions and the classical marginal problem are embedded in our setup (\autoref{classical}) and describe some immediate physical applications (\autoref{physical applications}).

\subsection{Density Matrices}
\label{density matrices}

\begin{dfn}
  A \emph{density matrix} is a positive-semidefinite Hermitian operator $\rho$ of trace one acting on a finite-dimensional Hilbert space $\mathcal H$. We will often choose coordinates and think of $\rho$ as a matrix.
  If $\rho$ is the orthogonal projection onto a one-dimensional subspace then we say that $\rho$ is a \emph{pure state}; otherwise, it is a \emph{mixed state}.
  An \emph{observable} is an arbitrary Hermitian operator acting on $\mathcal H$.
\end{dfn}

Density matrices on $\mathcal H$ describe the state of a quantum system modeled by the Hilbert space $\mathcal H$: According to the postulates of quantum mechanics, the expectation value of an observable $O$ is given by the pairing $\tr{O \rho} \in \RR$. Of course, $\rho$ is characterized by these expectation values, even if we use anti-Hermitian observables instead and restrict to trace zero (since the trace of $\rho$ is fixed). That is, we have an injection
\begin{equation}
  \label{density matrices to functionals}
  \rho \mapsto \left( X \mapsto i \tr{X \rho} \right) \in \su(\mathcal H)^*
\end{equation}
which extends to an isomorphism between the affine space of trace-one Hermitian operators on $\mathcal H$ and $\su(\mathcal H)^*$.
This isomorphism is $\SU(\mathcal H)$-equivariant; its inverse sends a coadjoint orbit $\mathcal O_\lambda$ to the set of Hermitian operators with eigenvalues
\begin{equation}
  \label{spectra to positive weyl chamber}
  \hat\lambda_j = \frac 1 {\dim \mathcal H} + i \lambda_j
  \quad
  (j=1,\ldots,\dim \mathcal H),
\end{equation}
where the $\lambda_j$ are the eigenvalues of $\lambda$ (eigenvalues are repeated according to their multiplicity).
Let us choose coordinates $\mathcal H \cong \CC^d$ and identify $\SU(\mathcal H) \cong \SU(d)$ accordingly. Then the eigenvalues $\lambda_i$ are just the diagonal entries of the matrix $\lambda \in \mathfrak t^*_+$ labeling the coadjoint orbit (see \autoref{notation} for our conventions).
It follows that \eqref{spectra to positive weyl chamber} defines a bijection between the positive Weyl chamber $\mathfrak t^*_+$ and the set of eigenvalue spectra of trace-one Hermitian operators, which we think of as elements of the set $\{ \hat\lambda \in \mathbf R^d : \hat\lambda_1 \geq \ldots \geq \hat\lambda_d, \sum_j \hat\lambda_j = 1 \}$.
Note that the set of pure states is identified with the coadjoint orbit through the highest weight of the defining representation of $\SU(\mathcal H)$, that is, with projective space $\PP(\mathcal H)$: The density matrix corresponding to a point $[\psi] \in \PP(\mathcal H)$ is simply the orthogonal projection $P_\psi$ onto $\CC \psi$.

The Liouville measure on a coadjoint orbit $\mathcal O_\lambda$ is identified via \eqref{density matrices to functionals} with the unique $\SU(\mathcal H)$-invariant measure on the set of Hermitian matrices with spectrum $\hat\lambda$, normalized to total volume
\begin{equation}
  \label{volume coadjoint su orbit}
  p_{\SU(\mathcal H)}(\lambda) =
  \prod_{j < k} \frac{i (\lambda_j - \lambda_k)}{k - j} =
  \prod_{j < k} \frac{(\hat\lambda_j - \hat\lambda_k)}{k - j}.
\end{equation}
The non-Abelian Duister\-maat--Heck\-man measure of $\mathcal O_\lambda$ corresponds to the Dirac measure at $\hat\lambda$, while the Abelian Duister\-maat--Heck\-man measure corresponds to the distribution of the diagonal entries of a density matrix with spectrum $\hat\lambda$ chosen according to the invariant measure.

\begin{exl}
  \label{bloch ball}
  The Lie algebra $\su(2)$ is three-dimensional, generated by $-\frac i 2$ times the Pauli matrices $\sigma_x, \sigma_y, \sigma_z$. Functionals in its dual $\su(2)^*$ can be identified with points in $\RR^3$ by evaluating them at these generators. In this picture, the inverse of \eqref{density matrices to functionals} associates to a vector $\vec r = (x,y,z)$ the Hermitian matrix $\rho(\vec r) = \frac 1 2 ( \Id + \vec r \cdot \vec\sigma )$, where $\vec\sigma$ is the Pauli vector $(\sigma_x, \sigma_y, \sigma_z)$. The $z$-axis is identified with the Lie algebra of the maximal torus, $\RR \sigma_z$, its positive half-axis with our choice of positive Weyl chamber $\mathfrak t^*_+$, and $(0,0,2)$ with the corresponding positive root $\alpha > 0$.
  The coadjoint action of elements in $\SU(2)$ amounts to rotating the \emph{Bloch vector} $\vec r$ via the two-fold covering map $\SU(2) \rightarrow \SO(3)$. Therefore, coadjoint orbits are spheres, commonly called \emph{Bloch spheres} in quantum mechanics. They can be labeled by their radius $r$, that is, by their intersection with the positive half of the $z$-axis. Points on such a sphere correspond to Hermitian matrices with eigenvalue spectrum $(\frac {1+r} 2, \frac {1-r} 2)$. Note that $\rho(\vec r)$ is positive-semidefinite (i.e., a density matrix) if and only if $\vec r$ is contained in the unit ball of $\RR^3$.

  The non-Abelian moment map is just the inclusion map of a Bloch sphere into $\RR^3$. Hence its composition with the quotient map $\tau_{\SU(2)}$ sends all points in a Bloch sphere of radius $r > 0$ to $r$, while the Abelian moment map projects all points onto the $z$-axis. The Liouville measure is equal to the usual round measure, normalized to total volume $r$.
  Therefore, the non-Abelian Duister\-maat--Heck\-man measure of a Bloch sphere with radius $r$ is equal to the Dirac measure $\delta_r$, while the Abelian Duister\-maat--Heck\-man measure is obtained by pushing forward the Liouville measure onto the $z$-axis (see \autoref{bloch ball figure}). As already observed by Archimedes, any two zones of the same height on a sphere have the same area. Hence this latter measure is proportional to Lebesgue measure on the interval $[-r,r]$. An analogous statement holds for arbitrary projective spaces (\autoref{projective space abelian via standard simplex}).

  \begin{figure}
    \centering
    \includegraphics[width=6cm]{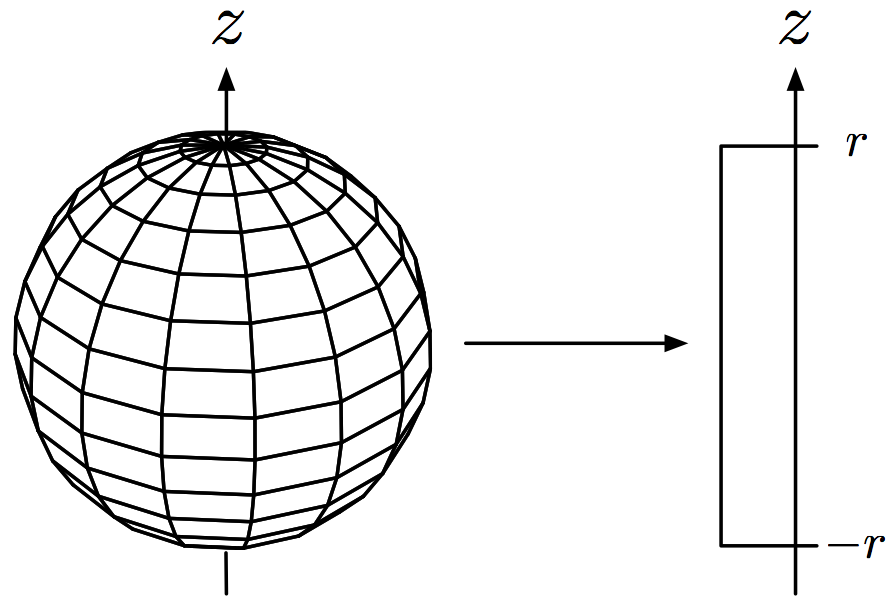}
    \caption{A Bloch sphere, its height function and the induced measure: a generic coadjoint $\SU(2)$-orbit, its Abelian moment map and the Abelian Duister\-maat--Heck\-man measure.}
    \label{bloch ball figure}
  \end{figure}
\end{exl}

\begin{rem}
  Observe that the components $\braket{X, \Phi_{\SU(\mathcal H)}} \colon \mathcal O_\lambda \rightarrow \RR$ of the moment map send a quantum state to the expectation value of the corresponding observable $-i X$. Without loss of generality, we may assume that $X$ generates a one-dimensional torus and that $X$ has one-dimensional eigenspaces. Then $\braket{X, \Phi_{\SU(\mathcal H)}}$ is just the moment map for the action of the torus generated by $X$ and its distribution can be computed immediately by using the Abelian Heckman formula (\autoref{abelian heckman}).
  This gives a short and conceptual proof of the formula derived recently in \cite{venutizanardi12}.
\end{rem}

\subsection{Reduced Density Matrices}
\label{reduced density matrices}

Composite quantum systems are modeled by the tensor product of the Hilbert spaces describing their constituents.
It is useful to think of these subsystems as individual particles, although they can be of more general nature; for instance, the subsystems can describe different degrees of freedom such as position and spin.
Depending on whether the particles are in principle distinguishable or indistinguishable, we distinguish two basic classes of composite systems, which are of fundamentally different nature.

\medskip

If the quantum system is composed of $N$ \emph{distinguishable particles}, its global quantum state is described by a density matrix on the tensor product $\mathcal H = \mathcal H_1 \otimes \ldots \otimes \mathcal H_N$, where the $\mathcal H_j$ are the Hilbert spaces describing the individual particles. Quantum mechanics also tells us that observables $O_j$ acting on a single subsystem $\mathcal H_j$ correspond to tensor product observables $\Id_{\mathcal H_1 \otimes \ldots \otimes \mathcal H_{j-1}} \otimes O_j \otimes \Id_{\mathcal H_{j+1} \otimes \ldots \otimes \mathcal H_N}$, which act by the identity on all other subsystems. By non-degeneracy of the inner product, there exists a unique density matrix $\rho_j$ on $\mathcal H_j$ such that
\begin{equation}
  \label{definition reduced density matrix}
  \tr{(\Id_{\mathcal H_1 \otimes \ldots \otimes \mathcal H_{j-1}} \otimes O_j \otimes \Id_{\mathcal H_{j+1} \otimes \ldots \otimes \mathcal H_N}) \rho} = \tr{O_j \rho_j}
\end{equation}
for all observables $O_j$. It describes the quantum state of the $j$-th subsystem.

\begin{dfn}
  The density matrix $\rho_j$ is called the \emph{(one-body) reduced density matrix} or \emph{quantum marginal} for the $j$-th particle of the quantum system.
\end{dfn}

Note that we can embed $\SU(\mathcal H_j)$ into $\SU(\mathcal H)$ by $U_j \mapsto \Id_{\mathcal H_1 \otimes \ldots \otimes \mathcal H_{j-1}} \otimes U_j \otimes \Id_{\mathcal H_{j+1} \otimes \ldots \otimes \mathcal H_N}$. This induces an embedding on the level of Lie algebras. The dual projection $\su(\mathcal H)^* \rightarrow \su(\mathcal H_j)^*$, given by restricting functionals to the subalgebra, is identified by \eqref{density matrices to functionals} with the map $\rho \mapsto \rho_j$. Similarly, the group homomorphism from the Cartesian product $\SU(\mathcal H_1) \otimes \ldots \otimes \SU(\mathcal H_N)$ to $\SU(\mathcal H)$ given by $(U_1,\ldots,U_N) \mapsto U_1 \otimes \ldots \otimes U_N$ induces the map $\rho \mapsto (\rho_1,\ldots,\rho_N)$ sending a density matrix to the tuple of all its one-body reduced density matrices.

The \emph{(one-body) quantum marginal problem for distinguishable particles} asks for the possible tuples of one-body reduced density matrices $(\rho_1, \ldots, \rho_N)$ of an arbitrary density matrix $\rho$ with fixed spectrum, or, equivalently, for the possible tuples of their eigenvalues.
By the above discussion, this is precisely equivalent to determining the moment polytope $\Delta_K(M)$ associated with the Hamiltonian action of the subgroup $K = \SU(\mathcal H_1) \otimes \ldots \otimes \SU(\mathcal H_N)$ on a coadjoint orbit $M = \mathcal O_{\tilde\lambda}$ for $\SU(\mathcal H)$, with moment map as defined in \eqref{restriction to subgroups}. The quantum marginal problem for globally pure states is the special case where $M = \PP(\mathcal H)$.
Moreover, the \emph{joint eigenvalue distribution of reduced density matrices} we set out to compute in this article corresponds to the non-Abelian Duister\-maat--Heck\-man measure $\DuHe^K_M$ as defined in \eqref{definition duistermaat-heckman measure}: Up to the identification $\hat\lambda \mapsto \lambda$ between spectra of trace-one Hermitian operators and the positive Weyl chamber as defined in \eqref{spectra to positive weyl chamber}, it is given by
\begin{equation}
  \label{eigenvalue distribution dist}
  (\tau_K)_* (\Phi_K)_*\left(\frac {\mu_M} {\vol M}\right) =
  \frac 1 {\vol M} p_K \DuHe^K_M.
\end{equation}
Note that we divide by the Liouville volume of $M$, which is just $p_{\SU(\mathcal H)}(\tilde\lambda)$, to obtain a probability measure.
Similarly, the joint distribution of the diagonal entries of the reduced density matrices corresponds to the Abelian Duister\-maat--Heck\-man measure $\DuHe^T_M$.

\medskip

If the quantum system is composed of \emph{indistinguishable particles}, each particle is of course modeled by the same Hilbert space $\mathcal H_1$. The global state of the system is described by a density matrix supported on an irreducible sub-representation $\mathcal H \subseteq \mathcal H_1^{\otimes N}$, namely $\mathcal H = \Sym^N(\mathcal H_1)$ for \emph{bosons} and $\Alt^N(\mathcal H_1)$ for \emph{fermions} (but we can in principle also consider other irreducible sub-representations which correspond to more exotic statistics). Note that since every such density matrix commutes with permutations, all the one-body reduced density matrices are equal.
We can let single-particle observables $O$ act more intrinsically by the symmetric expression $\frac 1 N (O \otimes \Id_{\mathcal H_1^{\otimes (N-1)}} + \ldots + \Id_{\mathcal H_1^{\otimes (N-1)}} \otimes O)$, without changing their expectation values. Up to a factor $N$, this corresponds to the embedding of Lie algebras induced by the diagonal map $U \mapsto U \otimes \ldots \otimes U$, which is of course precisely the action of $\SU(\mathcal H_1)$ on the representation $\mathcal H$. This embedding therefore induces the map $\rho \mapsto \sum_j \rho_j = N \rho_1$ in the same way as described above. It follows that the \emph{(one-body) quantum marginal problem for indistinguishable particles} amounts to determining $\frac 1 N \Delta_K(M)$ for the induced action of $K = \SU(\mathcal H_1)$ on a coadjoint orbit $M = \mathcal O_{\tilde\lambda}$ of $\SU(\mathcal H)$, and that, up to the identification $\hat\lambda \mapsto \lambda$,  the \emph{eigenvalue distribution of the reduced density matrix} is given by
\begin{equation}
  \label{eigenvalue distribution indist}
  \kappa_* (\tau_K)_* (\Phi_K)_*\left(\frac {\mu_M} {\vol M}\right) =
  \frac 1 {\vol M} p_K \kappa_*(\DuHe^K_M),
\end{equation}
where the linear map $\kappa(\lambda) = \frac \lambda N$ counteracts the factor $N$ in the moment map.

\begin{table}
  \centering
  \begin{tabular}{r|cc}
    Setting & Hilbert space $\mathcal H$ & Group $K$ \\
    \hline
    $N$ distinguishable particles & $\CC^{d_1} \otimes \ldots \otimes \CC^{d_N}$ & $\SU(d_1) \times \ldots \times \SU(d_N)$ \\
    $N$ bosons & $\Sym^N(\CC^d)$ & $\SU(d)$ \\
    $N$ fermions & $\Alt^N(\CC^d)$ & $\SU(d)$ \\
    \hline
  \end{tabular}
  \caption{The quantum marginal problem is modeled by the action of the group $K$ on a coadjoint $\SU(\mathcal H)$-orbit $M = \mathcal O_{\tilde\lambda}$.}
  \label{QMP summary table}
\end{table}

\begin{rem}
  In \autoref{QMP summary table} we have summarized the groups and spaces relevant for the quantum marginal problems of main physical interest, and we will focus on these in the remainder of this article. One can also combine both cases, e.g., to describe a quantum system composed of two different sorts of indistinguishable particles (as happens for the purified double of a bosonic or fermionic quantum marginal problem as defined in \autoref{purification} below), or a number of indistinguishable particles each of which have multiple internal degrees of freedom. In the latter case, arbitrary irreducible representations of the special unitary group can appear if one restricts to the reduced density matrices corresponding to only some of the degrees of freedom (see, e.g., \cite{klyachkoaltunbulak08}).
\end{rem}

\subsection{Purification}
\label{purification}

Let $\mathcal H$ be an arbitrary finite-dimensional Hilbert space.
It is well-known that every density matrix $\rho$ on $\mathcal H$ is the reduced density matrix of a pure state in $[\psi] \in \PP(\mathcal H \otimes \mathcal H)$, called a \emph{purification} of the quantum state $\rho$. Indeed, if $\rho = \sum_i p_i P_{v_i}$ is the spectral decomposition of $\rho$ then we can simply choose $\psi = \sum \sqrt{p_i} v_i \otimes v_i$. In this sense, the global state of a quantum system can always be described by a pure state; reduced density matrices occur only in the description of the states of its subsystems. This motivates the following definition:

\begin{dfn}
  For any unitary $K$-representation $\mathcal H$, we define the \emph{purified double} to be the Hamiltonian $K \times \tilde K$-manifold $\PP(\mathcal H \otimes \mathcal H)$, where $\tilde K = \SU(\mathcal H)$, equipped with the moment map constructed in the usual way by embedding into $\mathfrak u(\mathcal H \otimes \mathcal H)^*$ and ``restricting'' the functionals to elements in $\mathfrak k \oplus \mathfrak {\tilde k}$.
\end{dfn}

Observe that if $\mathcal H$ is one of the representations of \autoref{reduced density matrices} modeling a setup of the quantum marginal problem, then the purified double corresponds to the pure-state quantum marginal problem where one has adjoined a single distinguishable particle modeled by $\mathcal H$.

The purification $[\psi] \in \PP(\mathcal H \otimes \mathcal H)$ of a quantum state $\rho$ on $\mathcal H$ is unique up to a unitary acting on the second copy of $\mathcal H$. Evidently, such operations do not change the reduced density matrix $\rho = (P_\psi)_1$ and they leave the eigenvalue spectrum of $(P_\psi)_2$ invariant. In particular, the eigenvalue spectra of the reduced density matrices $(P_\psi)_1$ and $(P_\psi)_2$ are always equal. This implies that we can reduce the quantum marginal problem to the case of globally pure states, both for distinguishable and indistinguishable particles:

\begin{prp}
  \label{purification polytope}
  Let $\mathcal O_{\tilde\lambda}$ be a coadjoint orbit of $\tilde K = \SU(\mathcal H)$, with $\tilde\lambda \in \mathfrak {\tilde t}^*_+$ corresponding to the eigenvalue spectrum of a density operator.
  Then $\lambda \in \Delta_K(\mathcal O_{\tilde\lambda})$ if and only if $(\lambda,\tilde\lambda) \in \Delta_{K \times \tilde K}(\PP(\mathcal H \otimes \mathcal H))$.
\end{prp}

In other words,
\begin{equation*}
  \Delta_{K \times \tilde K}(\PP(\mathcal H \otimes \mathcal H)) = \bigcup_{\tilde\lambda \in \tilde\Delta} \Delta_K(\mathcal O_{\tilde\lambda}) \times \{\tilde\lambda\},
\end{equation*}
where $\tilde\Delta = \{ \tilde\lambda \in \mathfrak {\tilde t}^*_+ : \widehat{\tilde\lambda}_j \geq 0 \}$ is the convex subset of the positive Weyl chamber corresponding to the eigenvalue spectra of density operators.
We can similarly reduce the problem of determining the joint eigenvalue distribution to the case of globally pure states:
For this, let us define probability measures
\begin{equation}
\begin{aligned}
  \label{measures to be purified}
  \Prob &= (\tau_{K \times \tilde K})_* (\Phi_{K \times \tilde K})_* \left( \frac {\mu_{\PP(\mathcal H \otimes \mathcal H)}} {\vol \PP(\mathcal H \otimes \mathcal H)} \right),\\
  \Prob_{\tilde\lambda} &= (\tau_K)_* (\Phi_K)_* \left( \frac {\mu_{\mathcal O_{\tilde\lambda}}} {\vol O_{\tilde\lambda}} \right),
\end{aligned}
\end{equation}
where $\vol$ denotes the Liouville volume.

\begin{prp}
  \label{purification measure}
  The measures in \eqref{measures to be purified} are related by
  \begin{equation*}
    \langle \Prob, f \rangle
  = \frac 1 Z \int_{\tilde\Delta} d\tilde\lambda ~
    p^2_{\tilde K}(\tilde\lambda) ~
    \langle \Prob_{\tilde\lambda}, f(-,\tilde\lambda) \rangle,
  \end{equation*}
  for all test functions $f \in C_b(\mathfrak t^*_+ \oplus \mathfrak {\tilde t^*}_+)$,
  where $d\tilde\lambda$ is Lebesgue measure on $\mathfrak {\tilde t}^*_+$ and $Z$ a suitable normalization constant.
\end{prp}
\begin{proof}
  Each of the one-body reduced density matrices of a Liouville-distributed bipartite pure state in $\PP(\mathcal H \otimes \mathcal H)$ is distributed according to the Hilbert--Schmidt measure restricted to the set of density matrices. In particular, its eigenvalues are distributed according to the well-known formula of \cite{lloydpagels88,zyczkowskisommers01}, so that
  \begin{equation*}
      (\tau_{\tilde K})_* (\Phi_{\tilde K})_* \left( \frac {\mu_{\PP(\mathcal H \otimes \mathcal H)}} {\vol \PP(\mathcal H \otimes \mathcal H)} \right)
    = \frac 1 Z \, p^2_{\tilde K}(\tilde\lambda) \, \Id_{\tilde\Delta}(\tilde\lambda) \, d\tilde\lambda,
  \end{equation*}
  with $\Id_{\tilde\Delta}$ the indicator function of $\tilde\Delta$ and $Z$ a suitable normalization constant. We have just seen that both reduced density matrices necessarily have equal eigenvalue spectrum. This implies that
  \begin{equation*}
      \langle (\Phi_{\tilde K \times \tilde K})_* \left( \frac {\mu_{\PP(\mathcal H \otimes \mathcal H)}} {\vol \PP(\mathcal H \otimes \mathcal H)} \right), g \rangle
    = \frac 1 Z \int_{\tilde\Delta} d\tilde\lambda \int_{\mathcal O_{\tilde\lambda} \times \mathcal O_{\tilde\lambda}} g.
  \end{equation*}
  See \autoref{lloyd pagels diagonal} for an independent derivation using the techniques of this paper.
  It follows that
  \begin{align*}
    \langle \Prob, f \rangle
    &= \langle (\tau_{K \times \tilde K})_* (\Phi_{K \times \tilde K})_* \left( \frac {\mu_{\PP(\mathcal H \otimes \mathcal H)}} {\vol \PP(\mathcal H \otimes \mathcal H)} \right), f \rangle \\
    &= \langle (\tau_{K} \Phi_{K} \times \tau_{\tilde K})_* (\Phi_{\tilde K \times \tilde K})_* \left( \frac {\mu_{\PP(\mathcal H \otimes \mathcal H)}} {\vol \PP(\mathcal H \otimes \mathcal H)} \right), f \rangle\\
    &= \frac 1 Z \int_{\tilde\Delta} d\tilde\lambda \int_{\mathcal O_{\tilde\lambda} \times \mathcal O_{\tilde\lambda}} (\tau_{K} \Phi_{K} \times \tau_{\tilde K})^* \big( f \big)\\
    &= \frac 1 Z \int_{\tilde\Delta} d\tilde\lambda ~ p_{\tilde K}(\tilde\lambda) ~ \int_{\mathcal O_{\tilde\lambda}} (\tau_{K} \Phi_{K})^* \left( f(-,\tilde\lambda) \right)\\
    &= \frac 1 Z \int_{\tilde\Delta} d\tilde\lambda ~ p^2_{\tilde K}(\tilde\lambda) ~ \langle \Prob_{\tilde\lambda}, f(-,\tilde\lambda) \rangle.
  \end{align*}
\end{proof}

Note that $\Prob_{\tilde\lambda}$, and in particular the eigenvalue distributions \eqref{eigenvalue distribution dist} and \eqref{eigenvalue distribution indist}, vary continuously with the global spectrum $\tilde\lambda$. \autoref{purification measure} therefore implies that we can reconstruct them from the eigenvalue distribution for the purified double by taking limits.

\medskip

We will now show that \autoref{main assumption} is always satisfied when working with the purified double. In quantum-mechanical terms, we have to show that there exists a global pure state $[\psi] \in \PP(\mathcal H \otimes \mathcal H)$ such that the eigenvalue spectra of all the reduced density matrices are non-degenerate (with respect to the quantum marginal problem where we have added a single distinguishable particle with Hilbert space $\mathcal H$).

For distinguishable particles, where $\mathcal H \cong \CC^{d_1} \otimes \ldots \otimes \CC^{d_N}$, this follows from the following general criterion, since the purified double is constructed by adding an additional Hilbert space of dimension $d_{N+1} = \dim \mathcal H = d_1 \cdots d_N$:

\begin{lem}
  \label{main assumption qmp}
  Let $N \geq 1$ and $d_1 \leq \ldots \leq d_N \leq d_{N+1}$. Then there exists a global pure state in $\PP(\CC^{d_1} \otimes \ldots \otimes \CC^{d_N} \otimes \CC^{d_{N+1}})$ whose one-body reduced density matrices have non-degenerate eigenvalue spectra if and only if
  \begin{equation*}
    d_{N+1} \leq \left( \prod_{i=1}^N d_i \right) + 1.
  \end{equation*}
\end{lem}
\begin{proof}
  The condition is clearly necessary, since it follows from the singular value decomposition that at most $\prod_{i=1}^N d_i$ eigenvalues of $\rho_{N+1}$ can be non-zero.

  For sufficiency, let us construct a state with the desired property:
  For this, we consider the standard tensor product basis vectors $e_{i_1} \otimes \ldots \otimes e_{i_N}$ of $\CC^{d_1} \otimes \ldots \otimes \CC^{d_N}$, labelled by integers $i_j \in \{1,\ldots,d_j\}$, $j=1,\ldots,N$.
  We choose a subset of $d_{N+1}-1$ many such basis vectors $e_{i_1(k)} \otimes \ldots \otimes e_{i_N(k)}$ in such a way that, for each subsystem $j=1,\ldots,N$, at least $d_j - 1$ of the $d_j$ integers occur. This clearly is possible by our assumptions. Finally, we set
  \begin{equation*}
    \psi = \sum_{k=1}^{d_{N+1}-1} 2^{-k} e_{i_1(k)} \otimes \ldots \otimes e_{i_N(k)} \otimes e_k.
  \end{equation*}
  Then $[\psi]$ is a pure state such that all of its one-body reduced density matrices have non-degenerate eigenvalue spectrum.
\end{proof}


For bosons and fermions, we will use the following lemma to show that \autoref{main assumption} is satisfied:

\begin{lem}
  \label{boson fermion lemma}
  The convex hull of the weights of $\Sym^N(\CC^d)$ has maximal dimension.
  The same is true for $\Alt^N(\CC^d)$ if $N<d$.
\end{lem}
\begin{proof}
  In the following we write $\omega_j$ for the weight corresponding to the character $\diag(t_1, \ldots, t_d) \mapsto t_j$, with $j=1, \ldots, d$.

  (1) $\Sym^N(\CC^d)$: The vectors $e_j \otimes \ldots \otimes e_j$ are weight vectors of weight $N \omega_j$, with $j=1,\ldots,d$. Clearly, the convex hull of these weights already has maximal dimension.

  (2) $\Alt^N(\CC^d)$, where $1 \leq N < d$: It is well-known that the weights are given by $\sum_{j \in J} \omega_j$ for all $N$-element subsets $J \subseteq \{1,\ldots,d\}$. (The corresponding weight vectors are the well-known occupation number basis vectors for fermions.)
  Fix any such weight, say, the one corresponding to $J = \{1,\ldots,N\}$. The difference vectors between this weight and the weights obtained by replacing a single element of $J$ are proportional to the positive roots $\alpha_{j,k}(\lambda) = i(\lambda_j - \lambda_k)$ with $j \in J$ and $k \in \{1,\ldots,d\} \setminus J$. There are at least $d-1$ such roots, and they form a basis of $\mathfrak t^*$.
\end{proof}

\autoref{boson fermion lemma} implies that \autoref{main assumption} is also satisfied for the purified double of the bosonic and fermionic quantum marginal problems: Indeed, it clearly suffices to show that in each case there exists a density operator $\rho$ on $\mathcal H$ such that both $\rho$ and its one-body reduced density matrix $\rho_1$ have non-degenerate eigenvalue spectrum (then any purification of $\rho$ has the desired properties). By \autoref{boson fermion lemma}, there exists a convex combination of weights $\sum_k p_k \omega_k \in \mathfrak t^*_{>0}$. By perturbing slightly, we can arrange for the weights $(p_k)$ to be mutually disjoint. Choose corresponding (orthogonal) weight vectors $v_k \in \mathcal H$ and consider the density matrix $\rho = \sum_k p_k P_{v_k}$. Clearly, both $\rho$ and its one-body reduced density matrix $\rho_1$have non-degenerate eigenvalue spectrum.

To summarize, we have shown that the problem of computing the joint eigenvalue distribution of reduced density matrices is equivalent to the computation of Duister\-maat--Heck\-man measures associated with certain Hamiltonian group actions (cf.\ \autoref{QMP summary table}). Moreover, by passing to the purified double, we can always reduce to the case where $M = \PP(\mathcal H)$ is a projective space satisfying \autoref{main assumption}.

\subsection{Probability Distributions}
\label{classical}

Under the identification \eqref{density matrices to functionals}, elements of the dual of the Lie algebra of the maximal torus correspond to diagonal density matrices. These are precisely the diagonal matrices with non-negative entries summing to one, and can therefore be interpreted as \emph{probability distributions} of a random variable $Z$ with values in the orthonormal basis $(e_i)$ we have chosen.
This interpretation is in agreement with quantum mechanics: If we perform an actual measurement of a density matrix $\rho$ with respect to this orthonormal basis then the probability of getting outcome $e_i$ is given precisely by the diagonal element $\tr{P_{e_i} \rho} = \braket{e_i, \rho \, e_i} = \rho_{i,i}$.

Note that the moment map for the action of the maximal torus $\tilde T \subseteq \SU(\mathcal H)$ on the projective space $\PP(\mathcal H)$ corresponds to sending a pure state $[\psi]$ onto its diagonal.
As we vary $[\psi]$ over all pure states in $\PP(\mathcal H)$, the diagonal entries attain all possible probability distributions.
In other words, the Abelian moment polytope $\Delta_{\tilde T}(\PP(\mathcal H))$ is just the simplex $\tilde\Delta$ defined in \autoref{purification}.
The corresponding Duister\-maat--Heck\-man measure is equal to a suitably normalized Lebesgue measure on $\tilde\Delta$ (this is a special case of \autoref{projective space abelian via standard simplex} below).

\medskip

Now consider as in \autoref{reduced density matrices} the case of $N$ distinguishable particles. Choose orthonormal bases to identify $\mathcal H_k \cong \CC^{d_k}$, and therefore $\mathcal H \cong \CC^{d_1} \otimes \ldots \otimes \CC^{d_N}$ using the tensor product basis.
Note that we can interprete diagonal density matrices $\rho$ on $\mathcal H$ as the joint probability distribution of a tuple of random variables $(Z_1,\ldots,Z_N)$, where each $Z_k$ takes values in the standard basis of the corresponding $\CC^{d_k}$, by setting
\begin{equation*}
  \Prob(Z_1 = e_{i_1}, \ldots, Z_N = e_{i_N}) =
  \tr{P_{e_{i_1} \otimes \ldots \otimes e_{i_N}} \rho}.
\end{equation*}
The marginal distributions of the random variables $Z_k$ in the sense of probability theory are then given by
\begin{align*}
  &\Prob(Z_k = e_{i_k}) \\
  = &\sum_{i_1, \ldots, \check{i_k}, \ldots, i_N} \tr{P_{e_{i_1} \otimes \ldots \otimes e_{i_N}} \rho}\\
  = &\tr{ \left( \Id_{\CC^{d_1} \otimes \ldots \otimes \CC^{d_{k-1}}} \otimes P_{e_{i_k}} \otimes \Id_{\CC^{d_{k+1}} \otimes \ldots \otimes \CC^{d_N}} \right) \rho} \\
   = &\tr{P_{e_{i_k}} \rho_k},
\end{align*}
where for the second identity we have used that $\rho$ is a diagonal matrix. That is, the marginal distributions of the $Z_k$ are precisely described by the reduced density matrices $\rho_k$ (i.e., by the quantum marginals), which are also diagonal if $\rho$ is diagonal.

Accordingly, the moment polytope $\Delta_T(\PP(\mathcal H))$ for the action of the maximal torus $T \subseteq \SU(d_1) \times \ldots \SU(d_N)$ on the set of pure states describes the tuples of marginal probability distributions that arise from joint distributions of the $(Z_1,\ldots,Z_N)$. This \emph{(univariate) classical marginal problem} is of course trivial, since there are no constraints on the joint distribution. However, its quantitative version, which corresponds to computing the Abelian Duister\-maat--Heck\-man measure $\DuHe^T_{\PP(\mathcal H)}$, is interesting and not at all trivial to solve.
In fact, the problem of computing joint eigenvalue distributions of reduced density matrices, which we set out to solve in this article, can be reduced to the computation of $\DuHe^T_{\PP(\mathcal H)}$. This reduction, or rather the generalization which we describe in \autoref{derivative principle} below, is at the core of the algorithms presented in \autoref{algorithms}.

\subsection{Physical Applications}
\label{physical applications}

As indicated in the introduction, the eigenvalue distributions \eqref{eigenvalue distribution dist} and \eqref{eigenvalue distribution indist} have direct applications to quantum physics.
In quantum statistical mechanics, among others, one typically studies bipartite setups $\mathcal H = \mathcal H_S \otimes \mathcal H_E$ composed of a system $S$ and an environment (or bath) $E$. Randomly-chosen pure states give rise to a distribution of reduced density matrices $\rho_S$, whose properties vary with the size of the environment. Physical motivations have lead to the computation of the corresponding eigenvalue distribution \cite{lloydpagels88}, which we can easily re-derive using the techniques of this paper (\autoref{lloyd pagels}). Note that many basic physical quantities are functions of the eigenvalues, such as the \emph{von Neumann entropy}
\begin{equation*}
  H(S) = H(\rho_S) = -\tr{\rho_S \log \rho_S} = \sum_j -\hat\lambda_j \log \hat\lambda_j,
\end{equation*}
where $(\hat\lambda_j)$ are the eigenvalues of $\rho_S$, or more general R\'enyi entropies and purities (cf.\ \autoref{purity example}). The average von Neumann entropy of a subsystem \cite{lubkin78,page93} in particular has featured in the analysis of the black hole entropy paradox \cite{HaydenPreskill}.
We can also consider other coadjoint orbits such as Grassmannians: Here, the density matrix corresponding to a $d$-dimensional subspace $\mathcal H' \subseteq \mathcal H_A \otimes \mathcal H_E$ is the normalized projection operator $\rho = \Id_{\mathcal H'}/d$, and the reduced density matrix $\rho_A$ is interpreted as a \emph{canonical state} in the sense of statistical mechanics \cite{popescushortwinter06,lloyd06,goldsteinlebowitztumulkaetal06}. The probability distributions we compute can therefore be used to analyze the typical behavior of canonical states.

The tripartite case, in itself already interesting from the perspective of the quantum marginal problem, is also highly relevant to applications: It corresponds to the situation where $S$ itself is composed of two particles $A$ and $B$, so that $\mathcal H = \mathcal H_A \otimes \mathcal H_B \otimes \mathcal H_E$.
In the study of quantum entanglement, remarkable recent progress has been made by analyzing the entanglement properties of the two-body reduced density matrix $\rho_{AB}$ of a randomly-chosen pure state in large dimensions, where the concentration of measure phenomenon occurs \cite{HaydenRandomizing, haydenleungwinter06,aubrunszarekye11b,aubrunszarekye11,collinsnechitaye11}. In particular, a negative resolution of the additivity conjecture of quantum information theory~\cite{shor-additivity} has recently been obtained by related methods~\cite{hastings-additivity, aubrun-hastings}. The joint eigenvalue distribution of the reduced density matrices in particular determines \emph{quantum conditional entropies} and \emph{quantum mutual informations}, that is, the quantities
\begin{align*}
  H(A|B) &= H(AB) - H(B) = H(E) - H(B),\\
  I(A:B) &= H(A) + H(B) - H(AB) = H(A) + H(B) - H(E),
\end{align*}
since the eigenvalue spectra of $\rho_{AB}$ and $\rho_E$ are equal (cf.\ \autoref{purification}). They have immediate applications to entanglement theory; for example, the quantum mutual information provides an upper bound on the amount of entanglement that can be distilled from a quantum state \cite{christandlwinter04}.


In all these applications, most known results are for large Hilbert spaces, since the techniques employed rely on asymptotic features such as measure concentration. Our algorithms require no such assumption. In particular, they are well-suited for low-dimensional systems, which previously remained inaccessible.

\section{Derivative Principle for Invariant Measures}
\label{derivative principle}

In this section we will describe a fundamental property of $K$-invariant measures on $\mathfrak k^*$ that are concentrated on the union of the maximal-dimensional coadjoint orbits (that is, on $K \cdot \mathfrak t^*_{>0}$).
Every such invariant measure can be reconstructed from its projection onto $\mathfrak t^*$ by taking partial derivatives in the direction of negative roots (\autoref{main theorem measures}).
In particular, this implies that the non-Abelian Duister\-maat--Heck\-man measure $\DuHe^K_M$ can be recovered from the Abelian Duister\-maat--Heck\-man measure $\DuHe^T_M$ (\autoref{main theorem}).

For the invariant probability measure supported on a single coadjoint orbit $\mathcal O_\lambda$ of maximal dimension (i.e., $\lambda \in \mathfrak t^*_{>0}$), this follows from a well-known formula of Harish-Chandra, which states that the Fourier transform is given by
\begin{equation*}
  \langle \DuHe^T_{\mathcal O_\lambda}, e^{i \langle -, X \rangle} \rangle =
  \sum_{w \in W} (-1)^{l(w)} e^{i \langle w\lambda, X \rangle} \prod_{\alpha > 0} \frac 1 {i \langle \alpha, X \rangle}
\end{equation*}
for every $X \in \mathfrak t$ which is not orthogonal to any root \cite[Theorem 2]{harishchandra57}.
Here, $l(w)$ is the length of the Weyl group element $w$.
This implies that the Abelian Duistermaat--Heckman measure is given by an alternating sum of convolutions,
  \begin{equation}
    \label{harish chandra}
    \DuHe^T_{\mathcal O_\lambda} = \sum_{w \in W} (-1)^{l(w)} \delta_{w \lambda} \star H_{-\alpha_1} \star \ldots \star H_{-\alpha_R}.
  \end{equation}
where we recall that $H_\omega$ is the Heaviside measure defined in \autoref{notation} by $\langle H_\omega, f \rangle = \int_0^\infty f(t \omega) dt$.
By the fundamental theorem of calculus, we have $\partial_\omega H_\omega = \delta_0$ (in the sense of distributions), so that
\begin{equation}
  \label{harish chandra derivative}
  \left( \prod_{\alpha > 0} \partial_{-\alpha} \right) \DuHe^T_{\mathcal O_\lambda} = \sum (-1)^{l(w)} \delta_{w \lambda},
\end{equation}
as was already observed by Heckman \cite[(6.5)]{heckman82}.
By restricting to the interior of the positive Weyl chamber, we thus obtain the basic relation
\begin{equation}
  \label{ableitungsformel coadjoint orbit}
  \restrict{\left( \prod_{\alpha > 0} \partial_{-\alpha} \right) \DuHe^T_{\mathcal O_\lambda}}{\mathfrak t^*_{>0}} = \delta_\lambda.
\end{equation}

\begin{exl}
  Every Bloch sphere of radius $r > 0$ is a coadjoint orbit of maximal dimension (cf.\ \autoref{bloch ball}).
  We have seen that $\DuHe^T_{\mathcal O_r}$ is equal to $\frac 1 2 \Id_{[-r,r]}(z) dz$, where $dz$ is Lebesgue measure on the $z$-axis. In agreement with \eqref{ableitungsformel coadjoint orbit}, we observe that
  \begin{equation*}
    \restrict{\partial_\alpha \DuHe^T_{O_r}}{\mathfrak t^*_{>0}} =
    \restrict{\partial_z \Id_{[-r,r]}(z) dz}{\RR_{> 0}} =
    \delta_r.
  \end{equation*}
\end{exl}

\begin{thm}
  \label{main theorem measures}
  Let $\nu$ be a $K$-invariant Radon measure on $\mathfrak k^*$ satisfying $\nu(\mathfrak k^* \setminus K \cdot \mathfrak t^*_{>0}) = 0$.
  Then,
  \begin{equation*}
    \restrict{\left(\prod_{\alpha > 0} \partial_{-\alpha} \right) (\pi_{K,T})_*(\nu)}{\mathfrak t^*_{>0}} =
    \restrict{\frac 1 {p_K} (\tau_K)_*(\nu)}{\mathfrak t^*_{>0}},
  \end{equation*}
  where the partial derivatives and the restriction are in the sense of distributions.
\end{thm}
\begin{proof}
  Let $f \in C_c^\infty(\mathfrak t^*_{>0})$ be a test function, which we extend by zero to all of $\mathfrak t^*$, and set $g := (\pi_{K,T})^*\left(\left(\prod_{\alpha > 0} \partial_\alpha \right) f\right)$. By definition and assumption, respectively,
  \begin{equation*}
    \langle \restrict{ \left(\prod_{\alpha > 0} \partial_{-\alpha} \right) (\pi_{K,T})_*(\nu)}{\mathfrak t^*_{>0}}, f \rangle =
    \langle \nu, g \rangle =
    \langle {\nu}\big|_{K \cdot \mathfrak t^*_{>0}}, g \rangle.
  \end{equation*}
  Since $\nu$ is a $K$-invariant measure, we can use Fubini's theorem to replace $g$ by its $K$-average.
  On each maximal-dimensional coadjoint orbit $\mathcal O_\lambda \subseteq K \cdot \mathfrak t^*_{>0}$, this average is given by
  \begin{equation*}
    \frac 1 {\vol \mathcal O_\lambda} \langle \mu_{\mathcal O_\lambda}, g \rangle =
    \frac 1 {p_K(\lambda)} \langle \mu_{\mathcal O_\lambda}, (\pi_{K,T})^*\left(\left(\prod_{\alpha > 0} \partial_\alpha \right) f\right) \rangle =
    \frac 1 {p_K(\lambda)} \langle \left(\prod_{\alpha > 0} \partial_{-\alpha} \right) \DuHe^T_{\mathcal O_\lambda}, f \rangle,
  \end{equation*}
  which by \eqref{ableitungsformel coadjoint orbit} is precisely equal to ${f(\lambda)}/{p_K(\lambda)}$.
  In other words, the averaged function is on $K \cdot \mathfrak t^*_{>0}$ equal to the pullback $(\tau_K)^* \left( f/{p_K} \right)$.
  We conclude that
  \begin{equation*}
    \langle {\nu}\big|_{K \cdot \mathfrak t^*_{>0}}, g \rangle =
    \langle {\nu}\big|_{K \cdot \mathfrak t^*_{>0}}, (\tau_K)^* \left( \frac f {p_K} \right) \rangle =
    \langle \restrict{\frac 1 {p_K} (\tau_K)_* (\nu)}{\mathfrak t^*_{>0}}, f \rangle.
    \qedhere
  \end{equation*}
\end{proof}

\begin{cor}
  \label{main theorem}
  The Duister\-maat--Heck\-man measures as defined in \autoref{notation} are related by
  \begin{equation*}
    \restrict{\left(\prod_{\alpha > 0} \partial_{-\alpha} \right) \DuHe^T_M}{\mathfrak t^*_{>0}} =
    {\DuHe^K_M}\bigg|_{\mathfrak t^*_{>0}}.
  \end{equation*}
\end{cor}
\begin{proof}
  \autoref{main assumption} guarantees that we can apply \autoref{main theorem measures} to the push-forward of the Liouville measure along the non-Abelian moment map $\Phi_K$.
\end{proof}

This is the \emph{derivative principle} alluded to in the title of this section. As we shall see in the following, it is a powerful tool for lifting results about the Duister\-maat--Heck\-man measure for torus actions to general compact Lie group actions.

\begin{rem}
  According to \cite[\S 3.5]{woodward05}, \autoref{main theorem} was already known to Paradan and also follows from a different result of Harish-Chandra. In \autoref{multiplicities of irreducibles} we will describe another way to establish it by using the connection between Duister\-maat--Heck\-man measures in algebraic geometry and multiplicities in group representations.
\end{rem}

\begin{rem}
  Note that \autoref{main theorem measures} completely determines the measure $\nu$ from its projection onto $\mathfrak t^*$, since $\nu$ is by assumption concentrated on the union of the coadjoint orbits of maximal dimension.
  Similarly, the non-Abelian Duister\-maat--Heck\-man measure $\DuHe^K_M$ can be fully reconstructed from $\DuHe^T_M$ by using \autoref{main theorem}.

  We stress that it is oftentimes not necessary to explicitly compute the non-Abelian Duister\-maat-Heck\-man measure.
  Indeed, \autoref{main theorem} is of course by definition equivalent to
  \begin{equation*}
    \langle \DuHe^K_M, f \rangle = \langle \DuHe^T_M, \left(\prod_{\alpha > 0} \partial_\alpha \right) f \rangle
  \end{equation*}
  for all $f \in C_c^\infty(\mathfrak t^*_{>0})$, so that we can reduce the computation of averages over $\DuHe^K_M$ directly to integrations with respect to the Abelian Duister\-maat--Heck\-man measure (cf.\ proof of \autoref{purity example}).
\end{rem}

\begin{rem}
  \label{finite union of regular chambers}
  It follows from \autoref{main theorem} and the discussion in \autoref{notation} that, on each (open) regular chamber, the non-Abelian Duister\-maat--Heck\-man measure also has a polynomial density, namely the partial derivative in the directions of the negative roots of the density of the Abelian measure.
  However, there could still be non-zero measure on the critical walls separating the regular chambers. If we would like to exclude this then we need to understand the smoothness properties of the Abelian density function in the vicinity of critical walls, or, equivalently, the nature of the term by which the polynomial density changes when crossing a critical wall. If this jump term vanishes to order at least $R$ on the wall, then the Abelian density function is at least $R$-times weakly differentiable in the vicinity of the wall, and therefore the non-Abelian Duister\-maat--Heck\-man density is also absolutely continuous there. This vanishing condition can be checked explicitly for each critical wall using the jump formula described in \autoref{projective space}.

  In case the vanishing condition is satisfied, the non-Abelian moment polytope $\Delta_K(M)$ is equal to the closure of a finite union of regular chambers for the Abelian moment map: Indeed, on each regular chamber the density polynomial is either equal to zero, or it is non-zero on an open, dense subset.
\end{rem}

We cannot resist giving an easy application of \autoref{main theorem} to the problem of describing the sum of two coadjoint orbits $\mathcal O_\lambda + \mathcal O_\mu$.
Mathematically, one considers the diagonal action of $K$ on $\mathcal O_\lambda \times \mathcal O_\mu$, which is Hamiltonian with moment map $(X,Y) \mapsto X+Y$, and one would like to describe the associated moment polytope or Duister\-maat--Heck\-man measure.

\begin{cor}[{\cite{dooleyrepkawildberger93}}] 
  \label{main theorem horn}
  Let $\lambda \in \mathfrak t^*_{>0}$ and $\mu \in \mathfrak t^*_+$. Then,
  \begin{equation*}
    \DuHe^K_{\mathcal O_\lambda \times \mathcal O_\mu} =
    \sum_{w \in W} (-1)^{l(w)} \delta_{w \lambda} \star \DuHe^T_{\mathcal O_\mu},
  \end{equation*}
  where $l(w)$ is the length of the Weyl group element $w$.
\end{cor}
\begin{proof}
  Clearly, since $T$ is Abelian,
  \begin{equation*}
    \DuHe^T_{\mathcal O_\lambda \times \mathcal O_\mu} =
    \DuHe^T_{\mathcal O_\lambda} \star \DuHe^T_{\mathcal O_\mu}.
  \end{equation*}
  Since $\lambda + \mu \in \mathfrak t^*_{>0}$, \autoref{main assumption} is satisfied. Therefore, \autoref{main theorem} is applicable, and the assertion follows together with \eqref{harish chandra derivative},
  \begin{equation*}
    \DuHe^K_{\mathcal O_\lambda \times \mathcal O_\mu} =
    \left( \prod_{\alpha > 0} \partial_{-\alpha} \right) \DuHe^T_{\mathcal O_\lambda} \star \DuHe^T_{\mathcal O_\mu} =
    \sum_{w \in W} (-1)^{l(w)} \delta_{w \lambda} \star \DuHe^T_{\mathcal O_\mu}.
    \qedhere
  \end{equation*}
\end{proof}

The general case where both $\lambda$ and $\mu$ are contained in the boundary of the positive Weyl chamber can be treated as in \cite{dooleyrepkawildberger93} by taking limits. Of course we can also expand $\DuHe^T_{\mathcal O_\mu}$ as an alternating sum of convolutions by using \eqref{harish chandra} or its version for lower-dimensional coadjoint orbits \cite[Theorem 7.24]{berlinegetzlervergne03}.

\section{Algorithms for Duister\-maat--Heck\-man Measures}
\label{algorithms}

In this section we present two algorithms for computing Duister\-maat--Heck\-man measures. Both algorithms are based on the derivative principle from \autoref{derivative principle}, in that they first compute the Abelian measure and then take partial derivatives according to \autoref{main theorem}.

The first algorithm, the \emph{Heckman algorithm}, is based on the Heckman formula by Guillemin, Lerman and Sternberg, which expresses the Abelian measure as an alternating sum of iterated convolutions of Heaviside measures. The density function of each such convolution is piecewise polynomial and can be evaluated inductively using recent work of Boysal and Vergne. While very useful for computing low-dimensional examples, the resulting algorithm is rather inefficient due to the large number of summands.

Our second algorithm, the \emph{single-summand algorithm}, is based on another formula for the Abelian Duister\-maat--Heck\-man measure in the case where $M$ is the projective space of an arbitrary finite-dimensional representation. It turns out that this formula is equivalent to evaluating a single iterated convolution of the above form (hence the name of the algorithm). It can therefore be computed in a similar way, but much more efficiently. Since by passing to the purified double the quantum marginal problem can always be reduced to the case where $M$ is a projective space (\autoref{purification}), this solves the problem of computing eigenvalue distributions of reduced density matrices in complete generality.

\subsection{Heckman Algorithm}
\label{heckman}

Before stating the Heckman formula by Guillemin, Lerman and Sternberg, let us recall the following renormalization process as described in \cite{guilleminlermansternberg88}:

Suppose that there are only finitely many fixed points of the action of the maximal torus $T$ on $M$.
For each such fixed point $p \in M^T$, consider the induced representation of $T$ on the tangent space $T_p M$. The weights of this representation are called isotropy weights and we can always choose a vector $\gamma \in \mathfrak t^*$ which is non-orthogonal to all isotropy weights (for all tangent spaces). The process of multiplying by $-1$ those isotropy weights that have negative inner product with $\gamma$ is then called \emph{renormalization}, and the resulting weights are called renormalized weights. See \autoref{fixed-point data projective space} for a discussion of the case where $M$ is a projective space and \autoref{examples} for examples.

\begin{thm}[\cite{guilleminlermansternberg88}]
  \label{abelian heckman}
  Suppose that there are only finitely many torus fixed points $p \in M^T$. Denote by $n_p$ the number of isotropy weights in $T_p M$ that are multiplied by $-1$ during renormalization and by $\hat\omega_{p,1}, \ldots, \hat\omega_{p,n}$ the resulting renormalized weights. Then,
  \begin{equation*}
    \DuHe^T_M = \sum_{p \in M^T} (-1)^{n_p} \delta_{\Phi_T(p)} \star H_{\hat\omega_{p,1}} \star \ldots \star H_{\hat\omega_{p,n}},
  \end{equation*}
  with $H_{\hat\omega}$ the Heaviside measure defined by $\langle H_{\hat\omega}, f \rangle = \int_0^\infty dt f(\hat\omega t)$.
\end{thm}

In other words, the stationary phase approximation for the Fourier transform of an Abelian Duister\-maat--Heck\-man measure is exact. This generalizes the Harish-Chandra formula for coadjoint orbits \eqref{harish chandra}, which we used to establish \eqref{ableitungsformel coadjoint orbit}.

Observe that each summand of the Heckman formula can be written as the push-forward of the standard Lebesgue measure $dt$ on $\RR^n_{\geq 0}$ along a linear map of the form $\hat P \colon (t_k) \mapsto \sum_{k=1}^n t_k \hat\omega_k$, translated by $\Phi_T(p)$, since
\begin{equation}
  \label{convolution translation}
  H_{\hat\omega_1} \star \ldots \star H_{\hat\omega_n} =
  {\hat P}_*(H_{e_1} \star \ldots \star H_{e_n}) =
  {\hat P}_*(\restrict{dt}{\RR^n_{\geq 0}}).
\end{equation}
In a recent paper \cite{boysalvergne09}, Boysal and Vergne have analyzed general push-forward measures of this form under the assumption that the vectors $\hat\omega_k$ span a proper convex cone (i.e., a convex cone of maximal dimension that does not contain any straight line). This ensures that the measure is locally finite and absolutely continuous with respect to Lebesgue measure on $\mathfrak t^*$.
This assumption is certainly satisfied for the renormalized isotropy weights occurring in the Heckman formula (by the very definition of renormalization and our assumption that the Abelian moment polytope has maximal dimension).

Let us briefly review their results: It is well-known that the push-forward measure has a piecewise homogeneous polynomial density function of degree $n-r$. Here, the \emph{chambers} are the connected components of the complement of the cones spanned by at most $r-1$ of the weights $(\hat\omega_k)$. Except for the unbounded chamber, they are open convex cones. \emph{Walls} are by definition the convex cones spanned by $r-1$ linearly independent weights.\footnote{\label{non-compact footnote}In fact, the $p$-th summand of \autoref{abelian heckman} is precisely the Duister\-maat--Heck\-man measure corresponding to the isotropy representation of $T$ on the symplectic vector space $T_p M$, which is of course a non-compact symplectic manifold and, strictly speaking, does not fit into our setup. The decomposition of $\mathfrak t^*$ into regular chambers for the moment map of $M$ is refined by the common refinement of the chamber decompositions for the $T_p M$ (cf.\ \autoref{notation}).} Similarly to \autoref{notation}, if the common boundary of the closure of two chambers is of maximal dimension then this common boundary is a wall; moreover, every wall arises in this way. Note that the union of the walls is precisely the complement of the union of the chambers.\footnote{This is our reason for choosing a different definition for walls than the one used in \cite{boysalvergne09}. There, walls were defined as linear hyperplanes spanned by $r-1$ linearly independent vectors.}

Let $\hat\Delta_\pm$ be two adjacent chambers which are separated by a wall $\hat W$, and choose a normal vector $\hat\xi \in \mathfrak t^*$ pointing from $\hat\Delta_-$ to $\hat\Delta_+$. Order the weights such that precisely $\hat\omega_1, \ldots, \hat\omega_m$ lie on the linear hyperplane spanned by $\hat W$. In the following, we shall freely identify differential forms and the measures induced by them. Denote by $d\hat w$ the Lebesgue measure on the hyperplanes parallel to $\hat W$, normalized in such a way that
\begin{equation}
  \label{wall quotient measure cone}
  d\lambda = d\hat w \wedge d\hat\xi,
\end{equation}
where $d\hat\xi$ is the pullback of the standard volume form of $\RR$ along the coordinate function $\langle -, \hat\xi\rangle$. Denote by $\hat f_\pm$ the homogeneous polynomials describing the density function $\hat f$ on $\hat\Delta_\pm$. Finally, consider the push-forward of Lebesgue measure on $\RR^m_{\geq 0}$ along the linear map $\hat P_{\hat W} \colon (u_k) \mapsto \sum_{k=1}^m u_k \hat\omega_k$. Its density with respect to $d\hat w$ is given by a single homogeneous polynomial on the wall $\hat W$, since $\hat W$ is always contained in the closure of a chamber for $\hat P_{\hat W}$. Denote by $\hat f_{\hat W}$ any polynomial function extending it to all of $\mathfrak t^*$. Then the result of Boysal and Vergne is the following \cite[Theorem 1.1]{boysalvergne09}: The jump of the density function across the wall is given by
\begin{equation}
  \label{boysal vergne}
    \hat f_+(\hat\lambda) - \hat f_-(\hat\lambda)
  = \restrict{\Res}{z=0} \left(
      \hat f_{\hat W}(\partial_{\hat x})
      \frac
        {e^{\langle \hat\lambda, {\hat x} + z\hat\xi \rangle}}
        {\prod_{k=m}^n \langle \hat\omega_k, {\hat x} + z \hat\xi \rangle}
    \right)_{{\hat x}=0},
\end{equation}
where $\restrict{\Res}{z=0} g = a_{-1}$ is the residue of a formal Laurent series $g = \sum_k a_k z^k$. (The residue appears as part of an inversion formula for the Laplace transform.)

In the case where only a minimal number of weights lie on the linear hyperplane spanned by $\hat W$ ($m=r-1$), the wall polynomial $\hat f_{\hat W}$ can be chosen as a constant, since the corresponding push-forward map is merely a change of coordinates:

\begin{lem}
  \label{minimal wall jump cones}
  Suppose that precisely $r-1$ weights $\hat\omega_1, \ldots, \hat\omega_{r-1}$ lie on $\linspan {\hat W}$. Then,
  \begin{equation*}
    \hat f_{\hat W}^{-1} \equiv |d\lambda\left(\hat\omega_1, \ldots, \hat\omega_{r-1}, \frac {\hat\xi} {\norm{\hat\xi}^2}\right)|.
  \end{equation*}
\end{lem}
\begin{proof}
  Since the map $\hat P_{\hat W} \colon \RR^{r-1} \rightarrow \linspan \hat W, (u_k) \mapsto \sum_{k=1}^{r-1} u_k \hat\omega_k$ along which we push forward is a linear isomorphism, the polynomial $\hat f_{\hat W}$ can be chosen as the constant of proportionality between the push-forward of Lebesgue measure on $\RR^{r-1}$ and the measure $dw$. 
  We can compute its value by comparing the volume of the parallelotope spanned by the $(\hat\omega_k)$ with respect to the two measure. For the former measure, this is of course one, while for the latter it follows from \eqref{wall quotient measure cone} that
  \begin{equation*}
    d\lambda(\hat\omega_1, \ldots, \hat\omega_{r-1}, \hat\xi) =
    dw(\hat\omega_1, \ldots, \hat\omega_{r-1}) \,\norm{\hat\xi}^2.
    \qedhere
  \end{equation*}
\end{proof}

This immediately gives rise to the following inductive algorithm:

\begin{alg}
  \label{boysal vergne algorithm}
  The following algorithm computes the piecewise polynomial density function of the push-forward of Lebesgue measure on $\RR^n_{\geq 0}$ along $(t_k) \mapsto \sum_{k=1}^n t_k \hat\omega_k$ with respect to $d\lambda$.
  \begin{enumerate}
    \item Start with the unbounded chamber, where $\hat f \equiv 0$.
    \item Iteratively jump over walls $\hat W$ separating the current chamber with an adjacent chamber:
      \begin{enumerate}
        \item Denote by $\hat\omega_1, \ldots, \hat\omega_m$ the weights which lie on the hyperplane through $\hat W$.
        \item If the wall is minimal ($m=r-1$), compute $\hat f_{\hat W}$ via \autoref{minimal wall jump cones}.
        \item Otherwise, recursively apply \autoref{boysal vergne algorithm} to compute the piecewise polynomial density function of the push-forward of Lebesgue measure on $\RR^m_{\geq 0}$ along $(u_k) \mapsto \sum_{k=1}^m u_k \hat\omega_k$ with respect to $d\omega$.%
\footnote{This density of course only depends on the hyperplane through $\hat W$, and can therefore re-used for all other walls that span the same hyperplane.}
        On $\hat W$ itself, it is given by a single homogeneous polynomial.
        Choose any polynomial extension $\hat f_{\hat W}$ to all of $\mathfrak t^*$.
        \item Compute the density on the adjacent chamber using \eqref{boysal vergne}.
      \end{enumerate}
  \end{enumerate}
\end{alg}

If the set of renormalized weights is not multiplicity-free then the one-dimensional walls are not necessarily minimal; \autoref{boysal vergne algorithm} can be modified in a straightforward way to include $r=1$ as an additional base case.

By combining \autoref{boysal vergne algorithm} with the Heckman formula, we arrive at the following algorithm for computing Duister\-maat--Heck\-man measures. We shall call it the \emph{(Abelian) Heckman algorithm}.

\begin{alg}
  \label{abelian heckman algorithm}
  Under the assumptions and using the notation of \autoref{abelian heckman}, the following algorithm computes the piecewise polynomial density function of the Abelian Duister\-maat--Heck\-man measure:
  \begin{enumerate}
    \item Compute the density of each of the $|M^T|$ iterated convolutions
      $\delta_{\Phi_K(p)} \star H_{\hat\omega_{p,1}} \star \ldots \star H_{\hat\omega_{p,n}}$
      using \autoref{boysal vergne algorithm}.
    \item Form their alternating sum according to \autoref{abelian heckman}.
  \end{enumerate}
  The non-Abelian Duister\-maat--Heck\-man measure can then be computed via \autoref{main theorem}.
  By passing to its support, we can also determine the non-Abelian moment polytope (cf.\ \autoref{finite union of regular chambers}).
\end{alg}

The algorithm as we have stated it assumes that the fixed-point data is part of the input. Let us describe it in the situations we are interested in:

\begin{rem}
  \label{fixed-point data projective space}
  Consider the projective space $M = \PP(V)$ associated with an arbitrary finite-dimensional, unitary $K$-representation $V$. Torus fixed points in $M$ correspond to weight vectors in $V$. Therefore, $M^T$ is finite if and only if all the weight spaces of $V$ are one-dimensional.
  If this is the case, let $V = \bigoplus_{k=0}^n \CC v_k$ be the weight-space decomposition, with $v_k$ weight vectors of pairwise distinct weight $\omega_k$, so that the torus fixed points are precisely the points $[v_0], \ldots, [v_n] \in M$. Then, \emph{before renormalization}, the isotropy weights in $T_{[v_k]} M$ are given by the vectors $\omega_l - \omega_k$ for $l \neq k$.

  Note that the representations associated with the pure-state quantum marginal problems displayed in \autoref{QMP summary table} indeed have one-dimensional weight spaces, so that \autoref{abelian heckman algorithm} is directly applicable: This is obvious for $\CC^{d_1} \otimes \ldots \otimes \CC^{d_N}$ and can also be verified for $\Sym^N(\CC^d)$ and $\Alt^N(\CC^d)$ (e.g., by observing that any single-row or single-column semistandard tableaux is already determined by its weight vector). However, other irreducible representations of $\SU(d)$, which correspond to indistinguishable particles of more exotic statistics, typically have weight spaces of dimension larger than one \cite{fulton97}.
\end{rem}

\begin{rem}
  \label{heckman for coadjoint orbit reductions}
  Consider more generally the action of $T$ on a coadjoint $\tilde K$-orbit $M = \mathcal O_{\tilde\lambda}$ induced by a group homomorphism $\varphi \colon T \rightarrow \tilde T \subseteq \tilde K$. Even though this action might have infinitely many fixed points, there is an obvious way to write down an alternating sum formula for $\DuHe^T_{\mathcal O_{\tilde\lambda}}$: Note that it follows directly from \eqref{restriction to subgroups} that
  \begin{equation*}
    \DuHe^T_{\mathcal O_{\tilde\lambda}} =
    \pi_* \DuHe^{\tilde T}_{\mathcal O_{\tilde\lambda}},
  \end{equation*}
  where $\pi = (d\varphi)^*$ is the dual map $\mathfrak {\tilde t}^* \rightarrow \mathfrak t^*$.
  Therefore, we can simply take the Abelian Heckman formula for the $\tilde T$-action (which is always applicable since the fixed point set of $\tilde T$ is the Weyl orbit of $\tilde\lambda$, hence finite), and push forward each summand along $\pi$. In the case of a maximal-dimensional coadjoint orbit and for a suitable choice of renormalization direction, the result is just the push-forward of the Harish-Chandra formula \eqref{harish chandra},
  \begin{equation}
    \label{heckman for maximal dimensional coadjoint orbit reductions}
    \DuHe^T_{\mathcal O_{\tilde\lambda}} = \sum_{\tilde w \in \tilde W} (-1)^{l(\tilde w)} \delta_{\pi(\tilde w \tilde\lambda)} \star H_{-\pi(\tilde \alpha_1)} \star \ldots \star H_{-\pi(\tilde \alpha_{\tilde R})},
  \end{equation}
  with $\tilde\alpha_1, \ldots, \tilde\alpha_{\tilde R}$ the positive roots of $\tilde K$.
  The formula for lower-dimensional coadjoint orbits can be obtained by using \cite[Theorem 7.24]{berlinegetzlervergne03} instead of \eqref{harish chandra}.

  In particular, this approach allows the computation of the Abelian Duister\-maat--Heck\-man measure for arbitrary setups of the quantum marginal problem by an obvious variant of \autoref{abelian heckman algorithm}.
\end{rem}

While \autoref{abelian heckman algorithm} and the variant described in \autoref{heckman for coadjoint orbit reductions} are useful for computing low-dimensional examples, any approach relying on the Heckman formula has the major problem that the number of summands in the Heckman formula is typically very large (e.g., it is exponential in the number of distinguishable particles or fermions). Moreover, even though the Boysal--Vergne algorithm computes the density of a single summand chamber-by-chamber, this is less straightforward for the alternating sum, where all summands have to be evaluated in parallel. In \autoref{projective space} below we will therefore derive an algorithm which does not suffer from these problems.

\medskip

There is also a non-Abelian Heckman formula due to Guillemin and Prato \cite{guilleminprato90} (which suffers from the same problems). It can be deduced directly from the Abelian one by applying the derivative principle:

\begin{thm}[{\cite[(2.15)]{guilleminprato90}}]
  \label{non-abelian heckman}
  Suppose that there are only finitely many torus fixed points $p \in M^T$ and that in each tangent space $T_p M$ each positive root $\alpha > 0$ or its negative occurs as an isotropy weight. Denote by $n_p$ the number of isotropy weights in $T_p M$ that are multiplied by $-1$ during renormalization. For each positive root $\alpha > 0$ and in each $T_p M$, remove either $\alpha$ or $-\alpha$ from the list of renormalized isotropy weights. Denote the remaining weights by $\hat\omega_{p,1}, \ldots, \hat\omega_{p,n-R}$, and let $k_p$ be the number of negative roots that have been removed. Then,
  \begin{equation*}
    \DuHe^K_M = \sum_{p \in M^T} (-1)^{n_p+k_p} \delta_{\Phi_K(p)} \star H_{\hat\omega_{p,1}} \star \ldots \star H_{\hat\omega_{p,n-R}} \Big|_{\mathfrak t^*_+}.
  \end{equation*}
  In particular, the second assumption is satisfied when the moment map $\Phi_K$ sends each torus fixed points to the interior of a Weyl chamber.
\end{thm}
\begin{proof}
  Since $\partial_{\hat\omega} H_{\pm{\hat\omega}} = \pm\delta_0$ (cf.\ the proof of \eqref{ableitungsformel coadjoint orbit}), the asserted formula follows at once by combining \autoref{main theorem} with \autoref{abelian heckman}.

  Only the final remark needs elaboration: As observed by Guillemin and Prato, the assumption that $\Phi_K(p) \in W \cdot \mathfrak t^*_{>0}$ implies that the $K$-stabilizer at each fixed point $p$ is precisely $T$, so that the infinitesimal action of $K$ generates a copy of $\mathfrak k / \mathfrak t$ inside the tangent space $T_p M$. Therefore, at any fixed point $p$, each positive root $\alpha > 0$ or its negative occurs as an isotropy weight.
\end{proof}

This gives rise to an obvious non-Abelian variant of \autoref{abelian heckman algorithm}:

\begin{alg}
  \label{non-abelian heckman algorithm}
  Under the assumptions and using the notation of \autoref{non-abelian heckman}, the following algorithm computes the non-Abelian Duister\-maat--Heck\-man measure:
  \begin{enumerate}
    \item Compute the $|M^T|$ iterated convolutions
      $\delta_{\Phi_K(p)} \star H_{\hat\omega_{p,1}} \star \ldots \star H_{\hat\omega_{p,n-R}}$
      using \autoref{boysal vergne algorithm} (see \autoref{subtle remark}).
    \item Form their alternating sum according to \autoref{non-abelian heckman}.
  \end{enumerate}
  By passing to its support, we can also determine the non-Abelian moment polytope (cf.\ \autoref{finite union of regular chambers}).
\end{alg}

\begin{rem}
  \label{subtle remark}
  There is a slight subtlety involved with the formulation of step (1) of \autoref{non-abelian heckman algorithm}: In case the renormalized isotropy weights $\hat\omega_{p,1}, \ldots, \hat\omega_{p,n-R}$ in some $T_p M$ do not span all of $\mathfrak t^*$, the corresponding iterated convolution is of course not absolutely continuous with respect to $d\lambda$, and \autoref{boysal vergne algorithm} cannot be applied directly (see, e.g., the first proof of \autoref{two qubits non-abelian}). Instead, we need to replace $\mathfrak t^*$ by the span of the $\hat\omega_{p,k}$ and apply \autoref{boysal vergne algorithm} accordingly.
\end{rem}


In \autoref{examples} we will use both the Abelian and the non-Abelian version of the Heckman algorithm to compute the eigenvalue distribution of the reduced density matrices of a random pure state of two qubits (\autoref{two qubits non-abelian}) and of $N$ bosonic qubits (\autoref{sym N qubits abelian}), as well as of random mixed states of two qubits (\autoref{bravyi non-abelian}).

\subsection{Single-Summand Algorithm for Projective Space}
\label{projective space}

We will now derive explicit formulas for the Duister\-maat--Heck\-man measure associated with a projective space, $M = \PP(V)$, where $V$ is a $(n+1)$-dimensional unitary representation of $K$, and where $M$ is equipped with the Fubini--Study symplectic form $\omega_\text{FS}$, normalized in such a way that its Liouville measure is equal to $\frac 1 {n!}$. The $K$-action is Hamiltonian, and a canonical moment map is given by \cite{kirwan84}
\begin{equation}
  \label{projective space non-abelian moment map}
  \Phi_K \colon
  \PP(V) \rightarrow \mathfrak k^*, \quad
  [v] \mapsto \left( X \mapsto \frac 1 i  \frac {\langle v, X v \rangle} {\langle v, v \rangle} \right).
\end{equation}

We start by decomposing the representation $V$ into one-dimensional weight spaces, $V = \bigoplus_{k=0}^n \CC v_k$, where $v_k$ is a weight vector of weight $\omega_k$ (repetitions allowed).
In the corresponding homogeneous coordinates, the Abelian moment map has the following simple form,
\begin{equation}
  \label{projective space abelian moment map}
  \Phi_T \colon \PP(V) \rightarrow \mathfrak t^*, \quad
  [z_0 : \ldots : z_n] \mapsto \frac {\sum_{k=0}^n |z_k|^2 \omega_k} {\sum_{k=0}^n |z_k|^2},
\end{equation}
and it is straightforward to see that the Abelian Duister\-maat--Heck\-man measure can be written as the push-forward of Lebesgue measure on the standard simplex along a linear map:

\begin{prp}
  \label{projective space abelian via standard simplex}
  We have
  \begin{equation*}
    \DuHe^T_{\PP(V)} = P_*(\restrict{dp}{\Delta_n}).
  \end{equation*}
  Here, $P$ is the linear map $\RR^{n+1} \rightarrow \mathfrak t^*, (t_k) \mapsto \sum_k t_k \omega_k$, and $dp$ is Lebesgue measure on the affine hyperplane $\Hbf := \{ (t_k) : \sum_k t_k = 1 \} \subseteq \RR^{n+1}$, normalized in such a way that the standard simplex $\Delta_n := 
  \{ (p_k) : p_k \geq 0, \sum_{k=0}^n p_k = 1 \}$ has measure $\frac 1 {n!}$.
\end{prp}
\begin{proof}
  The Fubini-Study measure is the push-forward of the usual round measure on the unit sphere $S^{2n+1} \cong \{ (z_0,\ldots,z_n) : \abs{z_0}^2 + \ldots + \abs{z_n}^2 = 1 \} \subseteq V$ along the quotient map $(z_0,\ldots,z_n) \mapsto [z_0:\ldots:z_n]$, normalized to total volume $\frac 1 {n!}$. On the other hand, the round measure on the unit sphere induces Lebesgue measure on the standard simplex when pushed forward along the map $(z_0,\ldots,z_n) \mapsto (|z_0|^2,\ldots,|z_n|^2)$. 
  Thus the claim follows from comparing \eqref{projective space abelian moment map} with $P \colon (t_k) \mapsto \sum_{k=0}^n t_k \omega_k$.
\end{proof}

\begin{rem}
  In other words, \autoref{projective space abelian via standard simplex} is proved by factoring the action of $T$ over the action of the maximal torus of $\SU(V)$, for which $\PP(V)$ is a symplectic toric manifold.
\end{rem}

Denote by $dp/d\lambda$ a differential form corresponding to Lebesgue measure on the affine subspaces $P^{-1}(\lambda) \cap \Hbf$, normalized in such a way that
\begin{equation}
  \label{quotient form polytope}
  dp = dp/d\lambda \wedge P^*(d\lambda)
\end{equation}
when restricted to the affine hyperplane $\Hbf$.

\begin{prp}
  \label{density in polytope picture}
  The density function $f \colon \mathfrak t^* \rightarrow [0,\infty)$ of the Abelian Duister\-maat--Heck\-man measure is given by
  \begin{equation*}
    f(\lambda) = \vol \, \{ p_k \geq 0 : \sum_{k=0}^n p_k \omega_k = \lambda, \sum_{k=0}^n p_k = 1 \},
  \end{equation*}
  where the volume is measured with respect to the measure induced by $dp/d\lambda$ on $P^{-1}(\lambda) \cap \Hbf$.
\end{prp}
\begin{proof}
  For all test functions $g \in C_b(\mathfrak t^*)$, we have
  \begin{equation*}
    \langle \DuHe^T_{\PP(V)}, g \rangle =
    \int_{\Delta_n} dp \, g(P(p)) =
    \int_{\mathfrak t^*} d\lambda \left( \int_{P^{-1}(\lambda) \cap \Delta_n} dp/d\lambda \right) g(\lambda),
  \end{equation*}
  by using \eqref{quotient form polytope} and Fubini's theorem for the fibration $\restrict P {\Hbf}$ \cite[pp.\ 307]{guilleminsternberg77}.
\end{proof}

That is, the Abelian Duister\-maat--Heck\-man density measures the volume of a family of convex polytopes parametrized by $\mathfrak t^*$. This is also true for the density of the iterated convolutions studied in \autoref{heckman} (see \eqref{single summand density} below).
There are exact numerical schemes that can be used to compute the polynomial density functions on each regular chamber which have already been implemented in software packages, e.g., the parametric extension of Barvinok's algorithm \cite{barvinok93} described in \cite{verdoolaegeseghirbeylsetal07,verdoolaegebruynooghe08}. We will not pursue this route any further. However, in \autoref{multiplicities for projective space} we will show that its ``quantized'' counterpart gives rise to an efficient way of computing the corresponding representation-theoretic quantities (in particular, the Kronecker coefficients).

\medskip

In the following, we will instead describe a combinatorial algorithm based on the same principles as our Heckman algorithm.
Before doing so, let us determine explicitly the regular chambers for the Abelian moment map, i.e., the connected components of the set of regular values of $\Phi_T$, each on which the measure is given by a polynomial. For this, we define the \emph{support} of a point $p = [v] \in \PP(V)$ as the set of weights which contribute to the weight-space decomposition of $v$,
\begin{equation*}
  \supp p := \{ \omega_k : z_k \neq 0, p = [z_0 : \ldots : z_n] \}.
\end{equation*}
The significance of this definition is that the support of a point already fully determines whether it is regular or critical:

\begin{lem}
  \label{regular points}
  Let $p \in P(V)$.
  Then $p$ is a regular point of the Abelian moment map if and only if
  \begin{equation*}
    \linspan \{ \omega - \omega' : \omega, \omega' \in \supp p \} = \mathfrak t^*.
  \end{equation*}
\end{lem}
\begin{proof}
  It follows readily from the definition of the moment map that a point $p$ is regular if and only if $\mathfrak t_p$, the Lie algebra of its stabilizer, is trivial \cite[Lemma 2.1]{guilleminsternberg82}. But $\mathfrak t_p$ is already determined by the support of $p$:
  \begin{equation*}
    \mathfrak t_p = \{ X \in \mathfrak t: \omega(X) = \omega'(X) \quad \forall \omega, \omega' \in \supp p \}
  \end{equation*}
  This is the annihilator of the linear span in the statement of the lemma.
\end{proof}

We arrive at the following characterization of the set of critical values of the Abelian moment map:

\begin{prp}
  \label{critical values}
  The set of critical values of $\Phi_T$ is the union of all convex hulls of subsets containing (at most) $r$ weights,
  \begin{equation*}
    \bigcup_{\#I = r} \conv \{ \omega_k : k \in I \} =
    \bigcup_{\#I \leq r} \conv \{ \omega_k : k \in I \}.
  \end{equation*}
\end{prp}
\begin{proof}
  It is clear from \eqref{projective space abelian moment map} and \autoref{regular points} that the convex hull of any subset of weights of cardinality at most $r$ consists of critical values. The converse follows from Carath\'{e}odory's theorem.
\end{proof}

From this description we can easily determine the regular chambers and critical walls. Observe again that there is a single unbounded regular chamber.

\medskip

We will now use the result of Boysal and Vergne described in \autoref{heckman} to derive intrinsic formulas for the jumps of the Duister\-maat--Heck\-man density when crossing a critical wall. Recall that the measures they consider are push-forwards of Lebesgue measure on the convex cone $\RR^{n+1}_{\geq 0}$ rather than of Lebesgue measure on the standard simplex $\Delta_n$, which is of course the intersection of $\RR^{n+1}_{\geq 0}$ with the affine hyperplane $\Hbf = \{ (t_k) : \sum_{k=0}^n t_k = 1 \}$.
It is however straightforward to translate between both pictures:
In order to avoid confusion, we shall use the same convention as in \autoref{heckman} that hatted quantities correspond to the Boysal--Vergne picture.
Let us consider the ``extended'' weights $\hat\omega_k := (\omega_k,1) \in \mathfrak t^* \oplus \RR$ ($k=0,\ldots,n$) together with the corresponding linear map
\begin{equation*}
  \hat P \colon \RR^{n+1} \rightarrow \mathfrak t^* \oplus \RR, \quad (t_k) \mapsto \sum_{k=0}^n t_k \hat\omega_k = (P(t_0, \ldots, t_n), \sum_{k=0}^n t_k).
\end{equation*}
Denote by $dt$ standard Lebesgue measure on $\RR^{n+1}$ and equip $\mathfrak t^* \oplus \RR$ with the measure $d\hat\lambda = d\lambda ds$, where $ds$ is standard Lebesgue measure on $\RR$. Choose a differential form $dt/d\hat\lambda$ inducing Lebesgue measure on the fibers of $\hat P$, normalized in such a way that
\begin{equation}
  \label{quotient form cone}
  dt = dt/d\hat\lambda \wedge {\hat P}^*(d\hat\lambda) = dt/d\hat\lambda \wedge P^*(d\lambda) \wedge (dt_0 + \ldots + dt_N).
\end{equation}
Then one can establish just as in the proof of \autoref{density in polytope picture} the following formula for the density function of the push-forward of Lebesgue measure on $\RR^{n+1}_{\geq 0}$ along $\hat P$ with respect to $d\hat\lambda = d\lambda ds$,
\begin{equation}
  \label{single summand density}
  \hat f(\lambda, s) = \vol \, \{ t_k \geq 0 : \sum_{k=0}^n t_k \omega_k = \lambda, \sum_{k=0}^n t_k = s \},
\end{equation}
where the volume is measured with respect to $dt/d\hat\lambda$.
But comparing \eqref{quotient form polytope} and \eqref{quotient form cone} and noting that $dt = dp \wedge (dt_0 + \ldots + dt_N)$ on $\Hbf$%
, we see that in fact $dt/d\hat\lambda$ and $dp/d\lambda$ induce the same measure on the fibers $P^{-1}(\lambda) = \hat P^{-1}(\lambda, 1)$, so that
\begin{equation}
  \label{density transfer eqn}
  \DuHe^T_{\PP(V)} =
  f(\lambda) \, d\lambda =
  \hat f(\lambda, 1) \, d\lambda.
\end{equation}
This shows that we can work equivalently in the convex cone picture of Boysal and Vergne.%
\footnote{The push-forward of Lebesgue measure on $\RR^{n+1}_{\geq 0}$ along $\hat P$ can also be understood as the Duister\-maat--Heck\-man measure associated with the Hamiltonian $T \times \U(1)$-action on the complex vector space $V$, where $\U(1)$ acts by scalar multiplication (cf.\ \autoref{non-compact footnote}).}

\medskip

We shall now describe the jump formula. Let $W$ be a critical wall separating regular chambers $\Delta_\pm \subseteq \mathfrak t^*$, and choose a normal vector $\xi \in \mathfrak t^*$ pointing from $\Delta_-$ to $\Delta_+$. Order the weights such that precisely $\omega_0, \ldots, \omega_{m-1}$ lie on $W$. Denote by $dw$ Lebesgue measure on the hyperplanes parallel to $W$, normalized in such a way that
\begin{equation}
  \label{wall quotient measure polytope}
  d\lambda = dw \wedge d\xi,
\end{equation}
where $d\xi$ is the pullback of the standard volume form of $\RR$ along the coordinate function $\langle -, \xi\rangle$. Denote by $f_\pm$ the polynomials describing the density function $f$ on the regular chambers $\Delta_\pm$. Finally, consider the Duister\-maat--Heck\-man measure for the action of $T$ on the projective space over $V_W = \bigoplus_{k=0}^{m-1} \CC v_k$, the direct sum of the weight spaces corresponding to the weights which lie on the hyperplane through $W$. Its density with respect to $dw$ is given by a single polynomial on the critical wall $W$, since $W$ is always contained in the closure of a regular chamber for $\PP(V_W)$. Choose any polynomial function $f_W$ extending it to all of $\mathfrak t^*$.

\begin{prp}
  \label{wall jump abelian}
  The jump of the Abelian Duister\-maat--Heck\-man density across the critical wall is given by
  \begin{equation*}
    f_+(\lambda) -f_-(\lambda)
  = \restrict{\Res}{z=0} \left(
      \hat f_{\hat W}(\partial_x, \partial_y)
      \frac
        {e^{z \langle \lambda - \omega_0, \xi \rangle + \langle \lambda, x \rangle + y}}
        {\prod_{k=m}^n z \langle \omega_k - \omega_0, \xi \rangle + \langle \omega_k, x \rangle + y}
    \right)_{x=0, y=0}.
  \end{equation*}
  Here, $\hat f_{\hat W}(\lambda,s) = s^{m-r} f_W(\frac \lambda s)$ is the homogeneous ``extension'' of $f_W$ to $\mathfrak t^* \oplus \RR$.
\end{prp}
\begin{proof}
  The convex cones $\hat\Delta_\pm$ through $\Delta_\pm \times \{1\}$ are chambers in the sense of Boysal and Vergne. They are separated by a wall $\hat W$, namely the convex cone through $W \times \{1\}$. Note that $\hat\xi = (\xi,-\braket{\omega_0,\xi})$ is a normal vector to $\hat W$. Denote by $\hat f_\pm$ the homogeneous polynomials describing the density function of the push-forward of Lebesgue measure on $\RR^{n+1}_{\geq 0}$ along $\hat P$. It is clear that $d\hat w = ds \wedge dw$ induces Lebesgue measure on $\hat W$ and that it is normalized in such a way that $d\hat\lambda = d\hat w \wedge d\hat\xi$.
  By \eqref{density transfer eqn} and the jump formula \eqref{boysal vergne} of Boysal and Vergne, we have
  \begin{align*}
    f_+(\lambda) - f_-(\lambda)
  = \hat f_+(\lambda, 1) - \hat f_-(\lambda, 1)
  = \restrict{\Res}{z=0} \left(
      \hat f_{\hat W}(\partial_{\hat x})
      \frac
        {e^{z \langle (\lambda, 1), \hat x + z \hat \xi \rangle}}
        {\prod_{k=m}^n \langle \hat\omega_k, \hat x + z \hat\xi \rangle}
    \right)_{\hat x=0}.
  \end{align*}
  The polynomial $\hat f_{\hat W}$ as defined above agrees with its original definition in \autoref{heckman}, since it is a homogeneous polynomial and can thus be reconstructed from $f_W$, which by \eqref{density transfer eqn} is its restriction to the slice $\mathfrak t^* \times \{1\}$, by the formula given above.
  Writing $\hat x = (x,y) \in \mathfrak t^* \oplus \RR$ and expanding the hatted quantities, we arrive at the assertion.
\end{proof}

As in \autoref{heckman}, the case where only a minimal number of weights lie on the affine hyperplane through $W$ is particularly simple to evaluate:

\begin{lem}
  \label{minimal wall jump constant polytope}
  Suppose that precisely $r$ weights $\omega_0, \ldots, \omega_{r-1}$ lie on the affine hyperplane through $W$. Then,
  \begin{equation*}
    \hat f_{\hat W}^{-1} \equiv f_W^{-1} \equiv |d\lambda\left(\omega_1-\omega_0, \ldots, \omega_{r-1}-\omega_0, \frac{\xi}{\norm{\xi}^2}\right)|.
  \end{equation*}
\end{lem}
\begin{proof}
  We argue as in the proof of \autoref{minimal wall jump cones}: In view of \autoref{projective space abelian via standard simplex} and the minimality assumption, the map $(q_k) \mapsto \sum_{k=0}^{r-1} q_k \omega_k$ along which we push forward is an isomorphism, and $f_W$ is equal to the constant of proportionality between the push-forward of Lebesgue measure on $\Hbf$ (normalized in such a way that the standard simplex has measure $\tfrac 1 {d!}$) and the measure $dw$. We can compute this constant by comparing the volume of the parallelotope spanned by the $(\omega_k)$: For the former measure this constant is one (by its very normalization), while for the latter it follows from \eqref{wall quotient measure polytope} that
  \begin{equation*}
    d\lambda(\omega_1-\omega_0, \ldots, \omega_{r-1}-\omega_0, \xi) = d\omega(\omega_1-\omega_0, \ldots, \omega_{r-1}-\omega_0) \,\norm{\xi}^2. \qedhere
  \end{equation*}
\end{proof}

These results give rise to the following inductive algorithm for computing the Abelian and non-Abelian Duister\-maat--Heck\-man measure of a projective space. We will call it the \emph{single-summand algorithm}, since in view of \eqref{density transfer eqn} it amounts to computing a push-forward measure that is equivalent to a single summand of the Abelian Heckman formula (cf.\ \autoref{abelian heckman}).

\begin{alg}
  \label{projective space algorithm}
  The following algorithm computes the piecewise polynomial density function of the Abelian Duister\-maat--Heck\-man measure of the projective space $\PP(V)$:
  \begin{enumerate}
    \item Start with the unbounded regular chamber, where $f \equiv 0$.
    \item Iteratively jump over critical walls $W$ separating the current regular chamber with an adjacent regular chamber:
    \begin{enumerate}
      \item Denote by $\omega_0, \ldots, \omega_{m-1}$ the weights which lie on the hyperplane through $W$.
      \item If the wall is minimal ($m=r$), compute $f_W$ via \autoref{minimal wall jump constant polytope}.
      \item Otherwise, recursively apply \autoref{projective space algorithm} to compute the piecewise polynomial density of the Abelian Duister\-maat--Heck\-man measure of $\PP(V_W)$, where $V_W = \bigoplus_{k=0}^{m-1} \CC v_k$ is the direct sum of the weight spaces for the weights in (a).%
\footnote{This density of course only depends on the hyperplane through $W$, and can therefore be re-used for all other critical walls that lie on the same hyperplane.}
      On $W$ itself, it is given by a single polynomial. Choose any polynomial extension $f_W$ to all of $\mathfrak t^*$.
      \item Compute the density on the adjacent chamber using \autoref{wall jump abelian}.
    \end{enumerate}
  \end{enumerate}
  The non-Abelian Duister\-maat--Heck\-man measure can then be computed via \autoref{main theorem}.
  By passing to its support, we can also determine the non-Abelian moment polytope (cf.\ \autoref{finite union of regular chambers}).
\end{alg}

If there are degenerate weight spaces, not every zero-dimensional wall will be minimal.
\autoref{projective space algorithm} can be straightforwardly adapted by including $r = 0$ as an additional base case (here the moment polytope is a single point and the density a scalar that can be determined from our normalization conventions). 

\begin{rem}
  \label{complete solution}
  In view of \eqref{eigenvalue distribution dist} and \eqref{eigenvalue distribution indist} and by passing to the purified double (\autoref{purification}), \autoref{projective space algorithm} solves the problem of computing the eigenvalue distribution of reduced density matrices in complete generality.
\end{rem}

We conclude this section by explicitly stating the Abelian and non-Abelian jump formula for the case where only a minimal number of weights lie on the affine hyperplane through the wall. They will be used later for computing examples.

\begin{cor}
  \label{minimal wall jump abelian}
  Suppose that precisely $r$ weights $\omega_0, \ldots, \omega_{r-1}$ lie on the affine hyperplane through the critical wall $W$. Then the jump of the Abelian Duister\-maat--Heck\-man density across the wall is given by
  \begin{equation*}
    f_+(\lambda) - f_-(\lambda) =
    f_W
    \left( \prod_{k=r}^n \langle \omega_k - \omega_0, \xi \rangle \right)^{-1}
    \frac{\langle \lambda - \omega_0, \xi \rangle^{n-r}}{(n-r)!},
  \end{equation*}
  where $f_W$ is the constant from \autoref{minimal wall jump constant polytope}.
\end{cor}
\begin{proof}
  This follows immediately from \autoref{wall jump abelian} by pulling out the constant $f_W$, setting $x = y = 0$ and evaluating the residue at $z=0$.
\end{proof}

The non-Abelian formula follows directly by applying \autoref{main theorem}:

\begin{cor}
  \label{minimal wall jump non-abelian}
  Suppose that precisely $r$ weights $\omega_0, \ldots, \omega_{r-1}$ lie on the affine hyperplane through the critical wall $W$, and that $n-r \geq R$, so that the non-Abelian Duister\-maat--Heck\-man measure of $\PP(V)$ is absolutely continuous in the vicinity of $W$.
  Denote by $f^K_\pm$ the polynomials describing its density on the regular chambers. Then the jump across the wall is given by
  \begin{align*}
    f^K_+(\lambda) - f^K_-(\lambda)
    = f_W
    \left( \prod_{k=r}^n \langle \omega_k - \omega_0, \xi \rangle \right)^{-1}
    \left( \prod_{\alpha > 0}  - \langle \alpha, \xi \rangle \right)
    \frac{\langle \lambda - \omega_0, \xi \rangle^{n-r-R}}{(n-r-R)!},
  \end{align*}
  where $f_W$ is the constant from \autoref{minimal wall jump constant polytope}.
\end{cor}

\begin{rem}
  \autoref{minimal wall jump abelian} has already been established in \cite{guilleminlermansternberg88}, where the authors also envisaged an algorithm similar to our Heckman algorithm. They did however not have a general jump formula such as \eqref{boysal vergne} at their avail. Instead, they had to resort to an inexact formula which in general only holds in highest order (in the distance to the wall).
\end{rem}

\section{Examples}
\label{examples}

In this section we illustrate our algorithms by computing some eigenvalue distributions of reduced density matrices. The global quantum states will always be chosen according to one of the invariant probability measures described in \autoref{density matrices}. Many of our examples will involve \emph{qubits}, i.e., quantum systems modeled by two-dimensional Hilbert spaces, so that the algorithms can be nicely visualized. But of course our algorithms can be used to determine the eigenvalue distributions for arbitrary instances of the quantum marginal problem (see \autoref{complete solution}).

\subsection{Pure States of Multiple Qubits}
\label{pure states qubits}

We start by considering pure states of $N$ qubits, where $K = \SU(2)^N$ acts on $M = \PP((\CC^2)^{\otimes N})$ by tensor products (cf.\ \autoref{reduced density matrices}). It will be convenient to identify $\mathfrak t^* \cong \RR^N$ in such a way that the positive Weyl chamber corresponds to the cone $\RR^N_{\geq 0}$ and the fundamental weights to the standard basis vectors $e_j = (\delta_{j,k})$ ($k=1,\ldots,N$). That is, if $\lambda = (\lambda^{(j)}) \in \mathfrak t^*$ then we will by slight abuse of notation identify $\lambda^{(j)}$ with the scalar $i(\lambda^{(j)}_1 - \lambda^{(j)}_2) =$ $2 i \lambda^{(j)}_1$.
It follows that $d\lambda$ is simply the usual Lebesgue measure on $\RR^N$, that the symplectic volume polynomial is given by $p_K(\lambda) = \lambda^{(1)} \cdots \lambda^{(N)}$, and that the positive roots are $2 e_1, \ldots, 2 e_N$ (cf.\ \autoref{notation} and \eqref{volume coadjoint su orbit}).
Moreover, \eqref{spectra to positive weyl chamber} amounts to assigning to a point $(\lambda^{(j)}) \in \RR^N$ the tuple $(\rho_1,\ldots,\rho_N)$ of diagonal density matrices acting on $\CC^2$, where $\rho_j$ has maximal eigenvalue $\hat\lambda^{(j)}_{\max} = \frac 1 2 + i \lambda^{(j)}_1 = \tfrac {1+\lambda^{(j)}} 2$.

We first discuss in detail the toy example of $N=2$ qubits, demonstrating both the non-Abelian Heckman algorithm and the single-summand algorithm.

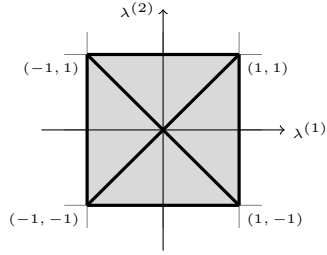
\begin{figure}
    \centering
  \begin{equation*}
    \begin{tikzpicture}
      \draw[help lines] (-1.3,-1.3) grid (1.3,1.3);
      \draw[->] (-1.6,0) -- (1.6,0) node[right] {\tiny $\lambda^{(1)}$};
      \draw[->] (0,-1.6) -- (0,1.6) node[left] {\tiny $\lambda^{(2)}$};
      \draw[fill=gray,opacity=0.3] (1,1) -- (-1,1) -- (-1,-1) -- (1,-1) -- (1,1);
      \draw[very thick] (1,1) -- (-1,1);
      \draw[very thick] (1,-1) -- (-1,-1);
      \draw[very thick] (1,-1) -- (1,1);
      \draw[very thick] (-1,-1) -- (-1,1);
      \draw[very thick] (1,1) -- (-1,-1);
      \draw[very thick] (-1,1) -- (1,-1);
      \draw(1,1) node[below right] {\tiny $(1,1)$};
      \draw(1,-1) node[below right] {\tiny $(1,-1)$};
      \draw(-1,1) node[below left] {\tiny $(-1,1)$};
      \draw(-1,-1) node[below left] {\tiny $(-1,-1)$};
    \end{tikzpicture}
  \end{equation*}
  \caption{Abelian moment polytope of two qubits (grey square) and its decomposition into four bounded regular chambers by the critical walls (thick lines).}
  \label{two qubits figure}
\end{figure}

\begin{prp}
  \label{two qubits non-abelian}
  The non-Abelian Duister\-maat--Heck\-man measure for the action of $\SU(2) \times \SU(2)$ on $\PP(\CC^2 \otimes \CC^2)$ is given by
  \begin{equation*}
    \langle \DuHe^{\SU(2) \times \SU(2)}_{\PP(\CC^2 \otimes \CC^2)}, f \rangle =
    \frac 1 2 \int_0^1 f(t,t) dt,
  \end{equation*}
  i.e., by a one-dimensional Lebesgue measure supported on the diagonal between the origin and $(1,1)$.
\end{prp}
\begin{proof}[Proof using the non-Abelian Heckman algorithm (\autoref{non-abelian heckman algorithm})]
  The four fixed points of the action correspond to the standard basis vectors $e_j \otimes e_k$ ($j,k=1,2$), which are weight vectors of weight $(\pm 1, \pm 1)$ using the conventions fixed above (the vertices of the grey rectangle in \autoref{two qubits figure}).
  Let us choose the direction $\gamma = (-2,-1)$ for renormalization.
  After removal of the positive and negative roots, $(\pm 2,0)$ and $(0,\pm 2)$, only a single renormalized isotropy weight remains at each fixed point (cf.\ \autoref{fixed-point data projective space}).
  Therefore, \autoref{non-abelian heckman} shows that the non-Abelian Duister\-maat--Heck\-man measure is given by the restriction to the positive Weyl chamber of
  \begin{equation*}
      \delta_{(1,1)} \star H_{(-2,-2)}
    - \delta_{(1,-1)} \star H_{(-2,2)}
    + \delta_{(-1,1)} \star H_{(-2,2)}
    - \delta_{(-1,-1)} \star H_{(-2,-2)}.
  \end{equation*}
  Only the first summand contributes to the positive Weyl chamber, and its restriction is given precisely by the formula displayed above (cf.\ \autoref{subtle remark}).
\end{proof}
\begin{proof}[Proof using the single-summand algorithm (\autoref{projective space algorithm})]
  Note that the non-Abelian wall jump formula (\autoref{minimal wall jump non-abelian}) is not directly applicable, since $n-r \not\geq R$. Indeed, as we have seen above, the non-Abelian measure does not have a Lebesgue density, since it is concentrated on the diagonal.

  Therefore, we will follow \autoref{projective space algorithm}, which uses the Abelian wall jump formula, and afterwards takes partial derivatives in direction of the negative roots according to \autoref{main theorem}: The decomposition of $\mathfrak t^*$ into regular chambers is indicated in \autoref{two qubits figure}. We start in the unbounded chamber, where the density is equal to the zero polynomial and cross the horizontal critical wall at the top. Evaluating the Abelian jump formula (\autoref{minimal wall jump abelian}; say, with $\omega_0 = (1,1)$ and $\xi = (0,-1)$), we find that the density polynomial on the upper regular chamber is equal to $\frac 1 8 (1 - \lambda^{(2)})$.
  %
  %
  %

  Next, we cross the diagonal critical wall separating the upper and the right-hand side regular chamber. Using the Abelian jump formula once again, we see that the density polynomial changes by $\frac 1 8 (\lambda^{(2)} - \lambda^{(1)})$.
  %
  %
  %

  Therefore, the Abelian Duister\-maat--Heck\-man measure has the following piecewise polynomial density on the positive Weyl chamber:
  \begin{equation*}
    \frac 1 8 \left( 1 - \max(\lambda^{(1)}, \lambda^{(2)}) \right)
  \end{equation*}
  Taking partial derivatives in the direction of the negative roots, $(-2,0)$ and $(0,-2)$, we arrive at the measure asserted above.
\end{proof}

\begin{cor}
  \label{two qubits marginals}
  The joint distribution $\Prob_{\spec}$ of the maximal eigenvalues of the reduced density matrices of a randomly-chosen pure quantum state of two qubits is given by
  \begin{equation*}
    \langle \Prob_{\spec}, f \rangle = 24 \int_{\frac 1 2}^1 f(s, s) \left( s - \frac 1 2 \right)^2 ds,
  \end{equation*}
  for all test functions $f(\hat\lambda^{(1)}_{\max}, \hat\lambda^{(2)}_{\max})$.
\end{cor}
\begin{proof}
  According to \eqref{eigenvalue distribution dist}, multiply the non-Abelian Duister\-maat--Heck\-man measure by the symplectic volume polynomial $p_{\SU(2) \times \SU(2)}(\lambda) = \lambda^{(1)} \lambda^{(2)}$, divide by $\frac 1 {3!}$, the volume of $\PP(\CC^2 \otimes \CC^2)$. Finally, push forward along $(\lambda^{(j)}) \mapsto (\hat\lambda^{(j)}_{\max} = \tfrac {1+\lambda^{(j)}} 2)$.
\end{proof}

This eigenvalue distribution is in fact known more generally for bipartite pure states chosen at random \cite{lloydpagels88,zyczkowskisommers01}. We will later show how to compute its generalization using the techniques of this paper (\autoref{lloyd pagels}).

\medskip

For higher tensor powers, evaluating the Heckman formula quickly becomes unwieldy. However, it can still be used to compute the Duister\-maat--Heck\-man measure locally:

\begin{prp}
  \label{n qubits local}
  The non-Abelian Duister\-maat--Heck\-man measure for the action of $\SU(2)^N$ on $\PP((\CC^2)^{\otimes N})$ is on the closures of the regular chambers that contain the vertex $(1,\ldots,1)$ given by the convolution product
  \begin{equation*}
    \delta_{(1,\ldots,1)} \star H_{\omega_1} \star \ldots \star H_{\omega_{2^N-N-1}},
  \end{equation*}
  where $\{ \omega_k \}$ is the set of weights of the form $(-2,\ldots,-2,0,\ldots,0)$ (at least two non-zero entries) as well as their $S_N$-permutations.
\end{prp}
\begin{proof}
  If we renormalize with respect to the direction $\gamma \approx (-1,\ldots,-1)$ then just as in the first proof of \autoref{two qubits non-abelian} only a single summand in the non-Abelian Heckman formula contributes in the vicinity of the vertex $(1,\ldots,1)$ and, moreover, this summand is of the above form: Indeed, the weights $\{ \omega_k \}$ are precisely the isotropy weights with the negative roots removed (cf.\ \autoref{fixed-point data projective space}). Since the density function of $\DuHe^K_M$ is polynomial on each regular chamber adjacent to the vertex, we can extend the local formula to their closures.
\end{proof}

It is in fact easy to see that the domain of validity of this formula is the intersection of the half-space
\begin{equation*}
  \left\{ \lambda : \sum_{j=1}^N \lambda^{(j)} \geq N-2 \right\}
\end{equation*}
with the positive Weyl chamber (the regular chambers not adjacent to $(1,\ldots,1)$ lie in the complement of this half-space).

\begin{rem}
  \autoref{n qubits local} gives a local description of the non-Abelian moment polytope, namely by the cone based at $(1,\ldots,1)$ and spanned by the rays with direction vectors $\{\omega_k\}$. By convexity, its intersection with the positive Weyl chamber is an outer approximation to the moment polytope.
\end{rem}

Let us specialize to the case $N=3$: Here, precisely the rays with the direction vectors $(-2,-2,0)$, $(-2,0,-2)$ and $(0,-2,-2)$ are extremal. Their intersection with the positive Weyl chamber has to be contained in the non-Abelian moment polytope: Otherwise, there would be additional vertices in the interior of the positive Weyl chamber --- but only $(1,1,1)$ is the image of a torus fixed point. Since also the origin is contained in the moment polytope (the Greenberger--Horne--Zeilinger state, $[\psi] = [e_1 \otimes e_1 \otimes e_1 + e_2 \otimes e_2 \otimes e_2]$, is a preimage of the origin \cite{greenbergerhornezeilinger89}), we conclude that the convex hull
\begin{equation*}
  \conv \{ (0,0,0), (1,0,0), (0,1,0), (0,0,1), (1,1,1) \}
\end{equation*}
is an inner approximation to the moment polytope. Both approximations are in fact equal and therefore describe the moment polytope precisely (\autoref{three qubits figure}). Inequalities characterizing the moment polytope for $N$ qubits have been determined in \cite{higuchisudberyszulc03}.

\begin{figure}
    \centering
  \tdplotsetmaincoords{70}{10}
  \begin{tikzpicture}[line join=bevel,scale=2.8,tdplot_main_coords]
    \coordinate (O) at (0,0,0);
    \coordinate (W) at (0.3333,0.3333,0.3333);
    \coordinate (W1) at (1,0,0);
    \coordinate (W2) at (0,1,0);
    \coordinate (W3) at (0,0,1);
    \coordinate (SEP) at (1,1,1);

    \draw[->] (-0.3,0,0) -- (1.3,0,0);
    \draw[->] (0,-0.3,0) -- (0,1.3,0);
    \draw[->] (0,0,-0.3) -- (0,0,1.3);
    \draw[-,very thin,color=gray] (1,-0.15,0) -- (1,1.15,0);
    \draw[-,very thin,color=gray] (-0.15,1,0) -- (1.15,1,0);
    \draw[-,very thin,color=gray] (1,1,-0.15) -- (1,1,1.15);

    \draw[-,very thin,color=gray] (-0.15,1,1) -- (1.15,1,1);
    \draw[-,very thin,color=gray] (0,-0.15,1) -- (0,1.15,1);
    \draw[-,very thin,color=gray] (0,1,-0.15) -- (0,1,1.15);

    \draw[-,thick,dashed] (W) -- (W1);
    \draw[-,thick,dashed] (W) -- (W2);
    \draw[-,thick,dashed] (W) -- (W3);
    \draw[-,thick,dashed] (W) -- (O);
    \draw[-,thick,dashed] (W) -- (W1);
    \draw[-,thick,dashed] (W) -- (W2);
    \draw[-,thick,dashed] (W) -- (W3);
    \draw[-,thick,dashed] (W) -- (SEP);
    \draw[-,very thick,dashed] (O) -- (W1) -- (W2) -- cycle;
    \draw[-,very thick,dashed] (SEP) -- (W1) -- (W2) -- cycle;
    \draw[-,very thick,dashed] (SEP) -- (W2) -- (W3) -- cycle;
    \draw[-,very thick,fill opacity=0.7,fill=gray!120] (O) -- (W3) -- (W1) -- cycle;
    \draw[-,very thick,fill opacity=0.7,fill=gray] (SEP) -- (W3) -- (W1) -- cycle;

    \draw (W1) node[below]{\tiny $(1,0,0)$};
    \draw (W2) node[left]{\tiny $(0,1,0)$};
    \draw (W3) node[left]{\tiny $(0,0,1)$};
    \draw (SEP) node[anchor=west]{\tiny $(1,1,1)$};
    \draw (O) node[below left]{\tiny $0$};
  \end{tikzpicture}
  \caption{Non-Abelian moment polytope of three qubits and its decomposition into domains of polynomiality for the non-Abelian Duistermaat--Heckman measure. Each domain is the union of two regular chambers.}
  \label{three qubits figure}
\end{figure}
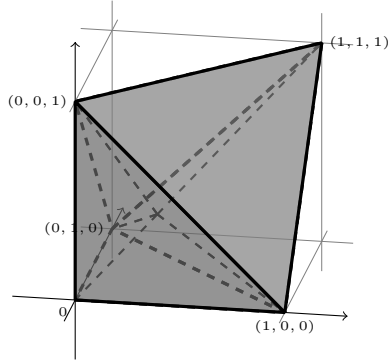

\begin{prp}
  \label{three qubits non-abelian density}
  The non-Abelian Duister\-maat--Heck\-man measure for the action of $\SU(2)^3$ on $\PP(\CC^2 \otimes \CC^2 \otimes \CC^2)$ has the piecewise linear Lebesgue density
  \begin{equation*}
    \begin{cases}
      \frac{1}{16} \min \lambda^{(j)} & \text{in the lower pyramid},\\
      \frac{1}{32} \left( 1 - \sum_{j=1}^3 \lambda^{(j)} + 2 \min \lambda^{(j)} \right) & \text{in the upper pyramid},\\
      0 & \text{otherwise}
    \end{cases}
  \end{equation*}
  (compare \autoref{three qubits figure}).
\end{prp}
\begin{proof}
  By \autoref{n qubits local}, the non-Abelian Duister\-maat--Heck\-man measure is on the closures of the regular chambers containing $(1,1,1)$ given by the convolution
  \begin{equation*}
    \delta_{(1,1,1)} \star H_{(-2,-2,-2)} \star H_{(-2,-2,0)} \star H_{(-2,0,-2)} \star H_{(0,-2,-2)}.
  \end{equation*}
  Using \eqref{convolution translation} we can readily compute its density:
  \begin{align*}
     &\int_0^\infty dt_1 \cdots \int_0^\infty dt_4 \, \delta\big(\matrix{1\\1\\1} + t_1 \matrix{-2\\-2\\-2} + t_2 \matrix{-2\\-2\\0} + t_3 \matrix{-2\\0\\-2} + t_4 \matrix{0\\-2\\-2} - \lambda\big)\\
     = &\frac{1}{32} \int_0^\infty ds_1 \int_{-\infty}^\infty ds_2 \cdots \int_{-\infty}^\infty ds_4 \, \Id_C(s_2,s_3,s_4) \, \delta\big((1-s_1) \matrix{1\\1\\1} - \matrix{s_2\\s_3\\s_4} - \lambda\big)\\
     = &\frac{1}{32} \int_0^\infty ds_1 \, \Id_C\big((1-s_1) \matrix{1\\1\\1} - \lambda\big)\\
     = &\frac{1}{32} \max \{ 0, 1 - \sum_{j=1}^3 \lambda^{(j)} + 2 \min(\lambda^{(j)}) \},
  \end{align*}
  where $\Id_C$ is the indicator function of the cone spanned by $\matrix{1\\1\\0}$, $\matrix{1\\0\\1}$, and $\matrix{0\\1\\1}$,
  i.e.,
  \begin{equation*}
    \Id_C(a,b,c) =
    \begin{cases}
      1 &\text{ if } a+b \geq c \text{ and } a+c \geq b \text{ and } b+c \geq a,\\
      0 &\text{ otherwise}.
    \end{cases}
  \end{equation*}
  We have therefore established the claimed density on the complement of the lower pyramid.

  According to \autoref{minimal wall jump non-abelian}, the jump across the hyperplane separating the upper and the lower pyramid is given by
  \begin{equation*}
    \frac 1 4 \left(- \frac 1 {64}\right) 8 \left(1 - \sum_{j=1}^3 \lambda^{(j)} \right) = \frac 1 {32} \left(\sum_{j=1}^3 \lambda^{(j)} - 1 \right)
    %
  \end{equation*}
  (the left-hand side terms are ordered just like in the jump formula). This is precisely the difference between the densities on the upper and lower pyramid as asserted in the statement of the proposition.
\end{proof}

It is straightforward to deduce from this the eigenvalue distribution (cf.\ the proof of \autoref{two qubits marginals}):

\begin{cor}
  \label{three qubits marginals}
  The joint distribution of the maximal eigenvalues of the reduced density matrices of a randomly-chosen pure quantum state of three qubits has Lebesgue density
  \begin{equation*}
    8! \left( \prod_{j=1}^3 \hat\lambda^{(j)}_{\max} - \frac 1 2 \right)
    \begin{cases}
      \min \hat\lambda^{(j)}_{\max} - \frac 1 2 & \text{if } \sum_{j=1}^3 \hat\lambda^{(j)}_{\max} \leq 2,\\
      \max \left\{ 0, \frac 1 2 \left( 1 - \sum_{j=1}^3 \hat\lambda^{(j)}_{\max} \right) + \min \hat\lambda^{(j)}_{\max} \right\} & \text{if } \sum_{j=1}^3 \hat\lambda^{(j)}_{\max} \geq 2,
    \end{cases}
  \end{equation*}
  on the space of maximal eigenvalues $(\hat\lambda^{(j)}_{\max}) \in [\frac 1 2,1]^3$.
\end{cor}


\begin{rem}
  Our use of the local convolution formula (\autoref{n qubits local}) and of the non-Abelian wall jump formula (\autoref{minimal wall jump non-abelian}) were merely convenient shortcuts: It is clear that we could have completely algorithmically computed the measure by following \autoref{projective space algorithm}.
\end{rem}

\subsection{Mixed States of Two Qubits}
\label{bravyi example}

We will now use the non-Abelian Heckman algorithm to treat the case of random two-qubit states with fixed, non-degenerate global eigenvalue spectrum. That is, we consider the action of $K = \SU(2) \times \SU(2)$ on a coadjoint $\SU(4)$-orbit through a point $\tilde\lambda$ contained in the interior of the positive Weyl chamber.

Recall that the Weyl group of $\SU(4)$ is the symmetric group $S_4$, with $(-1)^{l(\tilde w)}$ equal to the signum of a permutation $\tilde w \in S_4$. By \eqref{heckman for maximal dimensional coadjoint orbit reductions},
\begin{equation}
  \label{bravyi abelian}
  \DuHe^{T}_{\mathcal O_{\tilde\lambda}} = \sum_{\tilde w \in S_4} \sign(\tilde w) \, \delta_{\pi(\tilde w \tilde\lambda)} \star H_{-\pi(\tilde\alpha_1)} \star \ldots \star H_{-\pi(\tilde\alpha_6)},
\end{equation}
where $\tilde\alpha_1, \ldots, \tilde\alpha_6$ are the positive roots of $\SU(4)$ (see \autoref{notation} for our conventions), and where $\pi$ is the restriction map $\mathfrak {\tilde t}^* \rightarrow \mathfrak t^*$, with $\mathfrak {\tilde t}^*$ the dual of the Lie algebra of the maximal torus of $\SU(4)$. With respect to our identification $\mathfrak t^* \cong \RR^2$ fixed in \autoref{pure states qubits}, the map $\pi$ is given by
\begin{equation}
  \label{bravyi projection}
  \pi \colon
  \mathfrak {\tilde t}^* \rightarrow \RR^2, \quad
  (\tilde\lambda_1,\ldots,\tilde\lambda_4) \mapsto
  2i(\tilde\lambda_1+\tilde\lambda_2, \tilde\lambda_1+\tilde\lambda_3).
\end{equation}
One computes readily that the $-\pi(\tilde\alpha_k)$ are precisely the weights $(-2,2)$, $(-2,0)$ (twice), $(-2,-2)$ and $(0,-2)$ (twice). In particular, the two negative roots of $\SU(2) \times \SU(2)$ are contained in this list (each of them is in fact contained twice). By applying \autoref{main theorem} we arrive at the following formula:

\begin{prp}
  \label{bravyi non-abelian}
  The non-Abelian Duister\-maat--Heck\-man measure for the action of $\SU(2) \times \SU(2)$ on a coadjoint $\SU(4)$-orbit $\mathcal O_{\tilde\lambda}$ with $\tilde\lambda \in \mathfrak {\tilde t}^*_{>0}$ is given by
  \begin{equation*}
    \DuHe^{\SU(2) \times \SU(2)}_{\mathcal O_{\tilde\lambda}} =
    \restrict{
    \left( \sum_{\tilde w \in S_4} \sign(\tilde w) \, \delta_{\pi(\tilde w \tilde\lambda)} \right) \star H_{(-2,2)} \star H_{(-2,0)} \star H_{(-2,-2)} \star H_{(0,-2)}
    }{\mathfrak t^*_+}.
  \end{equation*}
\end{prp}

Following \autoref{non-abelian heckman algorithm}, we evaluate the right-hand side iterated convolution using \autoref{boysal vergne algorithm}. The result is the following:

\begin{lem}
  \label{bravyi convolution}
  The measure $H_{(-2,2)} \star H_{(-2,0)} \star H_{(-2,-2)} \star H_{(0,-2)}$ has Lebesgue density
  \begin{equation*}
    f(\lambda^{(1)}, \lambda^{(2)}) =
    \begin{cases}
      0 & \text{in chamber 0},\\
      \frac 1 {64} \left( \lambda^{(1)} + \lambda^{(2)} \right)^2 & \text{in chamber 1},\\
      \frac 1 {64} \left( \left( \lambda^{(1)} \right)^2 + 2 \lambda^{(1)} \lambda^{(2)} - \left( \lambda^{(2)} \right)^2 \right) & \text{in chamber 2},\\
      \frac 1 {32} \left( \lambda^{(1)} \right)^2 & \text{in chamber 3}.
    \end{cases}
  \end{equation*}
  See \autoref{bravyi figure} for the labelling of the chambers and an illustration of the density.
\end{lem}

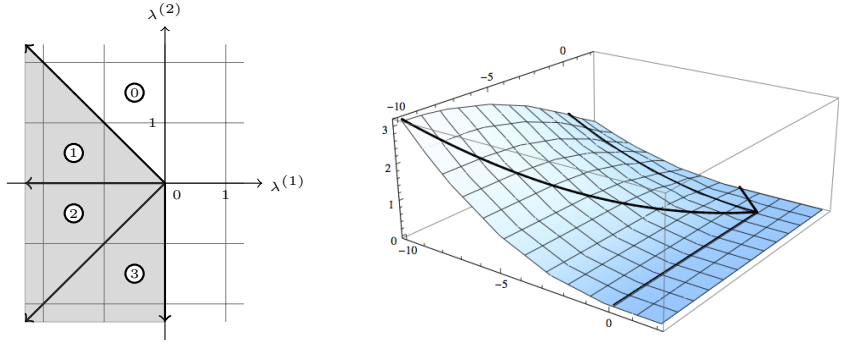
\begin{figure}
    \centering
  \begin{tikzpicture}[scale=0.8]
    \draw[help lines] (-2.3,-2.3) grid (1.3,2.3);
    \draw[->] (-2.6,0) -- (1.6,0) node[right] {\tiny $\lambda^{(1)}$};
    \draw[->] (0,-2.6) -- (0,2.6) node[above] {\tiny $\lambda^{(2)}$};
    \draw(0,0) node[below right] {\tiny $0$};
    \draw(1,0) node[below] {\tiny $1$};
    \draw(0,1) node[left] {\tiny $1$};
    \draw[thick,->] (0,0) -- (-2.3,2.3);
    \draw[thick,->] (0,0) -- (-2.3,0);
    \draw[thick,->] (0,0) -- (-2.3,-2.3);
    \draw[thick,->] (0,0) -- (0,-2.3);
    \draw[fill=gray,opacity=0.3] (0,0) -- (-2.3,2.3) -- (-2.3,-2.3) -- (0,-2.3) -- (0,0);
    \draw[fill=white,thick] (-0.5,1.5) circle (0.15) node {\tiny $0$};
    \draw[fill=white,thick] (-1.5,0.5) circle (0.15) node {\tiny $1$};
    \draw[fill=white,thick] (-1.5,-0.5) circle (0.15) node {\tiny $2$};
    \draw[fill=white,thick] (-0.5,-1.5) circle (0.15) node {\tiny $3$};
    \draw(0,-3) node {~};
  \end{tikzpicture}
  \quad\quad\quad
  \includegraphics[width=6cm,trim=0cm -0.8cm 0cm 0cm]{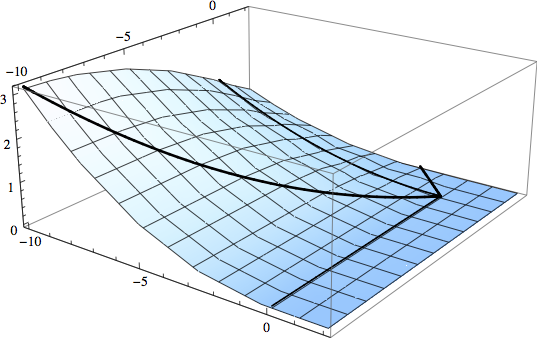}
  \caption{
    (a) Chambers and support (gray) and (b) density function of the iterated convolution computed in \autoref{bravyi convolution}.
  }
  \label{bravyi figure}
\end{figure}

The density of the non-Abelian Duister\-maat--Heck\-man measure is thus given by the restriction to the positive Weyl chamber of an alternating sum of $24$ copies of the density described in \autoref{bravyi convolution}, one copy attached to each of the points $\pi(\tilde w \tilde\lambda)$. In view of the geometry of the support of the latter density, it is clear that in fact only summands for points in the right halfplane $\{\lambda^{(1)} > 0\}$ contribute (i.e., at most half of the points). Using \eqref{bravyi projection}, one finds that the points $\pi(\tilde w \tilde\lambda)$ are the six points whose coordinates are equal to any two out of the three values $c_1 = |2i(\tilde\lambda_1+\tilde\lambda_4)|$, $c_2 = 2i(\tilde\lambda_1+\tilde\lambda_3)$, or $c_3 = 2i(\tilde\lambda_1+\tilde\lambda_2)$ (without repetitions), as well as their Weyl conjugates. See \autoref{bravyi figure two} for illustration.

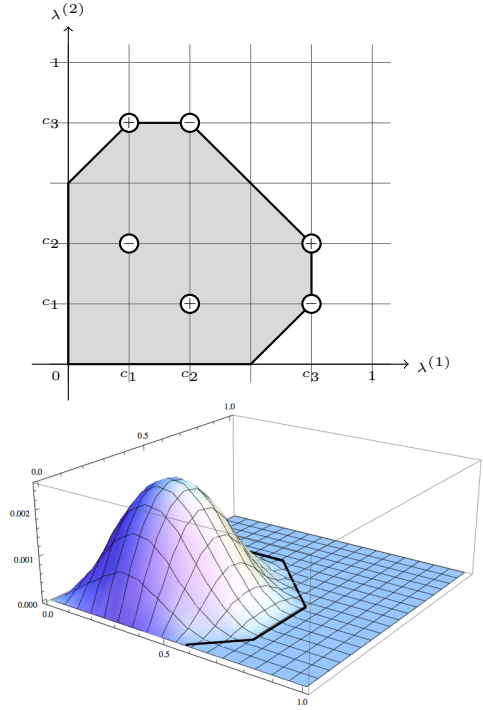
\begin{figure}
    \centering
  \begin{tikzpicture}[scale=0.8]
    \draw[help lines] (-0.3,-0.3) grid (5.3,5.3);
    \draw[->] (-0.6,0) -- (5.6,0) node[right] {\tiny $\lambda^{(1)}$};
    \draw[->] (0,-0.6) -- (0,5.6) node[above] {\tiny $\lambda^{(2)}$};
    \draw(0,0) node[below left] {\tiny $0$};
    \draw(5,0) node[below] {\tiny $1$};
    \draw(0,5) node[left] {\tiny $1$};
    \draw(1,0) node[below] {\tiny $c_1$};
    \draw(2,0) node[below] {\tiny $c_2$};
    \draw(4,0) node[below] {\tiny $c_3$};
    \draw(0,1) node[left] {\tiny $c_1$};
    \draw(0,2) node[left] {\tiny $c_2$};
    \draw(0,4) node[left] {\tiny $c_3$};
    \draw[fill=gray,opacity=0.3] (4,1) -- (3,0) -- (0,0) -- (0,3) -- (1,4) -- (2,4) -- (4,2) -- (4,1);
    \draw[thick,-] (4,1) -- (3,0) -- (0,0) -- (0,3) -- (1,4) -- (2,4) -- (4,2) -- (4,1);


    \draw[fill=white,thick] (4,2) circle (0.15) node {\tiny $+$};
    \draw[fill=white,thick] (4,1) circle (0.15) node {\tiny $-$};
    \draw[fill=white,thick] (2,1) circle (0.15) node {\tiny $+$};
    \draw[fill=white,thick] (2,4) circle (0.15) node {\tiny $-$};
    \draw[fill=white,thick] (1,4) circle (0.15) node {\tiny $+$};
    \draw[fill=white,thick] (1,2) circle (0.15) node {\tiny $-$};
  \end{tikzpicture}
  \includegraphics[width=6.8cm,trim=0cm -0.8cm 0cm 0cm]{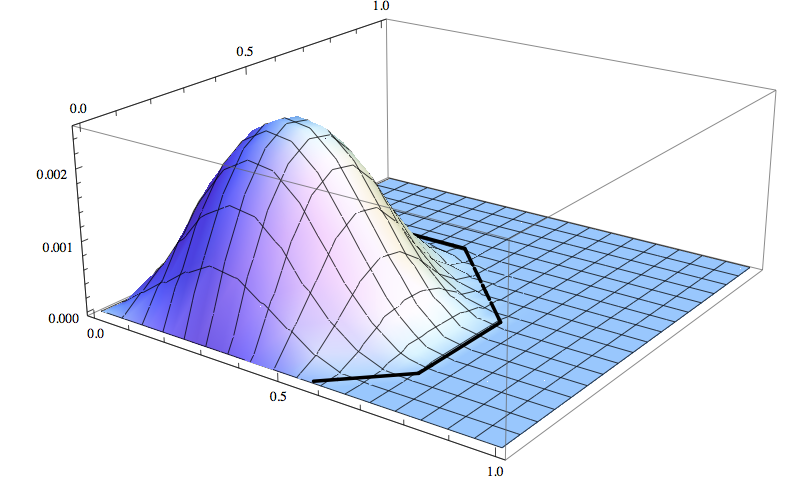}
  \caption{(a) Non-Abelian moment polytope (gray), 
    and the points $\pi(\tilde w \tilde\lambda) \in \mathfrak t^*_+$ together with $\sign(\tilde w)$.
    (b) Density of the Duister\-maat--Heck\-man measure for global spectrum $(4/7,2/7,1/7,0)$.}
  \label{bravyi figure two}
\end{figure}

Moreover, we can deduce that the non-Abelian moment polytope has the form described in the figure. To do so, we simply need to check in each regular chamber whether the density polynomial vanishes. By doing so and describing the resulting polytope in terms of inequalities, we recover a well-known result by Bravyi \cite{bravyi04}:

\begin{cor}
  The non-Abelian moment polytope for the action of $\SU(2) \times \SU(2)$ on a generic coadjoint $\SU(4)$-orbit is given by
  \begin{equation*}
    \Delta_{\SU(2) \times \SU(2)}(\mathcal O_{\tilde\lambda}) =
    \{
      (\lambda^{(1)},\lambda^{(2)}) :
      0 \leq \lambda^{(1)}, \lambda^{(2)} \leq c_3,
      \lambda^{(1)} + \lambda^{(2)} \leq c_2 + c_3,
      \abs{\lambda^{(1)} - \lambda^{(2)}} \leq c_3 - c_1
    \}.
  \end{equation*}
\end{cor}


In the limit where the global state becomes pure, the moment polytope converges to the diagonal between the origin and $(1,1)$. This is in agreement with \autoref{two qubits non-abelian}. One can similarly recover the eigenvalue distribution of the reduced density matrices of a random pure state of two qubits by taking a corresponding limit.

In view of \autoref{purification measure}, the distributions computed in \autoref{bravyi non-abelian} can be assembled to give the joint eigenvalue distribution of the reduced density matrices of a randomly-chosen pure state in $\PP(\CC^2 \otimes \CC^2 \otimes \CC^4)$.

\subsection{Pure States of Bosonic Qubits}
\label{pure states bosonic qubits}

We now turn to random pure states of $N$ bosonic qubits, where $K = \SU(2)$ and $M = \PP(\Sym^N(\CC^2))$. We will use the Abelian Heckman algorithm:

\begin{prp}
  \label{sym N qubits abelian}
  The Abelian Duister\-maat--Heck\-man measure for $\PP(\Sym^N(\CC^2))$ has Lebesgue density
  \begin{equation*}
    \frac 1 {2^N (N-1)! N!} \sum_{k=-N,-N+2,\ldots,N} (-1)^{\frac{N+k}{2}} {\binom{N}{\frac{N+k}{2}}} (\lambda - k)_+^{N-1}.
  \end{equation*}
  Here, we set $(\lambda-k)^{N-1}_+ = (\lambda-k)^{N-1}$ for $\lambda \geq k$ and $0$ otherwise.

  Equivalently, $N!$ times the Abelian Duister\-maat--Heck\-man measure is equal to the probability distribution of the sum of $N$ independent random variables that are uniformly distributed on the interval $[-1,1]$ \cite[\S I.9, Theorem 1a]{feller71}.
\end{prp}
\begin{proof}
  The weights of $\Sym^N(\CC^2)$ are $\{ -N, -N+2, \ldots, N \}$; let us write $v_k$ for a weight vector of weight $k$.
  The associated projective space has precisely $N+1$ torus fixed points. At any such fixed point $[v_k]$, the isotropy weights are given by
  \[
    \{ (l-k) : l = -N, -N+2, \ldots, \check{k}, \ldots, N \},
  \]
  and we will denote them by $\hat\omega_{k,1,}, \ldots, \hat\omega_{k,N}$ (cf.\ \autoref{fixed-point data projective space}).
  Observe that precisely $n_k = \frac{N+k}{2}$ of them are negative with respect to the renormalization direction $\gamma = +1$.
  By \eqref{convolution translation}, the corresponding summand of the Heckman formula is equal to the push-forward of Lebesgue measure on $\RR^N_{\geq 0}$ along
  \begin{equation*}
    P \colon \RR^N_{\geq 0} \rightarrow \mathfrak u_1^*, \quad
    (s_1, \ldots, s_N) \mapsto \sum_{i=1}^N s_i \abs{\hat\omega_{k,i}} + k.
  \end{equation*}
  We first compute its cumulative distribution function:
  \begin{align*}
    P_*(ds)\left((-\infty,(k+\lambda)]\right)
    &= ds(\{ s_1, \ldots, s_N \geq 0 : \sum_{i=1}^N s_i \abs{\hat\omega_{k,i}} \leq \lambda \})\\
    &= ds(\{ s_1, \ldots, s_N \geq 0 : \sum_{i=1}^N s_i \leq 1 \}) \, \frac 1 {\prod_i \abs{\hat\omega_{k,i}}} \, \lambda_+^N\\
    &= \frac 1 {N!} \frac 1 {2^N (\frac{N+k}{2})! (\frac{N-k}{2})!} \lambda_+^N\\
    &= \frac 1 {2^N N! N!} {\binom{N}{\frac{N+k}{2}}} \lambda_+^N.
  \end{align*}
  The density is then given by the derivative,
  \begin{equation*}
    f_k(\lambda) = \frac 1 {2^N (N-1)! N!} {\binom{N}{\frac{N+k}{2}}} (\lambda-k)_+^{N-1},
  \end{equation*}
  and by forming the alternating sum of these terms we arrive at the formula displayed above.
\end{proof}

In \autoref{plethysms} we give an alternative proof of \autoref{sym N qubits abelian} using representation theory and combinatorics.
It follows directly from the derivative principle that the non-Abelian Duister\-maat--Heck\-man measure is given by the following formula:

\begin{cor}
  \label{sym N qubits non-abelian}
  The non-Abelian Duister\-maat--Heck\-man measure for the action of $\SU(2)$ on $\PP(\Sym^N(\CC^2))$ with $N \geq 2$ has Lebesgue density
  \begin{equation*}
    \frac 1 {2^{N-1} (N-2)! N!} \sum_{k=-N,-N+2,\ldots,N} (-1)^{\frac{N+k}{2}+1} {\binom{N}{\frac{N+k}{2}}} (\lambda-k)_+^{N-2}
  \end{equation*}
  on $[0,\infty)$.
\end{cor}

Again, it is clear how to translate the above into the eigenvalue distribution of the one-body reduced density matrix by using \eqref{eigenvalue distribution indist}. See \autoref{sym two qubits figure} for an illustration in the case of $N=2$ bosonic qubits.

\begin{figure}
  \quad\quad\quad
  \begin{align*}
    \begin{tikzpicture}
      \draw[help lines] (-2.3,-0.3) grid (2.3,1.3);
      \draw[->] (-2.6,0) -- (2.6,0) node[right] {\tiny $\lambda$};
      \draw[->] (0,-0.6) -- (0,1.6);
      \draw(0,0) node[below left] {\tiny $0$};
      \draw(1,0) node[below left] {\tiny $1$};
      \draw(2,0) node[below left] {\tiny $2$};
      \draw(-1.6,0.2) node[above] {\tiny $\frac {2+\lambda} 8$};
      \draw(1.6,0.2) node[above] {\tiny $\frac {2-\lambda} 8$};
      \draw[very thick,dashed] (-2,0) -- (0,1);
      \draw[very thick] (2,0) -- (0,1);
    \end{tikzpicture}
    & \quad \quad
    \begin{tikzpicture}
      \draw[help lines] (-0.3,-0.3) grid (2.3,1.3);
      \draw[->] (-0.6,0) -- (2.6,0) node[right] {\tiny $\lambda$};
      \draw[->] (0,-0.6) -- (0,1.6);
      \draw(0,0) node[below left] {\tiny $0$};
      \draw(1,0) node[below left] {\tiny $1$};
      \draw(2,0) node[below left] {\tiny $2$};
      \draw(2,1) node[right] {\tiny $\frac 1 4$};
      \draw[very thick] (0,1) -- (2,1);
    \end{tikzpicture}
    & \quad \quad
    \begin{tikzpicture}
      \draw[help lines] (-0.3,-0.3) grid (1.3,1.3);
      \draw[->] (-0.6,0) -- (1.6,0) node[right] {\tiny $\hat\lambda_{\max}$};
      \draw[->] (0,-0.6) -- (0,1.6);
      \draw(0,0) node[below left] {\tiny $0$};
      \draw(1,0) node[below left] {\tiny $1$};
      \draw[very thick] (0,0) -- (0.5,0) -- (1,1);
      \draw(1,0.7) node[right] {\tiny $8 (\hat\lambda_{\max} - \tfrac 1 2)$};
    \end{tikzpicture}
  \end{align*}
  \caption{(a) and (b) Density of the Abelian and the non-Abelian Duister\-maat--Heck\-man measure for $\PP(\Sym^2(\CC^2))$, and (c) corresponding maximal eigenvalue distribution of the one-body reduced density matrix.}
  \label{sym two qubits figure}
\end{figure}
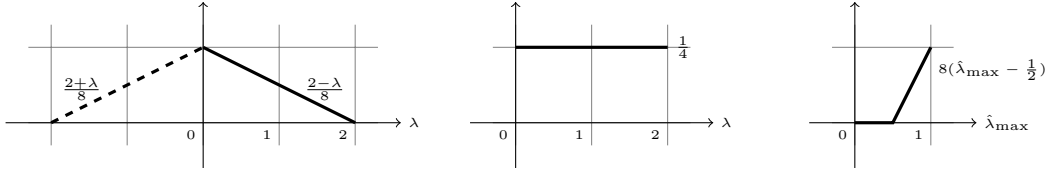

As an application, let us compute the average value of the reduced purity of a randomly-chosen pure state of bosonic qubits. The reduced purity is by definition equal to
\begin{equation}
  \label{definition reduced purity}
  \norm{\rho_1}_2^2 = \left( \hat\lambda_{\max}^2 + (1-\hat\lambda_{\max})^2 \right),
\end{equation}
where $\rho_1$ denotes the one-body reduced density matrix and $\hat\lambda_{\max}$ its maximal eigenvalue.

\begin{cor}
  \label{purity example}
  The average reduced purity \eqref{definition reduced purity} of a randomly-chosen pure state of $N$ bosonic qubits is given by
  \begin{equation*}
    \frac 1 2 + \frac 1 {2 N}.
  \end{equation*}
\end{cor}
\begin{proof}
  We will not use \autoref{sym N qubits non-abelian} directly, but instead work with the Abelian Duister\-maat--Heck\-man measure:
  Denote by $\Prob$ the probability distribution of the maximal eigenvalue of the one-body reduced density matrix.
  By using \eqref{eigenvalue distribution indist} and $\hat\lambda_{\max} = \frac 1 2 + \frac \lambda {2 N}$, we find that the average reduced purity is given by
  \begin{align*}
       &\int \left( \hat\lambda_{\max}^2 + \left( 1-\hat\lambda_{\max} \right)^2 \right) \; d\Prob(\hat\lambda_{\max})\\
    = &\frac 1 {4 N^2} \int \lambda \left( \left( N + \lambda \right)^2 + \left( N - \lambda \right)^2 \right) \, N! \, d\DuHe^{\SU(2)}_M(\lambda)\\
    = &\frac 1 {2 N^2} \int \left( N^2 \lambda + \lambda^3 \right) \, N! \, d\DuHe^{\SU(2)}_M(\lambda).
  \end{align*}
  By the derivative principle, \autoref{main theorem},
  this is equal to
  \begin{align*}
    &\frac 1 {N^2} \int_0^{\infty} \left( N^2 + 3 \lambda^2 \right) \, N! \, d\DuHe^{\U(1)}_M(\lambda)\\
    = &\frac 1 {2 N^2} \int \left( N^2 + 3 \lambda^2 \right) \, N! \, d\DuHe^{\U(1)}_M(\lambda).
  \end{align*}
  In \autoref{sym N qubits abelian} we have seen that $N! \, \DuHe^{\U(1)}_{M}$ is the probability distribution of the sum of $N$ independent random variables that are uniformly distributed on the interval $[-1,1]$. Since the variance of any such random variable is $\frac 1 3$ and since variances of independent random variables are additive, the above is equal to
  \begin{equation*}
      \frac 1 {2N^2} \left( N^2 + 3 \frac N 3 \right)
    = \frac 1 2 + \frac 1 {2N}.
    \qedhere
  \end{equation*}
\end{proof}

In accordance with the concentration of measure phenomenon, $\rho_1 \rightarrow \Id/2$ in distribution as $N \rightarrow \infty$. We remark that our result matches \cite[Theorem 34]{mullerdahlstenvedral11} if one works out the quantities left uncalculated therein.
Note that our proof illustrates the power of the derivative principle: Instead of explicitly computing the eigenvalue distribution, we can reduce to the Abelian Duister\-maat--Heck\-man measure by differentiating the quantity we are interested in.

\subsection{Pure States of Bipartite Systems}
We conclude this series of examples by re-deriving the eigenvalue distribution of the reduced density matrices of a randomly-chosen pure state in the case of a general bipartite quantum system, corresponding to the action of $\SU(a) \times \SU(b)$ on $M = \PP(\CC^{a} \otimes \CC^{b})$. Instead of following one of the algorithms it will be most convenient to directly work with the formula given in \autoref{density in polytope picture}.

Suppose that $b \geq a$. If $b > a+1$ then \autoref{main assumption} is not satisfied (\autoref{main assumption qmp}): Indeed, it always follows from the singular value decomposition that
\begin{equation}
  \label{bipartite spectra}
  \spec \rho_B = (\spec \rho_A, 0, \ldots, 0),
\end{equation}
so that in this case the non-Abelian moment polytope is contained in the boundary of the positive Weyl chamber. However, \eqref{bipartite spectra} of course implies that for any choice of $b \geq a$ the joint eigenvalue distribution is already determined by the eigenvalue distribution of $\rho_A$. That is, it suffices to compute the Duister\-maat--Heck\-man measure for the action of $K = \SU(a)$. Denote by $T$ the standard maximal torus of $\SU(a)$. As in \eqref{density matrices to functionals}, we identify points $\lambda \in \mathfrak t \cong \mathfrak t^*$ with diagonal density matrices $\hat\lambda = \frac 1 a + i \lambda$.

Clearly, the Abelian moment polytope consists of those $\lambda$ with $\hat\lambda \in \Delta_{a-1}$, and the non-Abelian moment polytope is its intersection with the positive Weyl chamber (cf.\ \autoref{purification}).

\begin{lem}
  \label{lloyd pagels abelian}
  On the Abelian moment polytope, the Duister\-maat--Heck\-man measure for the action of the maximal torus of $\SU(a)$ is proportional to
  \begin{equation*}
    \prod_{j=1}^a \hat\lambda_j^{b-1} ~ d\lambda =
    \prod_{j=1}^a \left(\frac 1 a + i \lambda_j\right)^{b-1} ~ d\lambda.
  \end{equation*}
\end{lem}
\begin{proof}
  Choose the weight-space decomposition of $\CC^a \otimes \CC^b$ given by the standard basis vectors $e_j \otimes e_k$ ($j=1,\ldots,a$ and $k=1,\ldots,b$). According to \autoref{density in polytope picture}, the density of the Abelian Duister\-maat--Heck\-man measure is given by
  \begin{equation*}
    f(\lambda) = \vol \, \{ p_{1,1}, \ldots, p_{a,b} \geq 0 : \sum_{k=1}^b p_{j,k} = \hat\lambda_j = \frac 1 a + i \lambda_j \quad (j=1,\ldots,a) \}
  \end{equation*}
  with respect to the volume measure $dp/d\lambda$ defined therein. Note that the right-hand side set is the Cartesian product of $a$ rescaled standard simplices. The measure factorizes accordingly, and it is easy to see that
  \begin{align*}
    f(\lambda)
    = \prod_{j=1}^a \vol \, \{ p_1, \ldots, p_b \geq 0 : \sum_{k=1}^b p_k = \hat\lambda_j \}
    = \frac 1 Z \prod_{j=1}^a \hat\lambda_j^{b-1}
  \end{align*}
  for $\hat\lambda \in \Delta_{a-1}$, and zero otherwise, with $Z$ a suitable normalization constant.
\end{proof}

\begin{cor}
  \label{lloyd pagels}
  On the non-Abelian moment polytope, the Duister\-maat--Heck\-man measure for the $\SU(a)$-action on $\PP(\CC^a \otimes \CC^b)$ is proportional to
  \begin{equation*}
    \prod_{j=1}^a \hat\lambda_j^{b-a}
    \prod_{j < k \leq a} (\hat\lambda_j - \hat\lambda_k) ~
    d\lambda.
  \end{equation*}
\end{cor}
\begin{proof}
  According to the derivative principle, we have to apply $\prod_{j < k} i(\partial_{\lambda_k} - \partial_{\lambda_j}) = \prod_{j < k} ( \partial_{\hat\lambda_k} - \partial_{\hat\lambda_j} )$
  to the Abelian Duister\-maat--Heck\-man density as computed in \autoref{lloyd pagels abelian}.

  This is a partial differential operator of order $\binom a 2$, therefore the resulting non-Abelian density polynomial has total degree at most $d_{\max} = a(b-1) - a(a-1)/2$. Since we differentiate each variable at most $a-1$ times, it is a multiple of the symmetric polynomial $\prod_{j=1}^a \hat\lambda_j^{b-a}$.
  On the other hand, the result is evidently antisymmetric, and therefore a multiple of the Vandermonde determinant $\prod_{j < k} (\hat\lambda_j - \hat\lambda_k)$. Since the total degrees add up to $d_{\max}$, this implies the assertion.
\end{proof}

In view of \eqref{eigenvalue distribution dist}, this result implies the following well-known formula \cite{lloydpagels88,zyczkowskisommers01}:

\begin{cor}
  The distribution of the eigenvalue spectrum $\hat\lambda = \spec \rho_1$ of a randomly-chosen bipartite pure state $\rho$ on $\CC^a \otimes \CC^b)$ has Lebesgue density proportional to
  \begin{equation*}
    \prod_{j=1}^a \hat\lambda_j^{b-a}
    \prod_{j < k \leq a} (\hat\lambda_j - \hat\lambda_k)^2
  \end{equation*}
  on the space of eigenvalue spectra $\{ \hat\lambda \in \Delta_{a-1} : \hat\lambda_1 \geq \ldots \geq \hat\lambda_a \}$.
\end{cor}

It is also easy to deduce the corresponding formula for the action of $\SU(a) \times \SU(b)$:

\begin{cor}
  \label{lloyd pagels diagonal}
  Denote by $\Delta = \{ \lambda \in \mathfrak t^*_+ : \hat\lambda \in \Delta_{a-1} \}$ the non-Abelian moment polytope for the $\SU(a)$-action. Then the push-forward of Liouville measure along the moment map for the $\SU(a) \times \SU(b)$-action is given by
  \begin{equation*}
      \left\langle (\Phi_{\SU(a) \times \SU(b)})_*(\mu_{\PP(\CC^a \otimes \CC^b)}), g \right\rangle
    = \frac 1 Z \int_{\Delta} d\lambda \int_{\mathcal O^{\SU(a)}_\lambda \times \mathcal O^{\SU(b)}_{((\hat\lambda,0,\ldots,0) - \frac {\Id} b)/i}} g,
  \end{equation*}
  where $Z$ is a suitable normalization constant.
\end{cor}
\begin{proof}
  By \eqref{bipartite spectra}, each coadjoint orbit $\mathcal O_\lambda$ for $\SU(a)$ is paired with the coadjoint $\SU(b)$-orbit through $((\hat\lambda,0,\ldots,0) - \frac{\Id}{b})/i$. By its $\SU(a) \times \SU(b)$-invariance, on each such pair of coadjoint orbits the push-forward measure is just a multiple of the usual Liouville measure. The assertion follows by observing that the density in \autoref{lloyd pagels} is at any point $\lambda$ proportional to the symplectic volume of the corresponding coadjoint $\SU(b)$-orbit.
\end{proof}


\section{Multiplicities of Representations}
\label{multiplicities section}

All results discussed so far can be considered as asymptotic limits of corresponding statements in representation theory, at least if the Hamiltonian $K$-manifold $M$ can be linearized (``quantized'') in a certain technical sense.
This is in the spirit of Kirillov's orbit method and the theory of geometric quantization \cite{guilleminsternberg77,guilleminsternberg84,woodhouse92,guilleminlermansternberg96,kirillov99}.

In particular, this is the case when $M$ is a $K$-invariant smooth irreducible complex projective subvariety of $\PP(V)$ for a finite-dimensional unitary $K$-representation $V$. In this situation, the Fubini--Study form of $\PP(V)$ restricts to a non-degenerate symplectic form on $M$, and the $K$-action is Hamiltonian with moment map the restriction of \eqref{projective space non-abelian moment map}. We still assume that \autoref{main assumption} is satisfied.

Coadjoint orbits $\mathcal O_\lambda$ through dominant integral weights $\lambda \in \Lambda^*$ (and only these) can be realized in this setup \cite{kirillov99}: They are in a natural way projective subvarieties of $\PP(V_\lambda)$, where $V_\lambda$ is the unitary $K$-representation with highest weight $\lambda$.

In particular, the quantum marginal problem can be analyzed in this framework: Coadjoint $\SU(d)$-orbits through integral highest weights correspond to Hermitian matrices with integral eigenvalue spectra, and it suffices to consider these, since we can always rescale and take limits, or simply pass to the purified double (\autoref{purification}).

In \autoref{semiclassicallimit subsection} we will recall the limit alluded to above. We then proceed to describe the representation-theoretic analogue of the derivative principle: Multiplicities of irreducible $K$-representations can be computed from weight multiplicities by taking finite differences (\autoref{steinberg lemma}). In the case of the projective space associated with a unitary $K$-representation, the relevant weight multiplicities are those for the symmetric powers of the representation. In \autoref{multiplicities for projective space}, we give a concrete formula describing these weight multiplicities as the number of integer points in certain rational convex polytopes; we indicate that this is again amenable to algorithmic implementation. Finally, we show that in the limit we recover the corresponding statements of \autoref{projective space}.

\subsection{The Semi-Classical Limit}
\label{semiclassicallimit subsection}

Since the $K$-action on $M \subseteq P(V)$ originates from a linear action on $V$, each graded part of the homogeneous coordinate ring $\CC[M]$ is naturally a finite-dimensional $K$-representation and can thus be decomposed into irreducible sub-representations,
\begin{equation*}
  \CC[M]
  \cong \bigoplus_{k=0}^\infty \bigoplus_{\lambda \in \Lambda^* \cap \mathfrak t^*_+} V_\lambda \otimes \Hom_K(V_\lambda, \CC[M]_k).
\end{equation*}
We shall encode their multiplicities, suitably re-scaled, in the following sequence of discrete measures,
\begin{equation}
  \label{discrete irrep measure}
  \mu^K_{M,k} := \frac 1 {k^{n - R}} \sum_{\lambda \in \Lambda^* \cap \mathfrak t^*_+} \dim \Hom_K(V^*_\lambda, \CC[M]_k) \, \delta_{\lambda/k}.
\end{equation}
The factor $k^R$ accommodates for the growth of the dimension of a generic irreducible representation in the coordinate ring, which has highest weight in $\mathfrak t^*_{>0}$ (we still assume that \autoref{main assumption} is in place).

It is well-known that in the \emph{semi-classical limit} $k \rightarrow \infty$ this sequence of measures converges in distribution to the non-Abelian Duister\-maat--Heck\-man measure \cite{heckman82,guilleminsternberg82b,sjamaar95,meinrenken96,meinrenkensjamaar99,vergne98},
\begin{equation}
   \label{semiclassicallimit}
  \mu^K_{M,k} \rightarrow \DuHe^K_M.
\end{equation}
In fact, one can show using the Hirzebruch--Riemann--Roch theorem that the piecewise polynomial density function of the Duister\-maat--Heck\-man measure is at any rational regular point $\lambda \in \mathfrak t^*_{>0}$ for $\Phi_K$ (equivalently, for $\Phi_T$) given by
\begin{equation}
  \label{semiclassicallimit pointwise}
  \lim_{k \rightarrow \infty}
  \frac 1 D
  \frac 1 {(kq)^{n-R-r}}
  \dim \Hom_K(V_{kq \lambda}^*, \CC[M]_{kq}).
\end{equation}
Here, $q > 0$ is an integer and $D$ is the cardinality of the common $T$-stabilizer of the points in $\Phi_K^{-1}(\mathfrak t^*_{>0})$.%
\footnote{Quotienting out a discrete subgroup leaves the Duister\-maat--Heck\-man measure invariant, but changes the weight lattice, and therefore the normalization of the Lebesgue measure $d\lambda$.}
The additional factor $(kq)^r$ comes from the fact that we consider the density with respect to the $r$-dimensional Lebesgue measure $d\lambda$.
If $T$ acts freely on $\Phi_K^{-1}(\lambda)$ then we can choose $q=1$ \cite{guilleminsternberg82b}.

It is well-known that the multiplicity of $V_{k\lambda}^*$ in $\CC[M]_k$ is a \emph{quasi-polynomial} in $k$, i.e., a polynomial whose coefficients are periodic functions of $k$. Therefore, the existence of the limit \eqref{semiclassicallimit pointwise} implies that (a) the degree of this quasi-polynomial is at most $n-R-r$, and (b) if it is of maximal degree then the leading order coefficient (for $k \equiv 0$ modulo the period) is equal to the value of the density function, i.e., of the limit \eqref{semiclassicallimit pointwise}. We shall therefore call this value the \emph{maximal-order growth coefficient} of the quasi-polynomial.

The rational points of the moment polytope $\Delta_K(M)$ are precisely those of the form $\lambda / k$ with $V^*_\lambda \subseteq \CC[M]_k$ \cite{brion87}. In other words, $\supp \DuHe^K_M \cap \mathfrak t^*_{\QQ} = \bigcup_k \supp \mu^K_{M,k}$.

By restricting to the action to the maximal torus $T \subseteq K$, we observe that the Abelian Duister\-maat--Heck\-man measure $\DuHe^T_M$ captures the asymptotic distribution of weights in the homogeneous coordinate ring of $M$, i.e., the asymptotics of the character of $\CC[M]_k$ as $k \rightarrow \infty$.

\begin{exl}
  For strictly dominant and integral $\lambda \in \Lambda^* \cap \mathfrak t^*_{>0}$, the Borel--Weil theorem shows that the homogeneous coordinate ring of the coadjoint orbit $\mathcal O_\lambda \subseteq P(V_\lambda)$ is equal to
  \begin{equation*}
    \CC[\mathcal O_\lambda] = \bigoplus_{k=0}^\infty V^*_{k \lambda}.
  \end{equation*}
  Therefore, all the multiplicity measures $\mu_{\mathcal O_{\lambda},k}^K$ (and hence their limit) are equal to the Dirac measure at $\lambda$.

\end{exl}

\subsection{Multiplicities of Irreducible Representations via Finite Differences}
\label{multiplicities of irreducibles}

Multiplicities of weights and highest weights in finite-dimensional $K$-representations are related by iteratively taking (negative) finite differences in the directions of the positive roots. This can be seen as the ``quantized'' version of the derivative principle (\autoref{main theorem}). Its proof is in essence a rephrasing of the Weyl character formula, an idea which goes back at least to Steinberg \cite{steinberg61}.

\begin{lem}
  \label{steinberg lemma}
  Denote by $m_K$ and $m_T$ the highest weight and weight multiplicity function, respectively, of a finite-dimensional $K$-representation $V$.
  Then on the positive Weyl chamber we have
  \begin{equation*}
    m_K = \restrict{\left(\prod_{\alpha > 0} - D_\alpha \right) m_T}{\mathfrak t^*_+},
  \end{equation*}
  where $(D_\alpha m)(\lambda) = m(\lambda + \alpha) - m(\lambda)$ is the finite-difference operator in direction $\alpha$.
  Note that any two of the operators $D_\alpha$ commute, so that their product is independent of the order of multiplication.
\end{lem}
\begin{proof}
  By linearity of the finite-difference operators it suffices to establish the lemma for a single irreducible representation $V_\lambda$ of highest weight $\lambda$. It will be convenient to work with the formal character $\mathrm{ch}(V_\lambda) = \sum_\mu m_T(\mu) \, e^\mu$ \cite{cartersegalmacdonald95,knapp02}.
  By the Weyl character formula,
  \begin{equation*}
    \prod_{\alpha > 0} \left( 1 - e^{-\alpha} \right) \mathrm{ch}(V_\lambda) =
    e^{-\rho} \sum_{w \in W} (-1)^{l(w)} e^{w(\lambda + \rho)}
  \end{equation*}
  where $W$ is the Weyl group, $l(w)$ the length of a Weyl group element $w$, and $\rho$ half the sum of the positive roots.

  Now observe that the left-hand side is the generating function of $\left( \prod_{\alpha > 0} - D_\alpha \right) m_T$, since taking finite differences corresponds to multiplying the generating function by $1 - e^{-\alpha}$. Up to terms corresponding to non-dominant weights, the right-hand side is equal to $e^\lambda$, which is the generating function of $m_K = \delta_{\lambda,-}$. The assertion follows from this.
\end{proof}

Note that the Weyl character formula can be seen as the representation-theoretic analogue of the Harish-Chandra formula that was used to establish \autoref{main theorem}.
In the semi-classical limit \eqref{semiclassicallimit}, the finite differences become infinitesimal and we recover an alternative proof of \autoref{main theorem} in the algebro-geometric setting.

This argument can also be turned around to establish \eqref{semiclassicallimit} for general compact Lie groups $K$ from its Abelian version \cite[(34.8)]{guilleminsternberg84} and \autoref{main theorem}.\footnote{We thank Allen Knutson for pointing this out, as well as for sketching a self-contained proof of \eqref{semiclassicallimit}.}

\subsection{Multiplicities for Projective Spaces}
\label{multiplicities for projective space}

As in \autoref{projective space}, let $M = \PP(V)$ be the complex projective space for a unitary $K$-representation $V$. Its homogeneous coordinate ring is equal to the symmetric algebra,
\begin{equation*}
  \CC[\PP(V)] = \Sym(V^*) = \bigoplus_{k=0}^\infty \Sym^k(V^*).
\end{equation*}

Choose a weight-space decomposition $V = \bigoplus_{k=0}^n \CC v_k$, with weights $\omega_k$, and identify $V \cong \CC^{n+1}$ and $\U(V) \cong \U(n+1)$ accordingly. Observe that the maximal torus $T \subseteq K$ acts via the standard maximal torus of $\U(n+1)$, that is, the set of unitary diagonal matrices, which we denote by $\tilde T$.

Each symmetric tensor power $\Sym^k(\CC^{n+1})$ is an irreducible representation of $\U(n+1)$. Its weight spaces are all one-dimensional, and the weights that occur are precisely the $\lambda = \diag(\lambda_0, \ldots, \lambda_n)$ with $i \lambda_j \in \ZZ$, $\lambda_j \geq 0$, and $\sum_j \lambda_j = k$ \cite{fulton97}. Clearly, we can identify this set of weights with the integral points in $k \Delta_n$, where $\Delta_n$ is the $n$-dimensional standard simplex in $\RR^{n+1}$. In the language of Young diagrams, these are the weight vectors corresponding to semistandard Young tableaux of shape $(k)$ with entries in $\{0,\ldots,n\}$.

To determine the weight multiplicities with respect to $T \subseteq K$, we have to ``restrict'' each weight to $\mathfrak t$. This corresponds precisely to applying the map $P \colon \RR^{n+1} \rightarrow \mathfrak t^*, (t_k) \mapsto \sum_k t_k \omega_k$ introduced in \autoref{projective space abelian via standard simplex}. Therefore, the multiplicity in $\Sym^k(V^*)$ of the dual of a weight $\lambda \in \Lambda^*$ is given by counting integral points in a rational convex polytope parametrized by $k$ and $\lambda$:
\begin{equation}
  \label{weight multiplicities for projective space}
  m_{T,k}(\lambda) = \#\left(\Delta(\lambda, k) \cap \ZZ^{n+1}\right),
\end{equation}
where
\begin{equation}
  \Delta(\lambda, k) = \Big\{ (t_j) \in \RR^{n+1} : t_j \geq 0, \sum_{j=0}^n t_j \omega_j = \lambda, \sum_{j=0}^n t_j = k \Big\}.
\end{equation}

\begin{rem}
\label{future remark}
Such \emph{vector partition functions} can be evaluated efficiently using Barvinok's algorithm if the group $K$ and the ambient dimension $\dim V = n + 1$ is fixed \cite{barvinok93,barvinok94,barvinokpommersheim99}, namely in time $O(\mathrm{poly}(\log k))$. In fact, $m_{T,k}$ is a piecewise quasi-polynomial function in both $\lambda$ and $k$, and there are parametric generalizations of Barvinok's algorithm for computing these quasi-polynomials \cite{verdoolaegebruynooghe08,verdoolaegeseghirbeylsetal07}.
Since we can compute multiplicities of irreducible $K$-representations by taking finite differences of weight multiplicities in the direction of positive roots (\autoref{steinberg lemma}), this can also be done efficiently if $K$ is fixed.
We will report on a generalization of this technique to the general branching problem for compact connected Lie groups in a forthcoming article \cite{christandldoranwalter12}.
\end{rem}

There is also a jump formula by Boysal and Vergne \cite{boysalvergne09}, which as in \autoref{projective space} can be used to inductively compute the quasi-polynomials chamber by chamber.

\medskip

We now turn to the semi-classical limit. As $k \rightarrow \infty$, it is clear that
\begin{equation*}
  \mu^{\tilde T}_{\PP(V), k} = \frac 1 {k^n} \sum_{\lambda \in \Delta_n \cap \frac 1 k \ZZ^{n+1}} \delta_\lambda
\end{equation*}
converges to Lebesgue measure on the standard simplex $\Delta_n$, normalized to total volume
\begin{equation*}
  \lim_{k \rightarrow \infty} \frac 1 {k^n} \dim \Sym^k(V^*) = \lim_{k \rightarrow \infty} \frac 1 {k^n} {\binom{n+k}{n}} = \frac 1 {n!}.
\end{equation*}
Therefore, $\mu^T_{\PP(V),k}$ converges to the push-forward of Lebesgue measure on $\Delta_n$ along the map $P$. By the semi-classical limit \eqref{semiclassicallimit}, this is of course equivalent to the assertion of \autoref{projective space abelian via standard simplex}. Moreover, note that the quantity
\begin{equation*}
  m_{T,k}(k \lambda) =
  \#\left(\Delta(k \lambda, k) \cap \ZZ^{n+1}\right) =
  \#\left(\Delta(\lambda, 1) \cap \tfrac 1 k \ZZ^{n+1}\right)
\end{equation*}
is the Ehrhart quasi-polynomial associated to rational polytope $\Delta(\lambda, 1)$ \cite{beckrobins09}. It is intuitively clear that its growth in $k$ should be related to the volume of this polytope. Indeed,
\begin{equation*}
  m_{T,k}(k \lambda) = k^{n-r} \vol \Delta(\lambda, 1) + O(k^{n-r-1}),
\end{equation*}
where $\vol$ is the $(n-r)$-dimensional volume with respect to the measure $dt/d\hat\lambda$ defined in \autoref{projective space} \cite[Exercise 3.29]{beckrobins09}.
Observe that this agrees with \eqref{semiclassicallimit pointwise} and \autoref{density in polytope picture}: The maximal-order growth coefficient is a constant equal to the Abelian Duister\-maat--Heck\-man density at point $\lambda$.

\section{Kronecker and Plethysm Coefficients}
\label{kronecker section}

In this section, we describe the representation theory of the quantum marginal problem in more detail. For distinguishable particles, the relevant multiplicities can be expressed in terms of decomposing tensor products of irreducible representations of the symmetric group (\autoref{kronecker coefficients}). In particular, the joint eigenvalue distribution of the reduced density matrices of a tripartite pure state is determined by the asymptotics of the Kronecker coefficients (see \eqref{kroneckerlimit}). We emphasize that by specializing the method described in \autoref{multiplicities for projective space} we get a novel algorithm for computing Kronecker coefficients which is efficient for Young diagrams of bounded height.
Indistinguishable particles correspond to certain plethysm coefficients and we conclude by illustrating this connection (\autoref{plethysms}).

\subsection{Kronecker Coefficients}
\label{kronecker coefficients}

Recall that for $N$ distinguishable particles we have to consider the action of $K = \SU(d_1) \times \ldots \times \SU(d_N)$ on a coadjoint $\SU(d_1 \cdots d_N)$-orbit $M = \mathcal O_{\tilde\lambda}$, where we now assume that $\tilde\lambda$ is an integral weight in $\mathfrak t^*_+$. The multiplicity measures $\mu^K_{M,k}$ are determined by the decomposition of the homogeneous coordinate ring
\begin{equation*}
  \CC[M] = \bigoplus_{k=0}^\infty \left( V^{d_1 \cdots d_N}_{k \tilde\lambda} \right)^*
\end{equation*}
into $K$-isotypical components (the superscript labels the corresponding $\SU$).

We can express this equivalently using the representation theory of the symmetric group $S_m$. Recall that by Schur--Weyl duality the diagonal action of $\SU(d)$ and the permutation action of $S_m$ on $(\CC^d)^{\otimes m}$ generate each other's commutant, so that
\begin{equation}
  \label{schur weyl decomposition}
  (\CC^d)^{\otimes m} \cong \bigoplus_{\mu} V^d_\mu \otimes [\mu].
\end{equation}
Here, the sum runs over all Young diagrams $\mu=(\mu_1,\ldots,\mu_d)$ with $\abs \mu := \sum_j \mu_j = m$ boxes and at most $d$ rows, $V^d_\mu$ is the irreducible representation of $\SU(d)$ with highest weight $X \mapsto i \sum_j X_j \mu_j$, and $[\mu]$ is the corresponding irreducible representation of $S_m$ (see \cite{fulton97} for details). We shall freely identify Young diagrams and the corresponding highest weights.

In particular, we can realize the irreducible representation $V^{d_1 \cdots d_N}_{k \tilde\lambda}$ in $(\CC^{d_1} \otimes \ldots \otimes \CC^{d_N})^{\otimes \abs{k \tilde\lambda}}$. Comparing the Schur--Weyl decomposition \eqref{schur weyl decomposition} for the full Hilbert space with the tensor product of the decompositions for the individual subsystems, we find that
\begin{equation*}
  V^{d_1 \cdots d_N}_{k \tilde\lambda}
  \cong
  \bigoplus_{\lambda_1,\ldots,\lambda_N}
  V^{d_1}_{\lambda_1} \otimes \ldots \otimes V^{d_N}_{\lambda_N} \otimes
  \Hom_{S_{\abs{k\tilde\lambda}}}([k \tilde\lambda], [\lambda_1] \otimes \ldots \otimes [\lambda_N]),
\end{equation*}
where the sum runs over the Young diagrams $\lambda_i$ with $\abs{k\tilde\lambda}$ boxes and at most $d_i$ rows. Therefore,
\begin{equation}
\begin{aligned}
  \label{generalized kroneckers}
  \mu^K_{M,k}
  = &\frac 1 {k^{n-R}}
  \sum_{\lambda_1,\ldots,\lambda_N} \dim \Hom_{S_{\abs{k\tilde\lambda}}}([k \tilde\lambda], [\lambda_1] \otimes \ldots \otimes [\lambda_N]) ~ \delta_{(\lambda_1/k,\ldots,\lambda_N/k)}\\
  = &\frac 1 {k^{n-R}}
  \sum_{\lambda_1,\ldots,\lambda_N} \dim \Hom_{S_{\abs{k\tilde\lambda}}}([k \tilde\lambda]^*, [\lambda_1] \otimes \ldots \otimes [\lambda_N]) ~ \delta_{(\lambda_1/k,\ldots,\lambda_N/k)},
\end{aligned}
\end{equation}
where the latter identity holds due to the self-duality of the representations of the symmetric group.
In particular, the rational points of the non-Abelian moment polytope $\Delta_K(\mathcal O_{\tilde\lambda})$ are precisely
\begin{align}
  \label{rational points QMP moment polytope}
  \bigcup_k \left\{ (\lambda_1/k,\ldots,\lambda_N/k) : \begin{array}{l}
    [k\tilde\lambda] \subseteq [\lambda_1] \otimes \ldots \otimes [\lambda_N], \text{ where the}\\
    \text{$\lambda_i$ have $\abs{k\tilde\lambda}$ boxes and at most $d_i$ rows}
  \end{array} \right\}.
\end{align}
See \autoref{hedgehog figure} for an illustration of the multiplicity measures corresponding to the mixed-state quantum marginal problem for two qubits discussed in \autoref{bravyi example}.

\begin{figure}
  \centering
  \includegraphics[width=3.5cm]{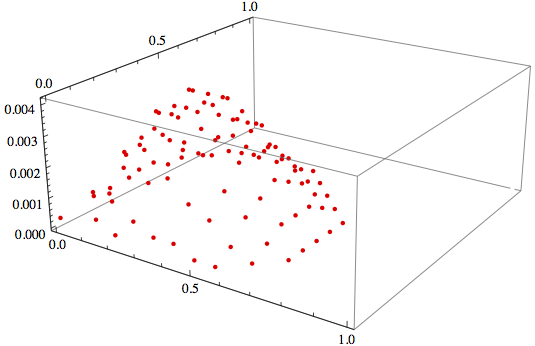}
  \includegraphics[width=3.5cm]{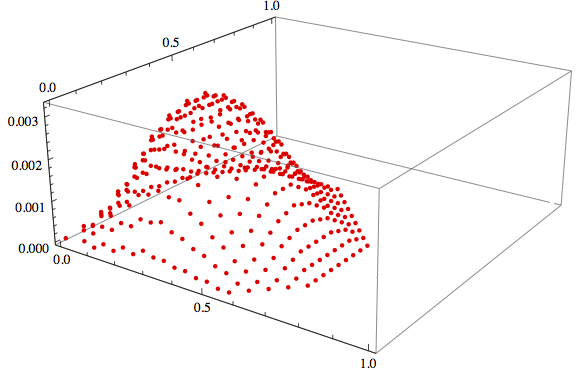}
  \includegraphics[width=3.8cm,trim=0cm 0cm 0cm 0cm]{bravyi_polytope_4210.png}
  \caption{(a) and (b) Illustration of the multiplicity measures $\mu^K_{M,k}$ for the mixed-state quantum marginal problem of two qubits with global spectrum $(4/7,2/7,1/7,0)$ and $k=28, 56$, which have been computed by the algorithm described in \autoref{kronecker coefficients}. (c) Their semi-classical limit, i.e., the corresponding Duister\-maat--Heck\-man measure as computed in \autoref{bravyi example}.}
  \label{hedgehog figure}
\end{figure}


\begin{rem}
We can write the multiplicities in \eqref{generalized kroneckers} in the following symmetric form:
\begin{equation*}
  \dim \Hom_{S_{\abs{k\tilde\lambda}}}([k \tilde\lambda]^*, [\lambda_1] \otimes \ldots \otimes [\lambda_N]) =
  \dim \left(
    [\lambda_1] \otimes \ldots \otimes [\lambda_N] \otimes [k \tilde\lambda]
  \right)^{S_{\abs{k\tilde\lambda}}}.
\end{equation*}
Observe that the right-hand side is a multiplicity for the pure-state quantum marginal problem for $\CC^{d_1} \otimes \ldots \otimes \CC^{d_N} \otimes \CC^{d_1 \cdots d_N}$. Indeed, the homogeneous coordinate ring of a projective space is just the symmetric algebra (\autoref{multiplicities for projective space}), whose graded parts correspond to the trivial representations of the respective symmetric groups.
This is the representation-theoretic perspective on purification (cf.\ \autoref{purification}, in particular \autoref{purification measure}).
\end{rem}

For the tripartite pure-state quantum marginal problem (equivalently, the mixed-state bipartite quantum marginal problem), the relevant multiplicities are the well-known \emph{Kronecker coefficients} of the symmetric group,
\begin{equation*}
  g_{\lambda,\mu,\nu} = \dim \left( [\lambda] \otimes [\mu] \otimes [\nu] \right)^{S_k}.
\end{equation*}
They are the symmetric group analogue of the Littlewood--Richardson coefficients of the unitary group (in fact, the latter can be considered as a special case) but much harder to compute in general, since there is no combinatorial description like the Littlewood--Richardson rule.

The corresponding characterization \eqref{rational points QMP moment polytope} of the non-Abelian moment polytope has already been observed in \cite{christandlmitchison06,klyachko04,christandlharrowmitchison07}, as well as in \cite{daftuarhayden04} for the projection onto two of the subsystems. The semi-classical limit refines this characterization: Not only can one read off the existence of quantum states with given marginal eigenvalue spectra from the asymptotic non-vanishing of the corresponding Kronecker coefficients $g_{k\lambda,k\mu,k\nu}$, but their growth also encodes the probability of finding these eigenvalue spectra when the global state is chosen according to the invariant probability measure.
Explicitly, \eqref{semiclassicallimit} states that
\begin{equation}
  \label{kroneckerlimit}
  \frac 1 {k^p} \sum_{\lambda,\mu,\nu} g_{\lambda,\mu,\nu} \, \delta_{\lambda/k, \mu/k, \nu/k}
  \rightarrow \DuHe^{\PP(\CC^a \otimes \CC^b \otimes \CC^c)}_{\SU(a) \times \SU(b) \times \SU(c)},
\end{equation}
where $p = n - R = {abc - 1 - a(a-1)/2 + b(b-1)/2 + c(c-1)/2}$, and where the sum runs over all Young diagrams $\lambda, \mu, \nu$ with $k$ boxes and at most $a$, $b$ and $c$ rows, respectively.

\medskip

The method described in \autoref{multiplicities for projective space} in particular provides a novel algorithm for computing the Kronecker coefficients which is efficient for Young diagrams of bounded height: Using the finite-difference formula of \autoref{steinberg lemma}, we can reduce to the computation of a bounded number of weight multiplicities \eqref{weight multiplicities for projective space}, which using Barvinok's algorithm can be evaluated in polynomial time in the input size, i.e., in time $O(\mathrm{poly}(\log k))$, where $k$ is the number of boxes of the Young diagrams. As mentioned in \autoref{future remark}, we will elaborate on this algorithm in a forthcoming article \cite{christandldoranwalter12}.

\subsection{Plethysm Coefficients}
\label{plethysms}

While the quantum marginal problem for distinguishable particles is connected to (generalized) Kronecker coefficients, it is for indistinguishable particles related to certain plethysm coefficients. Indeed, if $M = \PP(V_\lambda)$ for an irreducible $\SU(d)$-representation $V_\lambda$ then its coordinate ring consists of the \emph{plethysms}
\begin{equation*}
  \CC[M] = \bigoplus_{k=0}^\infty \Sym^k(V_\lambda^*).
\end{equation*}
See e.g.\ \cite{macdonald95} for more information on plethysms, which are in general defined as the composition of Schur functors.
In particular, the bosonic and fermionic pure-state marginal problem are related to the asymptotics of $\Sym^k(\Sym^N(\CC^d)^*)$ and $\Sym^k(\Alt^N(\CC^d)^*)$, respectively, as $k \rightarrow \infty$.

\medskip

Let us illustrate this by giving an alternative derivation of the Duister\-maat--Heck\-man measures for $N$ bosonic qubits (cf.\ \autoref{pure states bosonic qubits}). We will explicitly compute the asymptotic weight multiplicity distribution of the plethysm $\Sym^k(\Sym^N(\CC^2)^*) \cong \Sym^k(\Sym^N(\CC^2))$ as $k \rightarrow \infty$, and then apply the derivative principle.
The main combinatorial tool we shall employ are the $q$-binomial coefficients
\begin{equation*}
  \qbinom{n}{k}_q = \frac{[n]_q !}{[k]_q ! [n-k]_q !}
  .
\end{equation*}
Recall that these are defined in terms of the $q$-integers $[n]_q = \frac{1-q^n}{1-q}$ 
and $q$-factorials $[n]_q ! = [n]_q [n-1]_q\ldots [1]_q$.
We start with the following description of the character of $\Sym^k(\Sym^N(\CC^2))$ in terms of $q$-binomial coefficients:

\begin{prp}[\protect{\cite[pp. 53]{springer77}}]
\label{springer character qbinomial}
  Let $k,N \in \NN$ and $q=e^{\omega_1}$. Then,
  \begin{equation*}
    \mathrm{ch} \left( \Sym^k \left( \Sym^N \left( \CC^2 \right) \right) \right) = \qbinom{k+N}{N}_{q^2} q^{-kN}.
  \end{equation*}
\end{prp}

\begin{prp}
  \label{asy eq q-stuff}
  As functions on the open unit disk $\{ q \in \CC : |q|<1 \}$ one has for fixed $N \in \NN$ and $k \rightarrow \infty$ the following asymptotic equivalence
  \[
    [N]_q ! \qbinom{k+N}{N}_q \sim [k]_q^N
    .
  \]
\end{prp}
\begin{proof}
  Following \cite[(9.1)]{kaccheung02}, for any fixed $c \in \NN$ and $|q|<1$ one has
  \[
    \lim_{k \rightarrow \infty} \frac{1-q^{k+c}}{1-q} = \frac{1}{1-q}.
  \]
  By applying this identity both to the numerator and the denominator,
  \[
    \lim_{k \rightarrow \infty} \frac{[k+c]_q}{[k]_q} = 1
    .
  \]
  Hence
  \[
    \lim_{k \rightarrow \infty} \frac{ [N]_q !\qbinom{k+N}{N}_q}{[k]_q^N} =
    \lim_{k \rightarrow \infty} \frac{[k + N]_q [k+N-1]_q \ldots [k+1]_q}{[k]_q^N} =
    1
    .
    \qedhere
  \]
\end{proof}

The following corollary is an easy application of Osgood's theorem \cite{osgood01} to \autoref{asy eq q-stuff} (see e.g.~\cite{beardonminda03}).

\begin{cor}
  \label{osgood}
  Fix $N$, and define $f_k(q) = [N]_q ! \qbinom{k+N}{N}_q$ and $g_k(q) = [k]_q^N$. Then, the sequences $f_k$ and $g_k$ converge, as $k \rightarrow \infty$, pointwise to the same holomorphic function on some open dense subset in the open unit disk $\{ q \in \CC : |q|<1 \}$. In particular, the limits are equal as power series.
\end{cor}

We can use this result to extract asymptotic multiplicity information.

\begin{prp}
  \label{asymptotic uniformness}
  The discrete measures $\mu_{(N),k}^{\U(1)} := \mu_{\PP(\Sym^N(\CC^2)),k}^{\U(1)}$ as defined in \eqref{discrete irrep measure} tend in the limit $k \rightarrow \infty$ to $\frac 1 {N!}$ times the probability distribution of the sum of $N$ independent random variables uniformly distributed on $[-1,1]$.
\end{prp}
\begin{proof}
  By \autoref{springer character qbinomial}, $\mu_{(N),k}^{\U(1)}$ is a finite measure with generating function
  \begin{equation*}
    \int q^x \, d\mu_{(N),k}^{\U(1)}(x \, \omega_1) =  \frac 1 {k^N} \qbinom{k+N}{N}_{q^{2/k}} q^{-N}.
  \end{equation*}
  Let $\nu_k$, $\omega_k$ be finite measures with generating functions $[N]_{q^{2/k}}!$ and
  \begin{equation*}
    \frac 1 {k^N} [k]^N_{q^{2/k}} q^{-N} = \left( \frac{q^{-1} + q^{-1+2/k} + \ldots + q^{1-2/k}}{k} \right)^N,
  \end{equation*}
  respectively. Obviously, $\omega_k$ is asymptotically distributed like the sum of $N$ independent random variables uniformly distributed on the interval $[-1,1]$, and by \autoref{osgood} so is $\mu^{\U(1)}_{(N),k} \star \nu_k$.
  Since $\nu_k \rightarrow N! \, \delta_0$ as $k \rightarrow \infty$, this implies our assertion.
\end{proof}

By using the semi-classical limit \eqref{semiclassicallimit}, we conclude once again that the Abelian Duister\-maat--Heck\-man measure is given by the formula that was established in \autoref{sym N qubits abelian}. The non-Abelian Duister\-maat--Heck\-man measure is obtained as in \autoref{sym N qubits non-abelian} by applying the derivative principle.

\begin{exl}[$N=2$]
  For the plethysms $\Sym^k(\Sym^2(\CC^2))$ we can also illustrate the semi-classical limit for the $\SU(2)$-action, since the decomposition into irreducible $\SU(2)$-representations is explicitly known \cite[\S 1.5, Example 6 (a)]{macdonald95}:
  \begin{equation*}
    \Sym^k(\Sym^2 (\CC^2)) \cong
    \Sym^{2k}(\CC^2) \oplus
    \Sym^{2k-4}(\CC^2) \oplus
    \ldots \oplus
    \Sym^{2 / 0}(\CC^2)
  \end{equation*}
  The last summand is $\Sym^2 (\CC^2)$ for odd $k$, and $\Sym^0 (\CC^2)$ for even $k$.
  Therefore, the discrete measures as defined in \eqref{discrete irrep measure} are given by
  \begin{equation*}
    \mu^{\SU(2)}_{\PP(\Sym^2(\CC^2)),k} = \sum_{l=2k, 2k-4, \ldots, 2 \vert 0} \frac 1 {k^{2-1}} \delta_{\frac l k}.
  \end{equation*}
  In the limit $k \rightarrow \infty$, they converge to the non-Abelian Duister\-maat--Heck\-man measure as computed in \autoref{sym N qubits non-abelian},
  \begin{equation*}
    \DuHe^{\SU(2)}_{\PP(\Sym^2(\CC^2))} = - \tfrac 1 4 \left( (x+2)^0_+ - 2 x^0_+ + (x-2)^0_+ \right) dx =
    \tfrac 1 4 \Id_{[0,2)}(x) dx.
  \end{equation*}
  See \autoref{sym two qubits figure} for an illustration.
\end{exl}

\begin{rem}
	The description of the character of $\Sym^k(\Sym^N(\CC^2))$ via $q$-binomial coefficients has the additional advantage that one is able to compute all higher cumulants and moments of the associated distribution for any fixed $k$ and $N$ (see \cite{MR2019639}). This is due to a method by Panny \cite{MR845446}.
\end{rem}

\begin{acknowledgement}
We would like to thank Alonso Botero, Emmanuel Briand, Peter B\"urgisser, Beno\^{i}t Collins, David Gross, Christian Ikenmeyer, Eckhard Meinrenken, Markus P.\ M\"uller, Mercedes Rosas, and Volkher Scholz for helpful discussions. We thank Graeme Mitchison for joint initial discussions on the topic of asymptotics of Kronecker coefficients. The second author would like to express his particular gratitude to Frances Kirwan for many fruitful discussions regarding moment maps and invariant theory and their many uses.
We acknowledge financial support by the German Science Foundation (grant CH 843/2-1), the Swiss National Science Foundation (grants PP00P2\_128455, 20CH21\_138799 (CHIST-ERA project CQC), 200021\_138071), the Swiss National Center of Competence in Research 'Quantum Science and Technology (QSIT)', the Swiss State Secretariat for Education and Research supporting COST action MP1006 and the European Research Council under the European Union's Seventh Framework Programme (FP/2007-2013) / ERC Grant Agreement no. 337603. We also acknowledge support by the Excellence Initiative of the German Federal and State Governments through the Junior Research Group Program within the Institutional Strategy ZUK 43.
\end{acknowledgement}

\bibliographystyle{unsrt}
\bibliography{dhmeasure}

\end{document}